\documentclass[11pt]{article}
\usepackage{epsf}
\usepackage{epsfig}
\usepackage{graphicx}
\usepackage{color}
\usepackage{amsmath}
\usepackage{mathrsfs}

\usepackage{amsthm}
\usepackage{latexsym}
\usepackage{amssymb}
\usepackage{amsmath}

\usepackage{stmaryrd}

\usepackage{epsf}
\usepackage{verbatim}

\newtheorem{theorem}{Theorem}[section]

\newtheorem{lemma}{Lemma}[section]
\newtheorem{corollary}{Corollary}[section]

\newtheorem{definition}{Definition}[section]

\begin{document}
\title{{\bf Dichotomy for Real Holant$^c$ Problems}}

\vspace{0.3in}
\author{Jin-Yi Cai\thanks{University of Wisconsin-Madison.
 {\tt jyc@cs.wisc.edu}. }
\and Pinyan Lu\thanks{ITCS, Shanghai University of Finance and Economics {\tt
lu.pinyan@mail.shufe.edu.cn}} \and Mingji Xia\thanks{State Key Laboratory of Computer Science, Institute of Software, Chinese Academy of Sciences. University of Chinese Academy of Sciences. {\tt mingji@ios.ac.cn}}}

\date{}
\maketitle

\bibliographystyle{plain}

\begin{abstract}
Holant problems capture a class of Sum-of-Product
computations such as counting matchings.
It is inspired by holographic algorithms and
is equivalent to tensor networks, with counting CSP being a special case.
A classification for Holant problems is more difficult
to prove, not only because it implies a
classification for counting CSP, but also due to the
deeper reason that there exist 
more intricate polynomial time tractable problems in the broader framework.

We discover a new family of constraint functions $\mathscr{L}$ which
define polynomial time computable counting problems. These
do not appear in counting CSP, and no newly discovered
tractable constraints can be symmetric.
It has a delicate support structure related to error-correcting codes.
Local holographic transformations is fundamental in its tractability.
We prove a complexity dichotomy theorem for all
Holant problems
defined by any real valued constraint function set on Boolean
variables and contains
two 0-1 pinning functions. Previously, dichotomy for
 the same framework was only  known for \emph{symmetric} constraint functions.
 The set $\mathscr{L}$ 
supplies the last piece of tractability.
 We also prove a dichotomy for a variant of
counting CSP as a technical component toward this Holant dichotomy.
\end{abstract}

\thispagestyle{empty}
\newpage
\setcounter{page}{1}

\section{Introduction}
There has been great progress in the complexity classification program
for counting problems defined as Sum-of-Product computations.
An ideal outcome of such a result is usually stated in the form of
a dichotomy theorem, namely it classifies every single problem expressible
in the class as either \#P-hard or polynomial time solvable. Counting
Constraint Satisfaction Problems (\#CSP) is the most well-studied
framework in such context. For \#CSP over the Boolean domain,
 two explicit tractable families, namely $\mathscr{P}$ (product type)
and  $\mathscr{A}$ (affine type), are identified; any function set
not contained in these two families is proved to be \#P-hard.
The result was first proved for unweighted 0-1 valued constraint
 functions~\cite{CreignouH96},
 later for non-negatively weighted functions~\cite{weightedCSP},
 and finally for complex valued
 functions~\cite{CaiLX14}.  From  non-negative values to complex values,
 the tractable family $\mathscr{A}$ expands highly non-trivially;
the tractability incorporates cancelations and the proof depends on a
nice algebraic structure.
 Dichotomy theorems are also known for \#CSP over large domains
although the tractability criterion is not very explicit
and it is  not even known to be decidable in the case of complex weighted
constraint functions~\cite{Bulatov08,dyerstoc10,dyer11,CaiCL11,CaiC12}.
%These proof techniques are also very different.
In this paper, we focus on problems over the Boolean domain.

Unfortunately, not every problem defined by local constraints can be
described in the \#CSP framework, and thus  not every such problem is
 covered by the
 \#CSP dichotomies. E.g., the graph matching problem
is such an example~\cite{freedman-l-s}.
 %which plays an important role in the study of complexity and algorithms,
%is not covered in the CSP framework.
However,
it is naturally included in a more refined framework, called Holant problems.
This was defined in~\cite{FOCS08},
and the name was inspired by the introduction of
{\it Holographic Algorithms} by L.~Valiant~\cite{bib:holographicalg,AA_FOCS}
(who first used the term Holant).
%Holant framework can describe graph matching problems naturally.
The Holant framework is essentially equivalent to tensor networks.
\#CSP can be viewed as a special case of Holant problems.
% where {\sc Equalities}
% of all arities are assumed to be among
%the constraints.
Compared to \#CSP, the Holant framework contains
more surprising tractable problems.
Consequently, it is also much more challenging to prove dichotomy
theorems in the Holant framework.
 After a great deal of work~\cite{holant-stoc,bib:holantc,parity,realholant,complexholant},  a dichotomy for Holant
problems was proved  for \emph{symmetric} constraint functions.
But obviously symmetric functions are only a tiny fraction of all
constraint functions.

Let us meet a function $f$ on 14 variables.
We will show that this $f$ is a new breed of functions which
define  tractable problems
in the Holant framework. It is not symmetric, and such tractable functions
do not show up in the \#CSP framework.

Let $H = \begin{bmatrix}
0 & 0 & 0 & 1 & 1 & 1 & 1 \\
0 & 1 & 1 & 0 & 0 & 1 & 1 \\
1 & 0 & 1 & 0 & 1 & 0 & 1
\end{bmatrix}$.
A standard definition of the $[7,4]$-Hamming code $C$ consists of
0-1 strings of length 7 with  $H$ as its parity check matrix:
$C = \{ {\bf x} \in \mathbb{Z}_2^7 \mid {\bf x} H = {\bf 0} \bmod 2\}$.
We consider the dual Hamming code $C^{\perp}$ which has $H$ as a generating
matrix. $C^{\perp}$ is a linear subspace of $\mathbb{Z}_2^7$ of dimension 3.
It is well-known that every nonzero word of  $C^{\perp}$ has
Hamming weight 4.
Let
\[S = \{{\bf w}\overline{{\bf w}} \in \mathbb{Z}_2^{14}
 \mid {\bf w} \in C^{\perp}\},\]
where $\overline{{\bf w}}$ flips every bit of ${\bf w}$.
Clearly $S$ is an affine linear subspace in $\mathbb{Z}_2^{14}$ of dimension 3.

Now our  function  $f$ is defined as follows:
$f$ has support $S$.
In the column order of $H$ we may take free variables $x_1, x_2, x_4$,
and on the support $S$ we have $x_3 = x_1 + x_2$, $x_5 = x_1 + x_4$,
$x_6 = x_2 + x_4$, and $x_7 = x_1 + x_2 + x_4$ (arithmetic in $\mathbb{Z}_2$).
There are 7 other variables $x_{7+i}$ ($ 1 \le i \le 7$), and
on $S$ we have $x_{7+i} = \overline{x_i}$.
In terms of the 0-1 valued free variables  $x_1, x_2, x_4$, $f$ takes
value $(-1)^{x_1x_2x_4}$ on $S$, and $0$ elsewhere.
Thus on the support set $S$, $f = 1$, except at one point
$x_1=x_2=x_4=1$ it takes value $-1$.

It turns out that this $f$ defines a tractable Holant problem,
even though it does not belong to any of the previously  known
tractable constraint function families for \#CSP.
The tractability of $f$ depends on the fact that
every ${\bf w}\overline{{\bf w}} \in S$ has Hamming weight exactly 7,
\emph{and} (as a consequence of $C^{\perp}$ being
a linear code where  every nonzero word has weight 4)
that for any ${\bf w}\not = {\bf w'}$ with both
${\bf w}, {\bf w'} \in C^{\perp}$,
 the number of common bit positions where both
${\bf w}\overline{{\bf w}}$
and ${\bf w'}\overline{{\bf w'}}$ have bit 1
is always 3.

% for w=0, and w'=0, then clearly # of common (1\\1) positions is 7.
% for w =0, and w' \not =0,
%  all (1\\1) positions are in 2nd half. since wt(w') =4,
% wt(\overline{w'}) = 3
%
% for w \not =0. and w' =w:
% exactly 7 positions are (1\\1)
%   and w' \not =0, and w' \not = w. and w' \not =0:
% in 1st half, dist(w, w') = 4, and both have wt 4, so # of
% (0\\1)  and  (1\\0) are the same, and total make up 4.  and so erach is 2.
% so among 7 exactly 2 each are different bits (0\\1) or (1\\0).
% so exactly 3 out of 7 are the same (0\\0) or (1\\1).
% every (0\\0) translates to (1\\1) in the second half.
% so total # of (1\\1) is exactly 3.
%
% upshot: for  w = w', the common (1\\1) in all 14 positions is 7.
% for w' \not = w,  the common (1\\1) in all 14 positions is 3.
% this is \sqrt^{ the common (1\\1) } that we multiply to check
% value at w' for the Transformed function T_w(f).
% the Transformation as a multiplier function , on support , pointwise
% multiplies original f, to make it in  Affine.
%
% so, e.g., if f has support dim 3. of 8 points exactly one is -1, rest are 1
% then this multiplier function just make it to Affine.

 For Holant problems with general (not necessary symmetric) functions,
the only known dichotomy is for a restricted
class called Holant$^*$ problems~\cite{CaiLX11}, where all unary functions
 are assumed to be available. How to extend this is a challenging
 open question. A very broad subclass of Holant problems is called
Holant$^c$, where only two unary pinning functions $\Delta_0$, $\Delta_1$
(that set a variable to 0 or 1) are assumed
to be available.  Holant$^c$ already covers a lot of ground, including
all of \#CSP, graph matching and so on.

 \#CSP is the special case of
 Holant problems where the constraint function set
is assumed to contain  {\sc Equality}
 of all arities.
 One can show that if we have an  {\sc Equality}
of odd arity  at least 3, we can realize  {\sc Equality} of all arities.
 But, if we have an  {\sc Equality}  of even arity, we can only realize
 {\sc Equality}  of even arities.
Dyer, Goldberg and Jerrum~\cite{weightedCSP} proved that
in the \#CSP framework one can realize the pinning functions
 $\Delta_0$ and $\Delta_1$.
% the 0-1 pinning functions.
  We will denote by \#CSP$_2^c$ the special case of \#CSP  where
 each variable appears an even number of times, and  $\Delta_0$,  $\Delta_1$
%0-1 pinning functions
are available.
  \#CSP$_2^c$ plays
 an important role.
% in the study of dichotomy theorems.
  A dichotomy for \#CSP$_2^c$ is somewhat unavoidable to get a
dichotomy for Holant. This is not only logically true in the sense
that a dichotomy for Holant will imply a dichotomy for \#CSP$_2^c$,
but also true in the sense that one usually proves
 a dichotomy \#CSP$_2^c$ as a major step toward a
 dichotomy of Holant~\cite{realholant,complexholant}.
 Previously one could only prove dichotomy for \#CSP$_2^c$
 for symmetric functions.
  Compared to the dichotomy for \#CSP, we already know that there is one more
 tractable family in the dichotomy for symmetric \#CSP$_2^c$.
It is a slight modification of the family $\mathscr{A}$, which
 is denoted by  $\mathscr{A}^\alpha$.
 Is this the only addition when we go from \#CSP to \#CSP$_2^c$
without the symmetry restriction?
%The answer is affirmative for symmetric functions, but how about general functions?
%CSP_2

\subsection{Our Results}
%In this paper, we resolved these two open questions.
%
In this paper we prove a complexity dichotomy for
 Holant$^c$ with general (not necessary symmetric)
real valued functions. In order to do that
we first prove a dichotomy for \#CSP$_2^c$
 with general (not necessary symmetric) \emph{complex} valued functions.
In addition to
 the two tractable families $\mathscr{P}$ and  $\mathscr{A}$ for
 \#CSP, and the known modification $\mathscr{A}^\alpha$,
 we discover a brand-new tractable family, denoted by
 $\mathscr{L}$, which we call
 \emph{local affine functions}.
 The dichotomy for \#CSP$_2^c$ says that these four ($\mathscr{P}$, $\mathscr{A}$, $\mathscr{A}^\alpha$ and $\mathscr{L}$) are exactly all the
 tractable families.
 The dichotomy for
 Holant$^c$ problems basically says that
the tractable family for Holant$^c$ is precisely the union of
tractable families of \#CSP$_2^c$ and Holant$^*$.
%In this reduction, we show that we can interpolate and realize all the unary functions or equality functions with even arity.

Conceptually (and also technically but somewhat hidden),
the most important contribution of this work is the discovery
and identification of the new tractable family $\mathscr{L}$.
The formal definition and characterization is given in Section~\ref{sec:local-affine}.
Our function $f$ of arity 14 is among its smallest examples.
 Given the succinct mathematical definition of $\mathscr{L}$,
the description of
 the algorithm is very short. However, we would like to point out
 that this formal simplicity hides many interesting and
surprising structures.

For reasons that will become clearer, we will now denote our function
$f$ of  arity 14 as  $f_{7}^\alpha(+-)$.
Five years ago, we discovered
a polynomial time algorithm for counting problems
defined by $f_{7}^\alpha(+-)$ (and some similar
functions) in the Holant$^c$
and  \#CSP$_2^c$ setting.
The algorithm is non-trivial.
But we were not able to prove a dichotomy.
% we first use a combinatorial parity argument
%for the outer  $(+-)$s, then argue that
%either the value is vanishing or we can eliminate all these
%  $(+-)$s to get an instance of  $\mathscr{A}^\alpha$ and
% thus a polynomial time algorithm.

 %which stuck us for five years.
Let's consider another 0-1 valued function $f_{31}$:
 %In \#CSP$_2^c$, let's consider another
It has arity $31$. It is the 0-1 indicator function
of a (particular kind of) 5-dimensional linear subspace $S$ of
 $\mathbb{Z}_2^{31}$. Five of  $31$ variables are considered
 free variables and all $31$ variables on $S$ correspond to exactly
all possible non-empty linear combinations of the five free variables.
 The function  $f_{31}$ is a pure affine function in $\mathscr{A}$,
 and known to be tractable alone. On the other hand,
 the function $f_{7}^\alpha(+-)$  is neither in $\mathscr{A}$
nor in $\mathscr{A}^\alpha$, but we also had a polynomial time algorithm
for $f_{7}^\alpha(+-)$ type functions alone.
% But for $f_{7}^\alpha(+-)$ type functions alone,
% we also had a polynomial time algorithm five years ago.
 % All that was known to us five years ago.
 The real challenge, for the quest of a dichotomy,
is to put them together.
 What is the complexity for Holant$^c(f_{31}, f_{7}^\alpha(+-))$
or \#CSP$_2^c(f_{31}, f_{7}^\alpha(+-))$?
 If we replace $f_{31}$ with a smaller arity but of the
 same structure such as $f_{15}, f_7, f_3, f_1$, we can prove that
 the problem is \#P-hard.
It seems highly implausible that
%It seems inconceivable that
tractability would start to show up only at such high arity.
And so we conjectured that \#CSP$_2^c(f_{31}, f_{7}^\alpha(+-))$
is also \#P-hard. We tried to prove this for five years but failed.
We also tried to find a P-time algorithm without success,
 until now. It is quite tantalizing to think about
 what property is shared by  $f_{7}^\alpha(+-), f_{31}, f_{63}, \ldots$
 but not with $f_{15}, f_7, f_3, f_1$?
%It is very unclearly.

We now know that the explanation is this  new family $\mathscr{L}$.
Interestingly, the deceptively simple definition of $\mathscr{L}$ does
 include $f_{7}^\alpha(+-), f_{31}, f_{63}, \ldots$  but excludes
  $f_{15}, f_7, f_3, f_1$.
(This fact can be verified but is  not totally trivial.)
By the unifying notion of $\mathscr{L}$, we also have a
much simpler description of a polynomial time algorithm,
which starts with a global linear system and
a localized holographic transformation performed simultaneously everywhere.
(Because this description is much simpler, we will not describe
our earlier algorithm in this paper.)

Several facts about $\mathscr{L}$ are worth mentioning.
  These interesting structures can only appear for general functions
but not for symmetric ones. Secondly, although the definition of $\mathscr{L}$
 seems to involve complex numbers in an essential way, it does include
 some real valued functions such as $f_{7}^\alpha(+-)$.
 We cannot avoid going through $\mathbb{C}$ even if we only hope to
prove a dichotomy for real valued functions.
 Although the algorithm for  $\mathscr{L}$ looks short, it does
 have a very different nature compared to that for
 $\mathscr{P}$, $\mathscr{A}$ and $\mathscr{A}^\alpha$.
 The algorithms for previous known tractable families basically
perform a local elimination to handle the variables one by one.
 The algorithm $\mathscr{L}$ contains a global step, which is
to solve a global linear equation, followed by a
localized holographic transformation simultaneously everywhere.
We have tried many purely
 local algorithms and failed,
 until we reached this global algorithm.

\subsection{Techniques by Examples}
Let us first describe the proof that
\#CSP$_2^c(f_{15}, f_{7}^\alpha(+-))$ is \#P-hard.
Here $f_{15}$ is a 0-1
indicator function
of a 4-dim linear subspace $S$ of
 $\mathbb{Z}_2^{15}$; 4 variables are chosen as
 free variables and all $15$ variables on $S$ correspond to
all their non-empty linear combinations.
An instance of \#CSP$_2^c(f_{15}, f_{7}^\alpha(+-))$
is a bipartite graph $(V, U, E)$ where $V$ are variables,
$U$ are constraint functions from $\{f_{15}, f_{7}^\alpha(+-), \Delta_0,
\Delta_1\}$ and $E$ indicates how the constraints are applied. Being
in \#CSP$_2$, every $v \in V$ has even degree.
The Sum-of-Product computation is to evaluate $\sum_{\sigma: V \rightarrow
 \{0,1\}} \prod_{u \in U} f_u({\sigma})$,
where $f_u$ is the function at $u \in U$.

If a variable appears exactly twice, once in $\Delta_1$
and once as an input to $f_{15}$, this effectively pins that input of
$f_{15}$ to 1. This creates a function $g$ of arity 14,
which is ``realizable'' in \#CSP$_2$.
%; we say it is obtained by applying
%$\Delta_1$ to that
What is $g$? Even though $f_{15}$ is \emph{not} symmetric,
clearly not every subset of 4 variables can be chosen as free,
every single variable \emph{can} be free (as part of a subset of 4).
In group terminology, the symmetry group of  $f_{15}$ is
not ${\frak S}_{15}$, but there is a transitive group
of symmetry ${\bf GL}_4(\mathbb{Z}_2)$ acting on the nonzero vectors
of $\mathbb{Z}_2^4$.

Hence up to renaming the variables, $g$ is the same as
 setting $x_4$ of $f_{15}$ to 1.  This function has exactly
the same support structure as $f_{7}^\alpha(+-)$,
but the function values are all $1$
on its support, whereas $f_{7}^\alpha(+-)$ has value $-1$ when the 3
 free variables are all equal to $1$. We call this new function $f_7(+-)$.
For both  functions $f_{7}^\alpha(+-)$ and $f_7(+-)$
we can divide the 14 inputs into  $7$ pairs  in the same way,
%The functions $f_{7}^\alpha(+-)$ and $f_7(+-)$ have exactly
%the same support structure: the inputs are divided into  $7$ pairs
% (they
%always take opposite 0-1 values on the support),
which will be called bundles in this paper; each bundle has two input
variables which always take opposite 0-1 values on the support; among the $7$ bundles they have the same linear relation. Therefore, we can combine them in the following straightforward way:
for each corresponding bundle connect the two variables
labeled $(-)$, one from each bundle,
 and leave the  variables labeled  $(+)$ as inputs of the gadget. Technically
we have a  \#CSP$_2$ construction where
 for each corresponding bundle there is a variable that appears
exactly twice, once for each variable labeled $(-)$ in the bundle.
This gadget  realizes
a function $h$ with  $14$ inputs in $7$ bundles as well. The two inputs in each bundle must have the same value on the support,
 and the value of $h$ is the same as $f_{7}^\alpha(+-)$,
since $f_7(+-)$ is identically 1 on the support. So, this function can be denoted as $f_{7}^\alpha(++)$.
There is an easy reduction
\#CSP$(f_{7}^\alpha)\leq_T$\#CSP$_2^c(f_{7}^\alpha(++))$:
In any instance of \#CSP$(f_{7}^\alpha)$, 
replicate \emph{twice} every occurrence of variables in
constraints, and replace $f_{7}^\alpha$ by $f_{7}^\alpha(++)$.
Hence \#CSP$(f_{7}^\alpha)
  \leq_T$ \#CSP$_2^c(f_{15}, f_{7}^\alpha(+-))$. 
As $f_{7}^\alpha \not \in \mathscr{A} \cup \mathscr{P}$, 
\#CSP$(f_{7}^\alpha)$ is  \#P-hard.
%is 
%neither in $\mathscr{A}$ nor $\mathscr{P}$, it is  \#P-hard.

\vspace{.05in}
%Then the next question is:
How about \#CSP$_2^c(f_{7}, f_{7}^\alpha(+-))$?
Here $f_7$ has arity 7 and a support of dimension 3.  Similarly,
if  we pin a variable of $f_{7}$ to $1$ we get $f_3(+-)$.
But $f_{7}^\alpha(+-)$ and $f_3(+-)$ do not have the same
support structure.   Then we need the following more complicated
 gadget as shown in Fig.~\ref{Fig f3f7} to construct  $f_{7}^\alpha(++)$.

\begin{figure}[h]
	\begin{center}
		\includegraphics[width=0.6\textwidth]{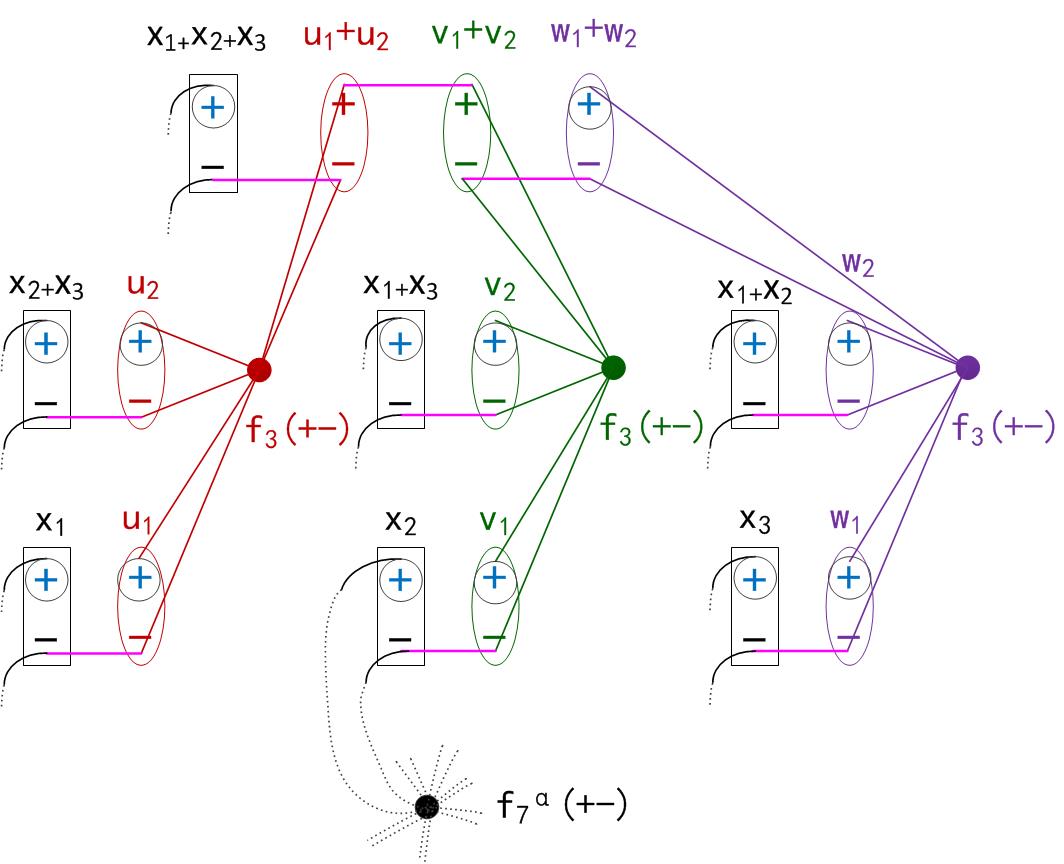}
	\caption{The gadget realizing $f_{7}^\alpha(++)$ is composed of one
copy of  $f_{7}^\alpha(+-)$ and 3 copies of $f_3(+-)$, shown as the 4 dots in the picture. Each $f_3(+-)$ function has  6 edges stretched out as input variables, which are grouped into 3 bundles, shown as ellipses. The function
 $f_{7}^\alpha(+-)$ has  14 edges stretched out shown
as dotted lines (only one pair is
completely shown). They are grouped  into 7 bundles, 
shown as rectangles. 9 pairs of variables
 are connected (shown as horizontal pink line segments), leaving
$3 \times 6 + 14 - 9 \times 2=14$ variables exposed as external variables
of the gadget, which are circled.}
	\label{Fig f3f7}
	\end{center}
\end{figure}
%
%\vspace{-1.5in}
%%% JYC
%%% here seems to be a lot of space wasted between the fig and the text below
% i can't make it disappear
%
%
%
We observe that this construction is very delicate.
After connecting the variables $u_1(-)$ and, resp. $u_2(-)$,
of one copy of $f_3(+-)$ with $x_1(-)$ and, resp. the variable labeled
 ``$x_2+x_3\/(-)$'', of $f_{7}^\alpha(+-)$,  we have in fact forced
the value of the variable labeled ``$u_1 + u_2\/(-)$''  of $f_3(+-)$ to equal
(on the support)
to the
variable labeled ``$x_1 + x_2+x_3\/(-)$'' of $f_{7}^\alpha(+-)$.
 Similarly after  connecting $v_1(-)$ and $v_2(-)$
with $x_2(-)$ and  ``$x_1+x_3\/(-)$'' of $f_{7}^\alpha(+-)$,
the  value of ``$v_1 + v_2\/(-)$''
is also forced to equal  ``$x_1 + x_2+x_3\/(-)$'', but that
variable has already been taken.
On the other hand,
after  connecting $w_1(-)$ and $w_2(-)$
with $x_3(-)$ and  ``$x_1+x_2\/(-)$'' of $f_{7}^\alpha(+-)$,
the  value of ``$w_1  + w_2\/(-)$''
is also forced to equal  ``$x_1 + x_2+x_3\/(-)$''.
Hence it is legitimate to connect ``$v_1 + v_2\/(-)$''
with ``$w_1  + w_2\/(-)$'' (finding a home for both orphans.)
Similarly,  both ``$u_1 + u_2\/(+)$''
and ``$v_1 + v_2\/(+)$'' are 
 forced to equal ``$x_1 + x_2+x_3\/(+)$'' (on the support),
hence  connecting them is also legitimate. 
In the meanwhile, the pair $x_1(+)$ and $u_1(+)$ must be equal
on the support, forced by the connection between  $u_1(-)$ and  $x_1(-)$,
making them a $(++)$ pair.
Similarly, $x_2(+)$ and $v_1(+)$ are forced to equal 
on the support, making a  $(++)$ pair,
and $x_3(+)$ and $w_1(+)$ are forced to equal
making another  $(++)$ pair.
Then ``$x_2 + x_3\/(+)$'' and $u_2(+)$ make a $(++)$ pair,
but this bundle satisfies the linear dependence that
it is equal to the sum of  the two free variables $x_2(+)$ and $x_3(+)$
on the support. The same can be said for the other 3 dependent bundles.
In all, it is clear that the  7 exposed pairs of variables
associated with $x_1(+)$, $x_2(+)$, $x_3(+)$
``$x_1+x_2\/(+)$'', ``$x_1+x_3\/(+)$'', ``$x_2+x_3\/(+)$''
and ``$x_1 + x_2+x_3\/(+)$'' form 7 bundles of equal variables,
and  they have precisely the 3-dimensional support structure in a 14-dimensional
space, as described for  $f_{7}^\alpha(+-)$.
As the value of $f_3(+-)$ is always 1 on the support,
it is clear that the function of the gadget is
$f_{7}^\alpha(++)$.

The above construction and proof of hardness are special cases of our Lemma \ref{lemma: together two +-}. We note that these functions are really sparse.
 For example the function $f_{7}^\alpha(+-)$ has
 only $8$ nonzero values  out of $16384$ $(=2^{14})$ values total.
They are very ``fragile'': If you do not make the connection
``just so'', then chances are that
the construction will collapse and no good reduction can be obtained.
On the other hand, precisely because of their  delicate structure
one can come up with extremely intricate designs. At the same
time, the interesting structure may also portend some unforeseen algorithms.

We will use these delicate structures to prove \#P-hardness.
 But to prove a dichotomy theorem, one needs to prove that an
arbitrary function set  not contained in
one of the 4 tractable families is \#P-hard. The given functions may not have any of the nice structure, then how can
we do the construction? To handle that, we have a number of regularization lemmas in Section~\ref{subsec:RegularizationLemmas}, showing that one
can always construct gadgets to regularize the
functions.  Starting with any function set not contained in
one of the tractable families,
we can produce functions with similar nice structures but still
 outside the respective tractable families
 unless we already can prove \#P-hardness outright.

Then, the question is
 how about \#CSP$_2^c(f_{31}, f_{7}^\alpha(+-))$?
Can we also construct the function $f_{7}^\alpha(++)$ or other functions to get \#P-hardness?
If we pin two free variables of $f_{31}$, we get either $f_{7}(++++)$ 
or $f_{7}(++--)$. They are not like $f_{7}(+-)$. With these, together with $f_{7}^\alpha(+-)$, we do not know how to construct functions like $f_{7}^\alpha(++)$ as before.
All attempts to construct similar gadgets like in
 Figure~\ref{Fig f3f7} failed.
Now as a consequence of our dichotomy theorem,
 assuming $\#{\rm P}$ is not equal to ${\rm P}$,
we can prove that no such construction can succeed.
The reason is that they both belong to the new tractable family $\mathscr{L}$.
 The algorithm for that and a characterization for $\mathscr{L}$
 is given in Section~\ref{sec:local-affine}.
 The criteria there can be used to show that $f_{31}, f_{7}^\alpha(+-)$
 are in the family  $\mathscr{L}$ while $f_{15}, f_7, f_3, f_1$ are not.

%The algorithm for that and a characterization showing that $f_{31}, f_{7}^\alpha(+-)$ are in the family while $f_{15}, f_7, f_3, f_1$ are not, can be found in Section \ref{sec:local-affine}.

%%%JYC:  actually in sec3 do we really show that explictly that
% $f_{31}, f_{7}^\alpha(+-)$ are in the family while $f_{15}, f_7, f_3, f_1$ are not, ???
% i think now. but only implicitly.
% and in fact earlier in the paper in intro sec, line3 of page 3
% i said this  "can be verified by not trivial."
% so i suggest you reconcile these two sayings here.
% probably remove the above
% only say:
%
% i suggest the following wording:
%
% The algorithm for that and a characterization for $\mathscr{L}$
% is given in Section~\ref{sec:local-affine}.
% The criteria there can be used to show that $f_{31}, f_{7}^\alpha(+-)$
% are in the family  $\mathscr{L}$ while $f_{15}, f_7, f_3, f_1$ are not.

\section{Preliminaries}
A (constraint) function of arity $n$ is a mapping from $\{0,1\}^n \rightarrow \mathbb{C}$.
We denote by $=_n$ the {\sc Equality} function of arity $n$. A
symmetric function $f$ on $n$ Boolean variables can be expressed by
$[f_0,f_1,\ldots,f_n]$, where $f_j$ is the value of $f$ on inputs of
Hamming weight $j$. Thus, $(=_n)=[1,0,\ldots,0,1]$ (with $n-1$
zeros). We also use $\Delta_0, \Delta_1$ to denote $[1,0]$ and $[0,1]$
respectively. A binary function $f$ is also expressed by the matrix
$\begin{bmatrix}f(0,0)& f(0,1)\\
f(1,0)& f(1,1)\end{bmatrix}$.

A {\it signature grid} $\Omega = (G, {\mathscr F}, \pi)$ consists of
a graph $G=(V,E)$, and a labeling $\pi$ of each vertex $v \in V$
with a function $f_v \in
{\mathscr F}$.
 The Holant problem on instance
$\Omega$ is to compute
${\rm Holant}_\Omega=\sum_{\sigma: E \rightarrow \{0,1\}}
\prod_{v\in V} f_v(\sigma|_{E(v)})$, where $\sigma|_{E(v)}$ is the assignment $\sigma$ restricted to the edges incident to $v$.
A Holant problem is parameterized by a set of functions.
\begin{definition}[Holant]
Given a set of functions ${\mathscr F}$, we define a counting
problem ${\rm Holant}({\mathscr F})$:

Input: A {\it signature grid} $\Omega = (G, {\mathscr F}, \pi)$;

Output: ${\rm Holant}_\Omega$.
\end{definition}

Suppose $c \in \mathbb{C}$ is  a nonzero number. As constraint
functions $f$ and $cf$ are equivalent in terms of the complexity of
Holant problems they define. Hence we will consider functions
$f$ and $cf$ to be interchangeable.
We would like to characterize the complexity of Holant problems in
terms of its function sets\footnote{
We allow ${\mathscr F}$ to be an infinite set.
${\rm Holant}({\mathscr F})$ is tractable means that
it is computable in P even when we include the description of the
functions in the input $\Omega$ in  the input size. ${\rm
Holant}({\mathscr F})$ is   \#P-hard  means that there exists a
finite subset of ${\mathscr F}$ for which the problem is  \#P-hard.
For considerations of
models of computation, function values are algebraic in ${\mathbb{C}}$.}
Some special families of Holant problems have already been widely
studied. For example, if all the  {\sc Equality} functions are in ${\mathscr F}$
then this is exactly the
weighted \#CSP problem.
Other well-studied special families of Holant are Holant$^*$ and Holant$^c$.
\begin{definition}
Let ${\mathscr U}$ denote the set of all unary functions.
Then ${\rm Holant}^* ({\mathscr F})
= {\rm Holant}({\mathscr F}\cup {\mathscr U})$.
\end{definition}

\begin{definition}
 ${\rm Holant}^c ({\mathscr F})
= {\rm Holant}({\mathscr F}\cup \{\Delta_0, \Delta_1\})$.
\end{definition}

\#CSP$({\mathscr F})$ is equivalent to
${\rm Holant}({\mathscr F}\cup \{\Delta_0, \Delta_1, =_1, =_2, =_3, \cdots\})$.
We define \#CSP$_2^c$ as the follows.
%ing special Holant problem.
\begin{definition}
 ${\rm \#CSP}_2^c ({\mathscr F})
= {\rm Holant}({\mathscr F}\cup \{\Delta_0, \Delta_1, =_2, =_4, =_6, \cdots\})$.
\end{definition}

In the above definitions, the functions and domain  $\{0,1\}$ are without
structures.
However, as we describe the complexity classification of counting problems,
especially for tractable problems, we may assign structures
to the domain and the functions.
We may consider polynomials in $\mathbb{Z}[x_1,x_2\ldots,x_n]$
with each $x_i$ taking values from $\{0,1\} \subseteq \mathbb{Z}$;
the evaluation in $\mathbb{Z}$.
In another setting,
we may consider the domain
as a finite field $\mathbb{Z}_2$ of size $2$, and  $\{0,1\}^n$
%$\mathbb{F}_2$ of size $2$, and  $\{0,1\}^n$
as a vector space of dimension $n$ over $\mathbb{Z}_2$.

%Each $x_i$ is from a domain of size $2$. A string $x=x_1,x_2,\ldots,x_n$ (maybe also denoted as $x_1x_2\cdots x_n$) of length $n$, indicates a value $f(x_1,x_2,\ldots,x_n)$ of the function. To work as the inputs or indices of $f$, the domain has not any structures. This is already enough to define the counting problems.
%
%However, the complexity classification of counting problems, especially the part of tractable problems, brings us interesting structures of the functions. To define and study these problems, we give the domain elements some operations and structures.
%When we consider polynomials over integers in $x_1,x_2,\ldots,x_n$. The value of each $x_i$ is from $\{0,1\}$, and they are evaluated as integers.
%When we talk about an affine subspace of $\{0,1\}^n$, the domain is looked as a finite field $\mathbb{F}_2$ of size $2$.

\begin{definition}[Support]
The underlying relation, also called the support of a function $f$ is given by ${\rm supp}(f)=\{x \in \{0,1\}^n| f(x) \neq 0\}$.
\end{definition}

We say a relation $R \subseteq \{0,1\}^n$ is affine if it is
an affine linear subspace of $\mathbb{Z}_2^n$.
It is composed of solutions of some system  $Ax=b$
of affine linear equations over $\mathbb{Z}_2$.
%%% I don't think we used that in this paper.
%equivalently, if $a,b,c \in R$, then
%$a\oplus b\oplus c \in R$.
If $\text{supp}(f)$ is affine, we say
$f$ has affine support.
% which has form $\{x| Ax=b\}$.
We also view this relation as a 0-1 valued indicator function $\chi_{Ax=b}$.
% $\chi_{Ax=B}$ from $\{0,1\}^n$ to $\{0,1\}$.

\begin{definition}[Compressed function \cite{CaiFu16}]
If $f$ has affine support of dimension $r$,
and $X = \{x_{j_1},\ldots,x_{j_r}\} \subseteq \{x_1,x_2, \ldots,x_n\}$
 is a set of free variables for ${\rm supp}(f)$, then
$\underline{f_X}$ is the compressed function of $f$ for $X$ such that $\underline{f_X}(x_{j_1},\ldots,x_{j_r})=f(x_1,x_2, \ldots,x_n)$, where $(x_1,x_2, \ldots,x_n) \in {\rm supp} (f)$. When it is clear from the context, we omit $X$ and use $\underline{f}$ to denote $\underline{f_X}$.
\end{definition}

% Obviously, the values of $v$ on elements from ${\rm supp}(f)$, are decided by $f$.

If $f$ has affine support,
then $r = \dim {\rm supp}(f)$
 is called the rank of $f$.
Usually, we may rename variables
so that $x_1,x_2 \ldots x_r$ is a set of free variables.

\begin{definition}[Product type: $\mathscr{P}$]
$\mathscr{P}$ denotes the class of functions which can be expressed
as a product of unary functions, binary equality functions
($[1,0,1]$) and binary disequality functions ($[0,1,0]$).
\end{definition}

\begin{definition}[Affine: $\mathscr{A}$]
$\mathscr{A}$ denotes all functions $f: \{x_1, x_2, \ldots, x_n\}\rightarrow \mathbb{C}$ satisfying the following conditions:
\begin{itemize}
\item ${\rm supp}(f)$ is affine $\chi_{(Ax=b)}$.
% function $\chi_{(Ax=B)}$.
\item  Assume  $x_1,x_2,\ldots,x_r$ are free variables.
% of
% $\chi_{(Ax=b)}$.
% $\chi_{(Ax=b)}$ are $x_1,x_2,\ldots,x_r$.
The compressed function of $f$ is
$\lambda \cdot i^{L(x_1,\ldots, x_r)+2Q(x_1,\ldots, x_r)}$
(for some nonzero constant $\lambda \in \mathbb{C}$)
%(up to a nonzero constant multiplier),
where $L$
is an integer coefficient linear polynomial
% in $x_1,x_2,\ldots,x_r$,
% of degree 1,
 and $Q$
is an integer coefficient multilinear polynomial
%in $x_1,x_2,\ldots,x_r$
where each monomials has degree 2.
\end{itemize}
Of course, $f$ is the product $
\lambda \cdot  \chi_{(Ax=b)} \cdot i^{L(x_1,\ldots, x_r)+2Q(x_1,\ldots, x_r)}$.
\end{definition}

We use  $\alpha$ to denote $e^{\frac{\pi}{4}i} = \frac{1+i}{\sqrt{2}}$,
a square root of $i$. The notation
$a \equiv b$ means $a =b \mod 2$.

%$\alpha=e^{\frac{\pi}{4}i}$.

%
A matrix $M \in \mathbb{C}^{2 \times 2}$ defines
a holographic transformation $f \mapsto  M^{\otimes n}(f)$,
% M^{\otimes {\rm arity}(f)}(f)$,
where we list the values of $f$ as a column vector indexed by $\{0,1\}^n$.
Let $M_\alpha=\begin{bmatrix}1 & 0\\
0 & \alpha \end{bmatrix}$, and $M_1=I_2=\begin{bmatrix}1 & 0\\
0 & 1 \end{bmatrix}$.
%Define $\mathscr{A}^ \alpha = \{M_\alpha^{\otimes {\rm arity}(f)} (f) \mid f \in \mathscr{A}\}$.
%We use $\mathscr{A}^ \alpha$ to denote the set functions
%obtained by applying the $M_\alpha$ transformation to functions
% from $\mathscr{A}$.
%The inverse transformation of  $M_\alpha$
%is $M_{\alpha^{-1}}$.
%We write $f$ as a column vector of dimension $n$ indexed by $\{0,1\}^n$.
% Equivalently, for any function in $\mathscr{A}^ \alpha$, if apply the inverse of $M_\alpha$ transformation to it, we get a function in $\mathscr{A}$.

\begin{definition}[$\alpha$ dual affine: $\mathscr{A}^ \alpha$]
 $\mathscr{A}^ \alpha = \{M_\alpha^{\otimes {\rm arity}(f)} (f) \mid f \in \mathscr{A}\}$.
\end{definition}
The inverse transformation of  $M_\alpha$
is $M_{\alpha^{-1}}$.
A function $f$ is in $\mathscr{A}^\alpha$ iff $M_{\alpha^{-1}}^{\otimes n} f$ is in $\mathscr{A}$.

\begin{theorem} \label{thm old csp dichotomy}
A $\#\text{\rm CSP}(\mathscr{F})$ problem has polynomial time algorithm, if one of the following holds,
\[\mathscr{F} \subseteq \mathscr{P} ~~~~~\mbox{or}~~~~~\mathscr{F} \subseteq \mathscr{A}.\]
%\begin{itemize}
%\item $\mathscr{F} \subseteq \mathscr{P}$,
%\item $\mathscr{F} \subseteq \mathscr{A}$.
%\end{itemize}
Otherwise, it is \#P-hard.
\end{theorem}

The following two families of functions are used in the dichotomy for $\#\text{Holant}^*(\mathscr{F})$.
$\mathcal{M}$ is the set of all functions $f$ such that $f$ is zero
except on $n+1$ inputs whose Hamming weight is at most $1$, where $n$ is the
arity of $f$. The name $\mathcal{M}$ is given for {\it matching}.
$\mathcal{T}$ is the set of all functions of arity at most $2$.
To discuss the complexity of Holant problem,
we may always remove identically zero functions.

\begin{lemma}\label{lemma-decomposable}
Let $f=g\otimes h$, none of them identically 0.
 Then Holant$^c(\mathscr{F} \cup \{f\}) \equiv_{\rm T}$ Holant$^c(\mathscr{F} \cup \{g, h\})$. Holant$^*(\mathscr{F} \cup \{f\}) \equiv_{\rm T}$ Holant$^*(\mathscr{F} \cup \{g, h\})$.
\end{lemma}

So we only work with functions which cannot be further decomposed.

\begin{theorem} \label{thm old holant-star-dichotomy}
Let $\mathscr{F}$ be a set of non-decomposable functions.
Then $\#\text{Holant}^*(\mathscr{F})$ problem has polynomial time algorithm, if one of the following holds,
\[\mathscr{F} \subseteq \mathscr{T} ~~\mbox{or}~~\mathscr{F} \subseteq H\mathscr{P}~~\mbox{or}~~\mathscr{F}
\subseteq Z\mathscr{P} ~~\mbox{or}~~\mathscr{F} \subseteq Z\mathscr{M},\]
where $H$ is an orthogonal matrix and  $Z=\begin{pmatrix} 1 &  1  \\
                i  & -i
\end{pmatrix}$, or
$\begin{pmatrix} 1 &  1  \\
                -i  & i
\end{pmatrix}$.
Otherwise, it is \#P-hard.
\end{theorem}

%These four families are all explicitly given in~\cite{CaiLX11}, but we do not need the exactly formula of these four families in this paper. We need the following two: one tractable family is functions which can be decomposed into functions of arity no larger than $2$; another tractable family is $Z\mathscr{P}$, where  $Z=\begin{pmatrix} 1 &  1  \\
%                i  & -i
%\end{pmatrix}$.
%
%%
%
%
%The next definition is crucial for this work.
%\begin{definition}[Local affine: $\mathscr{L}$]\label{def:local-affine}
%A function $f$ is in $\mathscr{L}$, if and only if for each $\sigma=a_1 a_2 \cdots a_n \in \{0, 1\}^n$ in the support of $f$, $(M_{\alpha^{a_1}} \otimes M_{\alpha^{a_2}} \otimes \cdots \otimes M_{\alpha^{a_n}} ) f$ is in $\mathscr{A}$.
%\end{definition}
%
%The notation $M_{\alpha^{a_j}}$ just means that
%% helps to describe the matrices used in a a holographic reduction. When
%when $a_j=0$ the $j$th input is transformed by the
% identity matrix (in fact, not transformed) and when $a_j=1$
% the $j$th input is transformed by $M_\alpha$.
%
%\begin{lemma}
%All {\sc Equality} functions of even arity $2k$, $(=_{2k})=[f_0=1,f_1=0,\ldots, f_{2k}=1]$,
% are in $\mathscr{P}$, $\mathscr{A}$, $\mathscr{A}^\alpha$, and $\mathscr{L}$.
%%
%If $f$ is in  $\mathscr{P}$  (resp. $\mathscr{A}$,  $\mathscr{A}^\alpha$,
% $\mathscr{L}$), then $\frac{1}{f}$ is in $\mathscr{P}$  (resp. $\mathscr{A}$, $\mathscr{A}^\alpha$, $\mathscr{L}$).
%\end{lemma}
%
%

\section{Local Affine Functions}\label{sec:local-affine}
The next definition is crucial for this work.

\begin{definition}[Local affine: $\mathscr{L}$]\label{def:local-affine}
A function $f$ is in $\mathscr{L}$, if and only if for each $\sigma=s_1 s_2 \cdots s_n \in \{0, 1\}^n$ in the support of $f$, $(M_{\alpha^{s_1}} \otimes M_{\alpha^{s_2}} \otimes \cdots \otimes M_{\alpha^{s_n}} ) f$ is in $\mathscr{A}$.
\end{definition}

The notation $M_{\alpha^{s_j}}$ just means that
% helps to describe the matrices used in a a holographic reduction. When
when $s_j=0$ the $j$th input is transformed by the
 identity matrix (in fact, not transformed) and when $a_j=1$
 the $j$th input is transformed by $M_\alpha$. This is very interesting since each arity performs a possible different holographic transformation. This is why we use the term "local" to name this family of functions.

 By Definition~\ref{def:local-affine},
$f \in \mathscr{L}$ if and only if
for each $\sigma=s_1 s_2 \cdots s_n
\in {\rm supp}(f)$, the transformed function
 ${M}_{\sigma} f: (x_1, \ldots, x_n) \mapsto
\alpha^{\sum_{i=1}^n s_i x_i} f(x_1, \ldots, x_n)$
is in $\mathscr{A}$.
Here each $s_i$ is a 0-1 valued integer, and the sum $\sum_{i=1}^n s_i x_i$
is evaluated as an integer (or an integer mod 8).

%The function ${M}_{\sigma} f$ can be thought of transforming
%the $i$-th input by $M_\alpha =
%\begin{bmatrix} 1& 0\\
%0& \alpha \end{bmatrix}$ if $s_i=1$, and do not transform the
%$i$-th input if $s_i=0$.
%%
%%The notation $M_{\alpha^{s_j}}$ denotes the matrix used in a holographic reduction: When $s_j=0$, the $j$-th input is transformed by
%%the 2 by 2
%% identity matrix (i.e., not transformed); when $s_j=1$, the $j$-th input is transformed by $M_\alpha =
%%\begin{bmatrix} 1& 0\\
%0& \alpha \end{bmatrix}$.

Of course the identically 0 function belongs to $\mathscr{L}$.
If $f \in \mathscr{L}$ and is not identically 0,
then there is some $\sigma \in {\rm supp}(f)$,
such that ${M}_{\sigma} f \in \mathscr{A}$.
It is apparent by its form,
${M}_{\sigma} f = \alpha^{\sum_{i=1}^n s_i x_i} f$,  that
$f^2\in\mathscr{A}$.
%in the modifier
%s_i are 0-1 bits x_i as 0-1 variables. as linear exponential on alpha
% its square gives a  linear on exponential of i. so in A.
Since $f$ and $f^2$ have the same support, it follows that
${\rm supp}(f)$
  is an affine linear subspace over $\mathbb{Z}_2$.
Assume that $f$ has affine support  with
 free variables $x_1, \ldots, x_r$, let ${\rm supp}(f)$ be  described by
the  data $(A, b)$, where
$A \in \{0, 1\}^{n \times r}$ is a 0-1 integer matrix
of which the top $r \times r$ matrix is $I_r$,
and $b \in \{0, 1\}^{n}$ is a 0-1 integer  vector
of which the  first $r$ entries are all 0.
The $i$-th 0-1 variable $x_i$ ($1 \le i \le n$)
is expressed as  $x_i \equiv \sum_{j=1}^r a_{ij} x_j  +  b_i  \bmod 2$.
Then by $f^2\in\mathscr{A}$, we have the following expression
\begin{equation}\label{generic-form-f-sq-in-A}
f = \lambda \cdot {\rm supp}(f)  \cdot \alpha^{\sum_{S \subseteq [r]} c_S \prod_{j \in S} x_j}
\end{equation}
where $\lambda \not =0$, $c_S \in \mathbb{Z}$ for all $S \subseteq [r]$.
This is easy to see because after a global scaling, each nonzero entry of $f^2$ is a power of $i$ and as a result
each nonzero entry of $f$ is a power of $\alpha$. Every such function can be expressed in the above formula.
Since $\alpha^8 =1$ we may consider all coefficients
$c_S$ belong to $\mathbb{Z}_8$.
The multilinear
polynomial in  $\mathbb{Z}_8 [x_1, \ldots, x_r]$
is unique for $f$.
%%%% the part can be omitted if short on space.
To see this,
take the quotient  $q(x)$ of two expressions on the
support, and write  $q(x) = {\rm supp}(f) \cdot \alpha^{P(x)}$.
If the multilinear polynomial $P(x) \in \mathbb{Z}_8 [x_1, \ldots, x_r]$
 is not identically zero, then let $S \subseteq [r]$
be of minimum cardinality such that $\prod_{j \in S} x_j$ is a term
with non-zero coefficient in $P(x)$.
Assigning $x_j$ to 1 for all $j \in S$, and all other $x_j$ to 0, for $j \in
[r] \setminus S$, shows that
$q(x)$ is not identically 1 on ${\rm supp}(f)$.
Then by the fact that $f^2\in\mathscr{A}$ we know that
$c_S \equiv 0 \bmod 2$ for $|S|=2$,  and
 $c_S \equiv 0 \bmod 4$ for
$|S| \ge 3$.
 We may normalize it so that $\lambda =1$ and
$c_{\emptyset} = 0$.  We can write it more explicitly as
\begin{equation}\label{generic-form-f-sq-in-A-detailed}
f =  \lambda \cdot {\rm supp}(f) \cdot  \alpha^{L(x) + 2 Q(x) + 4 H(x)}
\end{equation}
where  $L(x) = \sum_{j=1}^r c_j x_j$ is a linear function,
$Q(x) = \sum_{1 \le j < k \le r} c_{jk} x_j x_k$
 is a quadratic (multilinear)  polynomial,
and $H(x) = \sum_{1 \le j < k < \ell \le r} c_{jk\ell}  x_j x_k x_{\ell}
+ \cdots$
is a  (multilinear)  polynomial with all monomials  of
degree at least 3.

%We assume $r \ge 2$.
Any $(s_1, \ldots, s_r) \in \{0, 1\}^r$
determines  a unique point  $\sigma = (s_1, \ldots, s_r, s_{r+1},
\ldots, s_n) \in \{0, 1\}^n$ in ${\rm supp}(f)$,
from which we get the transformed function
$\alpha^{\sum_{i=1}^n s_i x_i} f \in \mathscr{A}$.
Here each $s_i$ is a 0-1 constant and $x_i$ is a 0-1 variable.

\subsection{Algorithm} \label{sec: L algorithm}

 The most interesting and surprising discovery of this work is the following polynomial time algorithm.

\begin{theorem}
There is a polynomial time algorithm for Holant$(\mathscr{L})$.
\end{theorem}

\begin{proof}
We first focus on the support but ignore the concrete value of the functions.  Since the support of the function at each vertex is affine, we can solve a linear system to get an assignment for all the edges which is on the support for all the functions. If the linear system does not have any solution, we can simply output zero since there is no assignment which can give possible non-zero value. Once we get a assignment for edges which is simultaneously on the support for all the functions, we can perform an $M_\alpha$ transformation on all the edges with assignment $1$. By the definition of $\mathscr{L}$, all the functions are in $\mathscr{A}$ after the transformation and the two inverse $M_\alpha$ transformations on a edge will also get a function in  $\mathscr{A}$.  Therefore,  we have an instance with same holant value but all the functions are $\mathscr{A}$. We know that there is a polynomial time algorithm to compute the holant value.
\end{proof}

It is clear that all the equality with even arity and the two constant unary function $\Delta_0, \Delta_1$ are in this family $\mathscr{L}$. So, we have the following corollary.

\begin{corollary} \label{corollary  csp2c L}
There is a polynomial time algorithm for \#{\rm CSP}$_2^c(\mathscr{L})$.
\end{corollary}

%We denote by ${\rm supp}(f)
%= \{ \sigma \in \{0,1\}^n \mid f(\sigma)  \neq 0\}$, the support of $f$.
%If ${\rm supp}(f)$ is an affine linear subspace over $\mathbb{Z}_2$,
%we say $f$ has affine support.
%The affine support of a (not identically zero)
 %function $f$ of arity $n$ has
%dimension $r$, for some $0 \le r \le n$.
%Without loss of generality suppose $x_1, \ldots, x_r$ are
%free variables, then  ${\rm supp}(f)$ is described by
%the  data $(A, b)$, where
%$A \in \{0, 1\}^{n \times r}$ is a 0-1 integer matrix
%of which the top $r \times r$ matrix is $I_r$,
%and $b \in \{0, 1\}^{n}$ is a 0-1 integer  vector
%of which the  first $r$ entries are all 0.
%The $i$-th 0-1 variable $x_i$ ($1 \le i \le n$)
%is expressed as  $x_i \equiv \sum_{j=1}^r a_{ij} x_j  +  b_i  \bmod 2$.

%\begin{definition}[Local Affine]
%A function $f$ of arity $n$
% is in $\mathscr{L}$, if and only if for each $\sigma=s_1 s_2 \cdots s_n
%\in {\rm supp}(f)$, the transformed function
% ${M}_{\sigma} f: (x_1, \ldots, x_n) \mapsto
%\alpha^{\sum_{i=1}^n s_i x_i} f(x_1, \ldots, x_n)$
%%= (M_{\alpha^{s_1}} \otimes M_{\alpha^{s_2}} \otimes \cdots \otimes M_{\alpha^{s_n}} ) f$
%is in $\mathscr{A}$.
%Here each $s_i$ is a 0-1 valued integer.
%\end{definition}

\subsection{Characterization}
From the definition of local affine function, it is not easy to check if a given function is in this family or not. It is not even clear
if there exists any interesting new function in this or not. In this subsection, we give an explicit characterization of this family. First of all, if $f^2 \not \in \mathscr{A}$, then $f \not \in \mathscr{L}$. So, we only need to characterize the functions with $f^2 \not \in \mathscr{A}$, or equivalently functions which are already in the form of   (\ref{generic-form-f-sq-in-A-detailed}).

\begin{theorem} \label{thm: Local affine}
A function $f$ defined in (\ref{generic-form-f-sq-in-A-detailed})
belongs to $\mathscr{L}$ iff $H$ is homogeneous of degree 3
and the following set of equations  hold
over $\mathbb{Z}_2$ relating the data $(A,b)$ for  ${\rm supp}(f)$
and the coefficients $c_S$,
\begin{equation}\label{eqn-homogeneous-combined}
 \sum_{i=1}^{n} \prod_{j \in S} a_{ij}  \equiv 0~~~~
(\forall S \subseteq [r], ~\mbox{such that}~ 1 \le |S| \le 4)
\end{equation}
and
\begin{equation}\label{eqn-inhomogeneous-combined}
\sum_{i=1}^{n} \prod_{j \in S} a_{ij}  b_i \equiv c_{S} ~~~~
(\forall S \subseteq [r], ~\mbox{such that}~ 1 \le |S| \le 3)
\end{equation}
\end{theorem}

%\begin{eqnarray}
%& & \sum_{i=1}^{n} a_{ij} b_i \equiv c_j,
%~~~(\forall 1 \le j\le r) \label{eqn-inhomogeneous-deg-1}\\
%& & \sum_{i=1}^{n} a_{ij}  a_{ik}  b_i \equiv c_{j k},
%~~~(\forall 1 \le j < k \le r) \label{eqn-inhomogeneous-deg-2} \\
%& & \sum_{i=1}^{n} a_{ij}  a_{ik}  a_{i\ell} b_i \equiv c_{j k \ell},
%~~~(\forall 1 \le j < k < \ell \le r) \label{eqn-inhomogeneous-deg-3}
%\end{eqnarray}

%%%%%%% the following part can be omitted for space

Note that (\ref{eqn-homogeneous-combined}) actually encodes equivalently
four sets of equations mod 2:
% Since each $a_{ij}$  is a 0-1 integer,
We may state (\ref{eqn-homogeneous-combined}) as
\begin{eqnarray}
& & \sum_{i=1}^{n} a_{ij} \equiv 0,
~~~~ (\forall 1 \le j\le r) \label{eqn-homogeneous-deg-1}\\
& & \sum_{i=1}^{n} a_{ij}  a_{ik}  \equiv 0,
~~~~ (\forall 1 \le j < k \le r) \label{eqn-homogeneous-deg-2}\\
& & \sum_{i=1}^{n} a_{ij}  a_{ik}  a_{i\ell}  \equiv 0,
~~~~ (\forall 1 \le j < k < \ell \le r) \label{eqn-homogeneous-deg-3}\\
& & \sum_{i=1}^{n} a_{ij}  a_{ik}  a_{i\ell} a_{im} \equiv 0,
~~~~ (\forall 1 \le j < k < \ell < m \le r) \label{eqn-homogeneous-deg-4}
\end{eqnarray}
Also  (\ref{eqn-homogeneous-combined}) is equivalent to
\[ \sum_{i=1}^{n} a_{ij}  a_{ik}  a_{i\ell} a_{im} \equiv 0,
~~~ (\forall 1 \le j \le k \le \ell \le m \le r)\]
since  $a_{ij}$ are 0-1 integers.

Similarly (\ref{eqn-inhomogeneous-combined}) encodes equivalently
three sets of equations mod 2:
\begin{eqnarray}
& & \sum_{i=1}^{n} a_{ij} b_i \equiv c_j,
~~~~(\forall 1 \le j\le r) \label{eqn-inhomogeneous-deg-1}\\
& & \sum_{i=1}^{n} a_{ij}  a_{ik}  b_i \equiv c_{j k},
~~~~(\forall 1 \le j < k \le r) \label{eqn-inhomogeneous-deg-2} \\
& & \sum_{i=1}^{n} a_{ij}  a_{ik}  a_{i\ell} b_i \equiv c_{j k \ell},
~~~~(\forall 1 \le j < k < \ell \le r) \label{eqn-inhomogeneous-deg-3}
\end{eqnarray}
Also (\ref{eqn-inhomogeneous-combined})  is equivalent to
\[ \sum_{i=1}^{n} a_{ij}  a_{ik}  a_{i\ell} b_i \equiv c_{j k \ell}
~~~~(\forall 1 \le j \le k \le \ell  \le r).\]

\begin{proof}
For any integer $z$,
we have $z \equiv 0 \bmod 2$ (respectively $1 \bmod 2$)
iff $z^2 \equiv 0 \bmod 4$ (respectively $1 \bmod 4$),
and also iff $z^4 \equiv 0 \bmod 8$ (respectively $1 \bmod 8$).
We will substitute the dependent $s_i$ ($r < i \le n$) in terms of
$s_j$ ($1 \le j \le r$),
\[s_i \equiv \sum_{j=1}^r a_{ij}s_j + b_i \bmod 2,\]
 and similarly for $x_i$ ($r < i \le n$) in terms of $x_j$
($1 \le j \le r$). But the dependent expressions must be valid modulo 8,
 since these appear on the exponent of $\alpha$.
Hence we get
\[
 f \cdot \alpha^{\sum_{i=1}^n \left[ (\sum_{j=1}^r a_{ij}s_j + b_i)^4
(\sum_{j=1}^r a_{ij}x_j + b_i)^4 \right]} \in  \mathscr{A},\]
as a function in $x_i$, valid for any
$(s_1, \ldots, s_r) \in \{0,1\}^r$.

The first simple observation is that the modifier expression
has terms of degree at most 4 in $x_j$'s on the exponent of $\alpha$, and
thus cannot cancel any term of degree greater than 4 in $H$,
which is
the higher order terms in (\ref{generic-form-f-sq-in-A-detailed}).
Moreover, any degree 4 (multilinear) term in $(\sum_{j=1}^r a_{ij}x_j + b_i)^4$
has the form $x_j x_k x_{\ell} x_m$ for some $1 \le j < k < \ell < m \le r$,
and each such term comes with a coefficient divisible by $4! \equiv 0 \bmod 8$.
Thus to get a function in $\mathscr{A}$,
there can be no terms of degree 4 or higher in $H$. Thus
\[H(x) = \sum_{1 \le j < k < \ell \le r} c_{jk\ell}  x_j x_k x_{\ell}.\]

We consider the  condition of  membership in $\mathscr{A}$
 for the linear terms.
The condition is that the function be expressible
as a linear function on the exponent of ${\frak i} = \sqrt{-1}$.
By the uniqueness of expression  of the  (multilinear) polynomial
on the exponent of $\alpha$, this condition is simply
that all coeffiients of linear terms be even.
Thus we can derive necessary conditions in $\mathbb{Z}_2$.
An advatange in working over $\mathbb{Z}_2$, is that we can avoid
the 4-th power expression.

If we set $(s_1, \ldots, s_r) = (0, \ldots, 0) \in \{0, 1\}^r$,
 the all zero string of length $r$, then
 \[f \cdot \alpha^{\sum_{i=1}^n \left[ b_i
(\sum_{j=1}^r a_{ij}x_j + b_i)^4 \right]} \in  \mathscr{A}.\]
Here we used the fact that $b_i^4 = b_i$
for 0-1 valued $b_i \in \mathbb{Z}$.
Computing mod 2,  for the linear terms a necessary condition is that
\begin{equation}\label{c=sum-ab}
c_j \equiv \sum_{i=1}^n a_{ij} b_i \bmod 2,
\end{equation}
for all $j \in [r]$. This is (\ref{eqn-inhomogeneous-deg-1}).
Here we used the fact that $(\sum_{j=1}^r a_{ij}x_j + b_i)^4
\equiv \sum_{j=1}^r a_{ij}x_j + b_i \bmod 2$.

We can also choose $(s_1, \ldots, s_r)$ so that a single $s_{j_0}=1$
and the other  $s_j$'s
are all zero, then
\[f  \cdot \alpha^{\sum_{i=1}^n \left[ (a_{i j_0} + b_i)^4
(\sum_{j=1}^r a_{ij}x_j + b_i)^4 \right]} \in  \mathscr{A}.\]
Again deriving a necessary condition by working over $\mathbb{Z}_2$,
we get
\begin{equation}\label{c=sum-ab-plus}
 c_j \equiv \sum_{i=1}^n (a_{i j_0} + b_i) a_{ij}  \bmod 2.
\end{equation}
Subtracting (\ref{c=sum-ab}) from (\ref{c=sum-ab-plus})
 we obtain both (\ref{eqn-homogeneous-deg-1})
 and (\ref{eqn-homogeneous-deg-2}).

Now we consider quadratic terms.
For this purpose we only need to ensure that
the coefficients (on the exponent of $\alpha$)
 of all quadratic monomials $x_j x_k$ ($1 \le j<k \le r$)
are 0 mod 4. To compute mod 4, we may use $z^2 \bmod 4$
replacing $z^4 \bmod 8$ for any integer $z$. Thus
a necessary condition is that, for all $1 \le j<k \le r$,
\[2 c_{jk} + \sum_{i=1}^n \left[ b_i^2 (2 a_{ij}a_{ik}) \right]
\equiv 0 \bmod 4\]
and
furthermore, for all $1 \le j_0 \le r$,
\[2 c_{jk} + \sum_{i=1}^n \left[ (a_{i j_0} + b_i)^2 (2 a_{ij}a_{ik}) \right]
\equiv 0 \bmod 4\]
Subtracting the two we get
(\ref{eqn-inhomogeneous-deg-2}) and (\ref{eqn-homogeneous-deg-3}).

Finally we consider the coefficients of cubic terms in
$(\sum_{j=1}^r a_{ij} x_j + b_i)^4$. We get,  for all $1 \le j<k<\ell \le r$,
\[4 c_{jk \ell} + \sum_{i=1}^n  \left[  b_i 4 a_{ij} a_{ik} a_{i \ell}
\right]
\equiv 0 \bmod 8\]
%%% that coeff 4 is tricky: there is 4! a_{ij} a_{ik} a_{i \ell} b_i
% from x_jx_kx_l*1. but that 4! becomes 0 mod 8.
% then there is {4 \choose 2} * 3 of a_{ij} a_{ik} a_{i \ell} * 2
% choose 2 out of 4 having the same x_s. gets x_s^2 = x_s. 3 choices for the
% repeated variable x_s. the other two gets 2 times in (a_k +a_ell)^2
This gives us (\ref{eqn-inhomogeneous-deg-3})
\[c_{jk \ell}  \equiv \sum_{i=1}^n  a_{ij} a_{ik} a_{i \ell} b_i
 \bmod  2.\]
Picking exactly one $s_{j_0} =1$ and all other $s_j=0$ (for $1 \le j \le r$)
 we get
furthermore (for all $1 \le j_0 \le r$)
\[4 c_{jk \ell} + \sum_{i=1}^n  \left[ (a_{ij_0} + b_i)^4
4 a_{ij} a_{ik} a_{i \ell}
\right]
\equiv 0 \bmod 8\]
i.e., $ c_{jk \ell} \equiv  \sum_{i=1}^n  \left[ (a_{ij_0} + b_i)
 a_{ij} a_{ik} a_{i \ell}
\right] \bmod  2$.
Subtracting (\ref{eqn-inhomogeneous-deg-3})
 from that we get (\ref{eqn-homogeneous-deg-4}).

Now we prove sufficiency.

By retracing the proof above we have the following
\begin{equation}\label{eqn-for-sufficiency}
f  \cdot \alpha^{\sum_{i=1}^n \left[ (\sum_{j \in S} a_{ij} + b_i)^4
(\sum_{j =1}^r  a_{ij} x_j + b_i)^4 \right]} \in  \mathscr{A},
\end{equation}
for all $S \subseteq [r]$ with candinality $|S| \le 1$.
We prove (\ref{eqn-for-sufficiency})  for all $S \subseteq [r]$
by induction on $|S|$.  Denote by $I_S = \sum_{j \in S}  a_{ij} + b_i$.

Suppose (\ref{eqn-for-sufficiency}) is true for some $S \subset [r]$
and let $j_0 \in [r] \setminus S$, we prove  (\ref{eqn-for-sufficiency})
for $S \cup \{j_0\}$.  We only need to prove that
\[\alpha^{\sum_{i=1}^n \left[ ( (a_{ij_0} + I_S)^4 - I_S^4 )
(\sum_{j =1}^r  a_{ij} x_j + b_i)^4 \right]} \in  \mathscr{A}.
\]
Note that
\[ (a_{ij_0} + I_S)^4 - I_S^4
= a_{ij_0} + 4 a_{ij_0} I_S (I_S^2 + 1) + 6 a_{ij_0} I_S^2
\equiv a_{ij_0} -2 a_{ij_0} I_S^2 \bmod 8.\]
For $S= \emptyset$ and $S = \{j_0\}$, we have
\[f  \cdot \alpha^{\sum_{i=1}^n  \left[ b_i (\sum_{j =1}^r  a_{ij} x_j + b_i)^4 \right] }
 \in  \mathscr{A}
~~~\mbox{and}~~~
f  \cdot \alpha^{\sum_{i=1}^n  \left[ (a_{ij_0} + b_i)^4
 (\sum_{j =1}^r  a_{ij} x_j + b_i)^4 \right]}
 \in  \mathscr{A}.
\]
For 0-1 valued integers $z$ and $z'$, we have
$(z+z')^4 \equiv z + z' -2 z z'  \bmod 8$, so we have
\[\alpha^{\sum_{i=1}^n  \left[ (a_{ij_0} -2 a_{ij_0} b_i)
(\sum_{j =1}^r  a_{ij} x_j + b_i)^4 \right]}
 \in  \mathscr{A}.
\]
Therefore we only need to prove that
\[\alpha^{\sum_{i=1}^n
\left[ \left( (a_{ij_0} -2 a_{ij_0} I_S^2) - (a_{ij_0} -2 a_{ij_0} b_i) \right)
(\sum_{j =1}^r  a_{ij} x_j + b_i)^4 \right]}
 \in  \mathscr{A}.
\]
i.e.,
\[{\mathfrak i}^{\sum_{i=1}^n
\left[ a_{ij_0} (b_i - I_S^2) (\sum_{j =1}^r  a_{ij} x_j + b_i)^4 \right]}
 \in  \mathscr{A}.
\]
But now the expression is on the exponent of ${\mathfrak i}$
and so we can calculate mod 4, which allows us to replace
$(\sum_{j =1}^r  a_{ij} x_j + b_i)^4$ by
$(\sum_{j =1}^r  a_{ij} x_j + b_i)^2$.
However ${\mathfrak i}$ raised to any sum of perfect squares of linear functions
of $x_1, \ldots, x_r$ is in $\mathscr{A}$. This completes the proof.
\end{proof}

\section{Complexity dichotomy theorem of \#{\rm CSP}$_2^c$}

\begin{theorem} \label{thm csp2 theorem}
A $\#\text{\rm CSP}_2^c(\mathscr{F})$ problem has polynomial time algorithm, if one of the following holds,
\[
 \mathscr{F} \subseteq \mathscr{P};~~~~~
 \mathscr{F} \subseteq \mathscr{A};~~~~~
 \mathscr{F} \subseteq \mathscr{A}^\alpha;~~~~~
~\mbox{or}~~~~~~
\mathscr{F} \subseteq \mathscr{L}.\]
%\begin{itemize}
%\item $\mathscr{F} \subseteq \mathscr{P}$,
%\item $\mathscr{F} \subseteq \mathscr{A}$,
%\item $\mathscr{F} \subseteq \mathscr{A}^\alpha$,
%\item $\mathscr{F} \subseteq \mathscr{L}$.
%\end{itemize}
Otherwise, it is \#P-hard.
\end{theorem}

The algorithm for $\mathscr{P}, \mathscr{A}, \mathscr{A}^\alpha$ are known and the algorithm for  $\mathscr{L}$
is in Section \ref{sec: L algorithm} Corollary \ref{corollary  csp2c L}.
In this section, we prove the  \#P-hardness part of this theorem,
we want to show that
if  $\mathscr{F} \not \subseteq \mathscr{P}$, $\mathscr{F} \not \subseteq \mathscr{A}$, $\mathscr{F} \not \subseteq \mathscr{A}^\alpha$  and $\mathscr{F} \not \subseteq \mathscr{L}$, then  $\#\text{CSP}_2^c(\mathscr{F})$
is \#P-hard. We have one function from the complement of
each tractable class, and we prove that, when putting these four (not
necessarily distinct) constraint functions together they define
a \#P-hard problem.
Starting from these functions, we manage to obtain other
functions outside of the respective tractable classes, but with
some specific properties.

Finally after we have gained a  sufficiently good
control on these functions we can corner the beast.

This complexity dichotomy theorem about $\#\text{CSP}_2^c$ generalizes the known complexity dichotomy theorem about $\#\text{CSP}$ (Theorem \ref{thm old csp dichotomy}), and its proof uses this known theorem in several places.

%\begin{theorem} \label{thm old csp dichotomy}
%A $\#\text{CSP}(\mathscr{F})$ problem has polynomial time algorithm, if one of the following holds,
%\[\mathscr{F} \subseteq \mathscr{P} ~~~~~\mbox{or}~~~~~\mathscr{F} \subseteq \mathscr{A}.\]
%%\begin{itemize}
%%\item $\mathscr{F} \subseteq \mathscr{P}$,
%%\item $\mathscr{F} \subseteq \mathscr{A}$.
%%\end{itemize}
%Otherwise, it is \#P-hard.
%\end{theorem}

\subsection{Notations}

In this subsection, we further introduce a number of definitions and notations, which shall be used in the proof.

\begin{definition}[Bundle and bundle type]
Suppose $f$ has affine support of rank $r$ with 
$\{x_1, \ldots, x_r\}$ as  a set of free variables. We use all
non-empty linear combinations $\sum_{j=1}^r{d_j x_j}$ ($d_j \in \mathbb{Z}_2$,
not all zero) of  $x_1, \ldots, x_r$  as the names of bundles of $f$.
The type of each bundle is a possibly empty multiset of
``$+$"'s and ``$-$"'s, and is  defined as follows:
For every input variable $x_k$ ($1 \le k \le n$) of $f$ there is a unique
bundle named $\sum_{j=1}^r{d_j x_j}$ such that
on ${\rm supp}(f)$, $x_k$ is either always equal to $\sum_{j=1}^r{d_j x_j}$
or always equal to   $\sum_{j=1}^r{d_j x_j}+1 \pmod 2$.
In the former case we add a ``$+$",
and in the latter case we add a  ``$-$" to the bundle type
for the bundle named $\sum_{j=1}^r{d_j x_j}$, and we say
the variable $x_k$ belongs to this bundle.

All input variables are partitioned into bundles.
The number $R$ of non-empty bundles is called the essential arity of $f$,
and $r \le R \le 2^r -1$.
\end{definition}

We can list a function's input variables, by listing all its non-empty bundles followed by the bundle type. For example, $f(x_1({\rm ++}),x_2(+),(x_1+x_2)(--))$ has rank 2, essential arity 3, and  arity 5.

\begin{definition}[Odd and even bundle, consistent and opposite bundle]
If the cardinality of a bundle  is odd (resp. even), we say it is an odd (resp. even) bundle.
For an even bundle, if there are even (resp. odd) many $``+"$  in its type, we say it is a consistent (resp. opposite) bundle. Obviously, a consistent (resp. opposite) bundle also has even (resp. odd) many $``-"$, since it is
an even bundle.
\end{definition}

An empty bundle is a consistent even bundle. Equivalently, if a bundle
is odd or opposite, then it is not empty.  %We
%usually do not list an empty bundle as a bundle in function nations.
%By the presence or absence of a particular bundle name in the
%notation for a function we can indicate
%whether a bundle exists or not.
% The non-empty bundle and empty bundle notation help us to say whether some bundles exist or not.
%%%JYC
%%% the following para is astill vague.  but i think people will understand...
% so i did only minimal edits
When constructing some function by a gadget, the bundles of
the function are usually the union
of some original bundles, after some possible flipping,
where a flipping changes all $``+"$ in a type to $``-"$, and changes $``-"$ to $``+"$ at the same time.
%%% JYC: so, doesn't flipping change + and -'s? then how to argue about
% maintain consistent, opposite etc?
% XMJ: it is a integral flipping, keeping consistence and opposite.
If we merge two bundle types $\alpha$ and $\beta$,
we get the union of two types $\alpha \cup \beta$.
Obviously, ${\rm even} \cup {\rm even}= {\rm even}$, ${\rm odd} \cup {\rm even}= {\rm odd}$, ${\rm odd} \cup {\rm odd}= {\rm even}$.  Similarly, ${\rm consistent} \cup {\rm consistent}= {\rm consistent}$, ${\rm consistent} \cup {\rm opposite}= {\rm opposite}$, ${\rm opposite} \cup {\rm opposite}= {\rm consistent}$.

\begin{definition}[Essential function]
Given a function $f$ with affine support, if we replace each (non-empty) bundle of variables by just one variable as the bundle name, keeping the compressed function unchanged, we get the essential function $\tilde{f}$ of $f$.
\end{definition}

For example, the essential function of $f(x_1({\rm ++}),x_2(+),(x_1+x_2)(--))$,
$\tilde{f}(x_1,x_2,(x_1+x_2))$ has arity 3, which is
 the same as the essential arity of $f$. Note that in this example,
the two variables that are both equal to $x_1+x_2+1$ on ${\rm supp}(f)$ have
been replaced by one variable which equals to $x_1+x_2$ on
${\rm supp}(f)$.

If each bundle of a function $f$ has the same type $\alpha$, for example, $f(x_1(\alpha),x_2(\alpha),(x_1+x_2)(\alpha))$,  we also denote it as $\tilde{f}(x_1,x_2,(x_1+x_2))(\alpha)$ through its essential function.
For example, $\tilde{f}({\rm ++})$ denotes a function all whose bundles are $({\rm ++})$.
Sometimes, $f$'s bundles have different types, we use $\tilde{f}(*)$ to denote $f$.
If all bundles of a function are consistent,
we say it is a function of the form $g(cc)$.
%%% JYC in the above notation I avoid talking about \tilde{f}
% we use $\tilde{f}(cc)$ to denote it.
%\footnote{We are in a dilemma of the notations. Consider $f(x_1({\rm ++}),x_2({\rm ++}),(x_1+x_2)(--))$.
%It can be expressed as $\tilde{f}(cc)$, where $\tilde{f}$ acts on $(x_1,x_2, x_1+x_2)$. We also like to express this $f$ as $\hat{f}({\rm ++})$, where $\hat{f}(x_1,x_2, y=x_1+x_2+1)=\tilde{f}(x_1,x_2,y-1)$.

%However, we define the essential function as $\tilde{f}$ but not $\hat{f}$ for other conveniences.
%Sometimes, we abuse $g({\rm ++})$ notations, where $g$ is not the essential function $\tilde{f}$ of $f=g({\rm ++})$, but $g$ is $\hat{f}$ some variable flipped $\tilde{f}$.
%}
%%%%%%%%%%JYC: ok this commented out footnote explains a lot.

%%%JYC: by triple operation, is it done to a single varialbe in a bundle
% or to all var in a bundle?

% XMJ: changed "one" to "a single"
We define two type operations. The first operation is called
triple. In a type multiset, triple can replace
a single $+$ by ${\rm ++}+$, or replace a single $-$ by $---$.
%%% JYC collision is a bad name. maybe collation?   Done.
The second  operation is called collation. Collation can remove ${\rm ++}$ or $--$ from one bundle type, as long as the bundle is still non-empty after the removal.
The type operations do not change the essential function. They do
 not change the properties
 that whether a type is empty, odd, even, opposite, consistent.

Suppose we can use $f$ and $(=_4)$ to construct gadgets.
 When connecting one input variable of $f$ to $(=_4)$,
 we get 3 copies of this variable as additional
 inputs of the new function. This implements
the triple operation.
We can also connect two ${\rm ++}$ or $--$ in the same bundle to make
 them disappear (provided the bundle is still non-empty after the removal).
% ; as long as the bundle is still non-empty after the removal,
This implements the collation operation.
In proving \#P-hardness (but not when designing algorithms),
 we can always do these operations on types.
In this sense, an odd bundle is either $(+)$ or $(-)$,
 and a consistent bundle is either empty or $({\rm ++})$, or $(--)$,
and an opposite bundle is $(+-)$.

$F_{2^r-1}$ denotes the set of functions of rank $r$ whose number of bundles achieves the maximum $2^r-1$ and each bundle has type $(+)$.
We use a super script $\mathscr{A}$ or $\alpha$ to indicate the function or the function set is contained in  $\mathscr{A}$ or $ \mathscr{A}^\alpha$.

Suppose for each function in $F$, all bundles are $(+)$,
 we define $F({\rm ++})=\{f({\rm ++})| f \in F\}$,
 and define $F(+-)$ similarly. Define $F(*)$ to be
the set of functions $f(*)$ with $f \in F$ with any bundle structure.
%%% Does it allow to have an empty bundle in a function in $F(*)$ ?
 Define $F(+)$ to be just $F$ where every function still has all
bundles of type $(+)$. 
%%% JYC by that, you mean every bundle name is nonempty and has (+), right?
For example,  $F_7^\alpha(+-)$ is the set of functions with rank 3 and 7 bundles each with type $(+-)$ whose essential functions are in $ \mathscr{A}^\alpha$.

\begin{definition}
Given a function $f(x_1,x_2,\ldots,x_n)$, we define an arity $2n$ function $f{\rm ++}$, such that $(x_1,y_1,x_2,y_2, \ldots, x_n, y_n) \in {\rm supp}(f{\rm ++})$ iff $x_j=y_j$, $j=1,2,\ldots,n $ and $(x_1,x_2, \ldots, x_n) \in {\rm supp}(f)$, and $f{\rm ++}(x_1,y_1,x_2,y_2, \ldots, x_n, y_n) = f(x_1,x_2, \ldots, x_n) $.

Define $\overline{\mathscr{A}}{\rm ++}=\{f{\rm ++} | f   \in \overline{\mathscr{A}}\}$, where $\overline{}$ denotes complement.
Define $\overline{\mathscr{P}}{\rm ++}=\{f{\rm ++} | f  \in \overline{\mathscr{P}}\}$.
\end{definition}

Because the
binary {\sc Disequality} function is in $\mathscr{A}$, 
and  $\mathscr{A}$ is closed under gadget constructions, 
it can be used to flip input variables, and  $\mathscr{A}$, it is not hard to see if $f \not \in \mathscr{A}$, $f(cc) \in \overline{\mathscr{A}}{\rm ++}$.

Given a function $f$, we  define
the function $f^2$ pointwise by $f^2(x) =(f(x))^2$. Define $\frac{1}{f}$
as the function with the same support as $f$,
but on  ${\rm supp}(f)$,  $\frac{1}{f}(x)=\frac{1}{f(x)}$.

\begin{lemma}\label{lemma:f2++}
 $\#{\rm CSP}_2^c(\{f^2 {\rm ++}\}) \leq_{\rm T}  \#{\rm CSP}_2^c(\{f\})$.
\end{lemma}

\begin{proof}
We take two copies of $f$, 
and connect 2 inputs of a copy of $(=_4)$ to each pair of
 the corresponding variables. The new function is $f^2{\rm ++}$.
\end{proof}

\subsection{Regularization Lemmas}\label{subsec:RegularizationLemmas}
Assume $\mathscr{F} \not \subseteq \mathscr{P}$, $\mathscr{F} \not \subseteq \mathscr{A}$, $\mathscr{F} \not \subseteq  \mathscr{A}^\alpha$ and $\mathscr{F} \not \subseteq \mathscr{L}$,
the high level idea to prove that
$\#\text{CSP}_2^c(\mathscr{F})$ is  \#P-hard is as follows.
 We take one function outside each tractable family and prove
that putting these four (not necessary distinct) functions
together makes a \#P-hard problem.
 If we only have four generic functions, it is difficult to prove anything.
So we wish to regularize and simplify these functions, while maintaining
the property that new functions are still outside the respective
tractable families.
By forming loops and by pinning individual variables
we can reduce the arity, or more precisely, the essential arity
 of the functions. However in fact the more important parameter
that we will try to reduce inductively is the rank of the function.
The hope is that when the number of (free) variables is small,
the  functions are sufficiently easy to handle, as they sit in a space of
smaller dimension.
% This has been a successful strategy in dealing with symmetric
%%% JYC it's nonzense (that I said) that this is true for symmetric...
% when it is symmetric there is no such thinga s bundles etc.
%% what is true for symmetric is they are sufficiently simple if arity is small.
This is true for symmetric
constraint functions. However, in the asymmetric setting,
they are still too complicated even for functions of small rank.
 In this section, we prove some useful
regularization lemmas, that allow us to further
regularize the functions at hand.

We treat the following generic situation. In all the lemmas
in Subsection~\ref{subsec:RegularizationLemmas} we assume there is
a constraint function set $\mathscr{F}$ that satisfies the following
conditions:
\begin{quote}
(1) Any function in $\mathscr{F}$ has affine support;\\
(2) $\mathscr{F}$ contains the pinning functions $\Delta_0$ and $\Delta_1$;\\
(3) $\mathscr{F}$ contains all {\sc Equalities} of even arities and $(=_2) (+-)$ (note that this function is the same as $[1,1]({\rm ++}--)$); \\
(4) $\mathscr{F}$ is closed under gadget constructions, i.e.,
the signature of any $\mathscr{F}$-gate is in $\mathscr{F}$;  \\
(5) $\mathscr{F}$ is closed under reciprocal, i.e., if $f \in \mathscr{F}$,
then $\frac{1}{f}  \in \mathscr{F}$, where $\frac{1}{f}$ is defined above;
%has the same
%support as $f$, but is the reciprocal of $f$ on the support.)  
and \\
(6) For any function $f$ which has affine support, the two bundle type operations do not change whether $f$ is in $\mathscr{F}$ or not. (This is a consequence of (3) and (4).)
\end{quote}

Property (6) is a corollary of (3) and (4).
For example, if a function $f(+--) \not \in \mathscr{F}$, then after applying collation operation on each bundle, we get a $f(+) \not \in \mathscr{F}$.
We prove this by contradiction. Assume $f(+) \in \mathscr{F}$, if we connect $[1,1]({\rm ++}--)  \in \mathscr{F}$ to each bundle, we get $f(+--)$, which is in $\mathscr{F}$ by  (3) and (4).

These properties hold for $\mathscr{A}, \mathscr{A}^\alpha$ and $\mathscr{L}$.
In the statements of Lemma~\ref{lemma r2 useful even odd} 
to \ref{lemma r3 useful odd to +}
  we make the implicit assumption
that $\mathscr{F}$ satisfies these conditions.

% shared by all these four families:

Starting from a constraint  function $f$ outside of a
tractable family,
we will generically try to reduce the rank by pinning at a variable
(and all other variables in the same bundle consistently),
%%% JYC: should we say, "more precisely we pin all variables appropriately
%%% in a bundle
while maintaining the property that the function is still outside of a
tractable family.
Note that pinning at any variable does reduce the rank.
(Every variable can be a member of a set of free variables,
but not every subset of $r_f$ variables  can be a set of free variables.)
We get stuck if pinning any variable (and its bundle)
 of $f$ produces a function
in the tractable family.  The following lemmas turn this
seemingly unfortunate situation into a positive outcome,
by using this to regularize the given function outside of a
tractable family.

Assume we have a function with rank $2$ or $3$ outside the respective
families $\mathscr{A}, \mathscr{A}^\alpha$ or $\mathscr{L}$,
the  following lemmas show
how to  get a function still outside
the respective families but with very regular bundle types.
 We can first regularize so that every bundle has the same parity.
If all bundles are even, we can further regularize
their bundle types to be either all consistent or all opposite.
If all bundles are odd, we can further regularize
 the support space to be a linear space (not just an affine space), which means all bundles are $(+)$.

\begin{lemma}  \label{lemma r2 useful even odd}
Suppose $f \not \in \mathscr{F}$ has rank 2
and pinning any variable of $f$ produces  a rank 1 function in $\mathscr{F}$.
Then we can construct a rank 2 function $g$, such that
$g \not \in \mathscr{F}$, and either all its 3 bundles are odd, or  all its 3 bundles are non-empty even, or it has exactly 2 non-empty even bundles.
\end{lemma}

\begin{proof}

After picking free variables $x_1$ and $x_2$, we have
up to three bundles,  named $x_1$, $x_2$, and $x_1+x_2$.

First suppose  the bundle $x_1+x_2$ is empty.
The  bundles  $x_1$ and $x_2$  are certainly both non-empty.
We will make both bundles  $x_1$ and $x_2$ non-empty and even.
 If the bundle   $x_1$ is odd
(which we may assume  it consists of  a singleton variable $x_1$),
we can pin appropriately on the $x_2$  bundle to get a unary function
$u(y)$ of rank 1 (with a singleton bundle $y$).
Use a
 gadget composed of one copy of $f$, one  copy of  $u$ and
 one  copy of $(=_4)$.
Use two variables of $(=_4)$ to connect $x_1$ of $f$ and $y$ of $u$,
The other two variables of $(=_4)$ are left as two input
variables of the gadget.
We effectively made the singleton bundle  $x_1$ to become
two  equal variables as an even bundle $({\rm ++})$.
%%JYC the following is pedantic
%We effectively merge the $x_1$ bundle of $f$ with the $y$ bundle of $u$.
%Since $(=_4)$ absorbs two variables but also leaves two input ports,
%the new bundle is simply the union of the two original bundles,
% with an extra condition that some two variables from the two original bundles respectively, are set to be equal.

Formally, we construct
 a function $h(*)$, such that $h(z_1,x_2)= \sum_{x_1, y}
f(x_1,x_2) \cdot u(y)
 \cdot (=_4)(x_1,y,z_1,z_1)$. The $z_1$ bundle of $h$ is even,
 and the $x_2$ bundle of $f$ is unchanged.

We claim that $h \not \in \mathscr{F}$.
For a contradiction suppose
 $h \in \mathscr{F}$. We construct a gadget by connecting the variable of $\frac{1}{u}$ to one input variable in the $z_1$ bundle of $h$.
 Because $u$ and $\frac{1}{u}$ are
in $\mathscr{F}$, we have $f(z_1,x_2(*))=
\sum_{z_1}h(z_1({\rm ++}),x_2(*)) \cdot \frac{1}{u}(z_1)
  \in \mathscr{F}$,
where the sum is over one variable of the bundle $z_1({\rm ++})$ in $h$,
equated with the only variable $z_1$ of $\frac{1}{u}$.
 A contradiction. Hence,  $h \not \in \mathscr{F}$.

Similarly, we can change the $x_2$ bundle to an
even bundle, without changing the $x_1$ bundle,
 keeping out of $\mathscr{F}$.
Therefore we can get a function $g \not \in \mathscr{F}$
with exactly 2 non-empty even bundles.

Now suppose the bundle $x_1+x_2$ is not empty. Then it is either odd,
 or it is even but non-empty.
If it is odd, then either all three bundles named $x_1$, $x_2$
and  $x_1+x_2$ are odd, in which case we are done,
or at least one of the bundles named $x_1$ or $x_2$
is even (and non-empty because of free variable status). Without loss of generality
suppose the bundle $x_1$ is even.  Now we pin the $x_2$ bundle appropriately,
%%% appropriately means, sometimes x_2 appears +-, we have to use both pins
% answer: "pin $x_2$" ---> "pin the $x_2$ bundle"
then the bundles of $x_1$ and  $x_1+x_2$ are merged,
creating a unary function $u$ of rank 1 with an odd bundle.
Use this $u$ (and $(=_4)$) we can again change all three bundles of $f$ to
be even, just like before, resulting in a rank 2 function
$g \not \in \mathscr{F}$.

If the bundle $x_1+x_2$ is  even and non-empty,
then either all three bundles named $x_1$, $x_2$
and  $x_1+x_2$ are even, in which case we are done because
the  bundles named $x_1$, $x_2$
are non-empty,
or at least one of the bundles named $x_1$ or $x_2$
is odd. Without loss of generality
suppose the bundle $x_1$ is odd.  Now we pin the $x_2$ bundle appropriately,
then the bundles of $x_1$ and  $x_1+x_2$ are merged,
creating a unary function $u$ of rank 1 with an odd bundle.
The rest of the proof is the same.
%Use this $u$ (and $(=_4)$) we can again change all three bundles of $f$ to
%be even, just like before, resulting in a rank 2 function
%$h \not \in \mathscr{F}$.
\end{proof}

We have several more lemmas in the same vein.
Two of them are still about rank $2$, and the remaining
 three are  about rank $3$.
The construction method in Lemma~\ref{lemma r2 useful even odd}
of using $(=_4)$
to merge two original bundles into a new bundle
and the argument that the new function is not in $\mathscr{F}$,
 is repeatedly used in the following five lemmas.
For simplicity of the statement, we just say which bundle is merged to which, by setting which two variables to be equal. For the rank $3$ case, we often merge 3 pairs of bundles at the same time, so  we need to consider the support structure is not affected.

\begin{lemma}  \label{lemma r2 useful  consistent opposite}
Suppose $f \not \in \mathscr{F}$ has rank 2
and pinning any variable of $f$ produces  a rank 1 function in $\mathscr{F}$.
If each bundle of $f$ is even, then we can construct a rank 2 function $g
\not \in \mathscr{F}$, such that either all its 3 bundles are opposite, or  all its 3 bundles are non-empty consistent, or it has exactly 2
 non-empty bundles which are both consistent.
\end{lemma}

\begin{proof}
The proof is similar to Lemma \ref{lemma r2 useful even odd}. We replace ``odd" by ``opposite", and ``even" by ``consistent".
\end{proof}

\begin{lemma}  \label{lemma r2 useful odd to +}
Suppose $f \not \in \mathscr{F}$ has rank 2
and pinning any variable of $f$ produces  a rank 1 function in $\mathscr{F}$.
If each bundle of $f$ is odd, then we can construct a rank 2 essential arity 3
function  $g(+) \not \in \mathscr{F}$.
\end{lemma}

%\begin{lemma}  \label{lemma r2 useful odd to +}
%Suppose $f \not \in \mathscr{F}$ has rank 2
%and pinning any variable of $f$ produces  a rank 1 function in $\mathscr{F}$.
%If each bundle of $f$ is odd, then we can construct a rank 2 essential arity function  $g(+)$ out of
%$\mathscr{F}$.
%\end{lemma}

\begin{proof}
By being odd, all 3 bundles of  $f$ are non-empty.
By the collation bundle type operation we may
assume  each bundle of $f$ has only one variable,
and so $f$ has arity $3$.
We pick two variables as free variables.
If the dependent bundle of $f$ has type $-$,
i.e., the input variables of $f$ are $x_1,x_2,x_1+x_2+1$,
then we pin $x_1$ to $0$ to get a rank 1 function $u(+-)$.
Then we merge the bundle of $u(+-)$ to the $x_1+x_2+1$ bundle of $f$
 by equating  the two $-$ variables, This produces the
desired function
 $g(+) \not \in \mathscr{F}$.
\end{proof}

We go on to the second batch of lemmas about rank $3$ functions.

\begin{lemma}  \label{lemma r3 useful even odd}
Suppose  $f \not \in \mathscr{F}$ has rank 3 and pinning any variable of $f$
produces a rank 2 function in $\mathscr{F}$.
We can construct a rank 3 function $h$, such that $h \not \in
\mathscr{F}$, and either all its 7 bundles are odd, or all its 7 bundles are non-empty even, or it has exactly 3 non-empty even bundles.
\end{lemma}

\begin{proof}
Because the rank of $f$ is 3, there are 3 non-empty independent bundles, named $x_1, x_2, x_3$ respectively.
We give a list, which covers all possibilities and each possibility ends with a gadget realizing a function as required by the conclusion.

\begin{enumerate}
\item All other 4 bundles named $x_1+x_2$, $x_1+x_3$, $x_2+x_3$
and $x_1+x_2+x_3$ are empty.

\begin{enumerate}
\item
All 3 non-empty bundles are even. Then we take $f$ itself.

\item
There is an odd bundle among  $x_1, x_2, x_3$.

We pick one odd bundle
and pin the variables in the other two bundles, and get a rank 1 function $u$. By the condition, $u \in \mathscr{F}$.
For any odd bundle of $f$, we merge $u$'s bundle with this bundle of $f$. Just as what we did in Lemma \ref{lemma r2 useful even odd}.
We get a function with exactly 3 non-empty even bundles.
\end{enumerate}

\item All  7 bundles are non-empty.

\begin{enumerate}
\item  All  7 bundles are even. We take $f$ itself.

\item  All  7 bundles are odd. We take $f$ itself.

\item At least 4 bundles are even, and at least one bundle is odd.

No 4 nonzero vectors of $\mathbb{Z}_2^3$ can be contained in a 
2 dimensional subspace. So there are 3 linearly independent vectors.
Whether $f \in \mathscr{F}$ is independent of the choice of 
free variables for its support.
So among the even bundles, we can pick 3 linearly
independent bundles and name them  $x_1, x_2, x_3$ respectively.
Under this renaming of the variables and bundles,
  $x_1, x_2, x_3$ are even bundles.

\begin{enumerate}
\item Bundle $x_1+x_2$ is even.

No matter which bundles of the rest 3 bundles are odd, 
we can always pin to get a rank 2 function $g$ in $\mathscr{F}$ 
containing 3 non-empty bundles of different parity types.
Indeed, if $x_1+x_3$ is an odd bundle, we can pin $x_1$,
and the bundles of $x_3$, $x_1+x_3$  are merged producing
an odd bundle,  and  the bundles of $x_2$, $x_1+x_2$ are merged producing
an even bundle. Similarly if $x_2+x_3$ is an odd bundle,  we can pin $x_2$.
%% x3 merged with x_2+x_3 --> odd, and x1 and x1+x2 --> even
If both $x_1+x_3$ and  $x_2+x_3$  are even bundles,
then $x_1+x_2+x_3$ is an odd bundle, then we pin $x_3$.
%% x1 , x1+x3 --> even,  x1+x2 , x_1+x_2+x_3 --> odd.
Go on to pin $g$ to get a rank 1 function $u$, which has one odd bundle.
% JYC: because two exisiting bundles are merged. choose the one pinned is
% is the one that does not have the single type in parity.
% ie if there are two odd, one even, then pin the odd.
% can initially choose any two as free, in rank 2, and all 3 non-empty.
Using $u$, we can change all bundles of $f$ into even bundles.

By symmetry, the proof is the same if either  bundles
$x_1+x_3$ or  $x_2+x_3$ is even.

\item Bundle $x_1+x_2+x_3$ is even.

%Using the same strategy as the previous case, we can get a proper function $u$ to change all bundles of $f$ into even bundles, except when all the other 3 bundles $x_1+x_2, x_1+x_3, x_2+x_3$ are odd bundles.

We may assume the 3 bundles $x_1+x_2, x_1+x_3, x_2+x_3$ are odd bundles.
We pin $x_3$ to 0, to get a rank 2 function $g$.
 All 3 bundles $y_1,y_2, y_1+y_2$ of $g$ are odd.
%Utilizing triple and doubling flip operations, we can promise that each bundle of $g$ contains at least one $+$ and one $-$. We do the same thing to $f$'s bundles $x_1+x_2, x_1+x_3, x_2+x_3$.
We will
merge $g$'s bundles  $y_1,y_2, y_1+y_2$
 to $f$'s bundles $x_1+x_2, x_1+x_3, x_2+x_3$  respectively.
Notice that the same linear dependence holds for these
the respective three bundles.
To effect this merging
we make one variable from the bundle $x_1+x_2$ 
 equal to one variable from the bundle $y_1$ utilizing $(=_4)$. 
Then we make  one variable from  the bundle $x_1+x_3$ 
 equal to one variable from  the bundle $y_2$ utilizing another $(=_4)$.
The last pair of bundles are already merged automatically.
 To avoid introducing extra linear restriction on the support of $f$, 
we do not use any superfluous $(=_4)$ to merge this last pair of bundles.
 We get a function $h$ of rank 3 with all 7  bundles being non-empty and even.
\end{enumerate}

\item At least 4 bundles are odd, and at least one bundle is even.

The proof is parallel to the case 2 (c), except that in the last case it 
ends with a function $h$ of rank 3 with all 7  bundles being  odd.
\end{enumerate}

\item There is a non-empty bundle among $x_1+x_2, x_1+x_3, x_2+x_3, x_1+x_2+x_3$.
(This is logically the complement of case  1. We will use case 2
as a special subcase and reduce this case 3 to case 2.)

We can pin a 
%%% JYC what does it mean by "proper" bundle?
%proper 
bundle to get a rank 2 function $g(y_1,y_2,y_1+y_2)(*)$, whose 3 bundles are not empty and it is in $\mathscr{F}$. Similarly, we can merge the $x_1$ and $y_1$ bundles, and merge the $x_2$ and $y_2$ bundles, and then the $x_1+x_2$ and $y_1+y_2$ bundles are merged automatically, to make the $x_1+x_2$ bundle not empty, keeping the function outside of $\mathscr{F}$.
If the bundle $x_1+x_2+x_3$ is empty, we can merge 
the bundle $y_1$ to $x_1$, and the bundle  $y_2$ to $x_2+x_3$,
and then automatically  the bundle  $y_1+y_2$ to $x_1+x_2+x_3$.

We get a rank 3 function outside of $\mathscr{F}$, with 7 non-empty bundles. Then, we go to the proof in case 2.
\end{enumerate}

\end{proof}

\begin{lemma}  \label{lemma r3 useful opposite consistent}
Suppose $f \not \in \mathscr{F}$ has rank 3
and pinning any variable of $f$ produces  a rank 2 function in $\mathscr{F}$.

If each bundle of $f$ is either consistent or opposite,
 then we can construct a rank 3 function $h \not \in 
\mathscr{F}$, such that either all its 7 bundles are opposite, or  all its 7 bundles are non-empty consistent, or it has exactly 3 non-empty 
bundles (with linearly independent names)
  which are consistent bundles.
\end{lemma}

\begin{proof}
The proof is similar to Lemma \ref{lemma r3 useful even odd}. We replace ``odd" by ``opposite", and ``even" by ``consistent".
\end{proof}

\begin{lemma}  \label{lemma r3 useful odd to +}
Suppose $f \not \in \mathscr{F}$ has rank 3
and pinning any variable of $f$ produces  a rank 2 function in $\mathscr{F}$.
If each bundle of $f$ is odd, 
then we can construct a rank 3 function $h(+)  \not \in \mathscr{F}$.
\end{lemma}

\begin{proof}
By being odd, all 7 bundles of  $f$ are non-empty.
Using the collation operation on the bundle types, we can assume all bundles of $f$ are singletons.
We pick 3 independent bundles of $f$  as free variables, 
so they are given type $x_1(+), x_2(+), x_3(+)$.

We define condition ({\bf F}): 
\begin{quote}
There are four bundles which contain a common free variable $x_j$
 in their names and an odd number of them are of the
$(-)$ type.
\end{quote}

Suppose  condition ({\bf F}) holds.
Such  four bundles correspond to a face (subcube) $\{0, 1\}^2$ of the form
$x_j=1$ in  the cube  $\{0, 1\}^3$.
If we pin the other two free variables to $0$, 
these  four bundles are merged
into a single bundle of  4 variables, and
$(+)$ type (respectively, $(-)$ type)
 variables in these four bundles remain $(+)$ type (respectively, $(-)$ type). 
So there is an odd number of $(-)$ type variables among the 4 variables.
After collation, we get a rank 1 function $u(+-)$ in 
$\mathscr{F}$.
Using $u(+-)$, we can change all $(-)$ type
variables of $f$ to $(+)$ type, keeping it  outside of $\mathscr{F}$.
This proves the lemma under condition ({\bf F}).

If the bundle $x_1+x_2+x_3$ has type $(-)$,
consider the three faces (subcubes) $\{0, 1\}^2$ of the form
$x_j=1$ in  the cube  $\{0, 1\}^3$.
If we assign a number $1 \in \mathbb{Z}_2$ for a $(-)$ type at a vertex
of $\{0, 1\}^3$, and $0  \in \mathbb{Z}_2$ for a $(+)$ type,
and let $s_j$ be the sum in $\mathbb{Z}_2$ over the face corresponding to
$x_j=1$ and let $s = \sum_{j=1}^3 s_j$ in $\mathbb{Z}_2$,
then the value 1 at $x_1+x_2+x_3$ contributes 3 times $\bmod 2$ to $s$,
each value at a point of Hamming weight two contributes 2 times $\bmod 2$
(thus 0, regardless of its value), and the value from 
 $x_1(+), x_2(+), x_3(+)$ are all 0. Hence $s \equiv 1 \bmod 2$,
and therefore  some $s_j \equiv 1 \bmod 2$. Thus condition ({\bf F}) holds.
 The lemma has been proved in this  case. 

In the following we can assume the bundle $x_1+x_2+x_3$ has type $(+)$,
and condition ({\bf F}) does not hold.
Then $s_1 \equiv s_1 \equiv s_3 \equiv 0  \bmod 2$
 implies that all three bundles at 
$x_1+x_2, x_1+x_3, x_2+x_3$
are of the same type, all $(+)$ or all $(-)$.
If they are all of type $(+)$ then we are done.
Suppose all three bundles at 
$x_1+x_2, x_1+x_3, x_2+x_3$
 have type $(-)$.
We can pin one free variable to 0, and get a function
of  rank 2 and essential arity 3. 
This function has type $(+-)$ in all three bundles $y_1$, $y_2$, $y_1+y_2$,
and we will denote it as  $g(+-)$.
%%% JYC: merging at two points:  wt 1 with wt 2, (+) with (-) --> (-)
%%% merging wt 3 with wt 2 --> (-)
 Now we merge the 3 bundles $y_1$, $y_2$, $y_1+y_2$
of $g(+-)$ to the
 3 bundles  $x_1+x_2, x_1+x_3, x_2+x_3$ of $f$
 respectively, by equating the $(-)$ variables in each bundle pair 
(utilizing a copy of $(=_4)$ as  by now the standard way).
Notice that the three bundles of $g$
satisfy the same linear dependence as
the bundles $x_1+x_2, x_1+x_3, x_2+x_3$ of $f$.
This changes the types of these three bundles of $f$
 from $(-)$ to $(+--)$, and then we can further change them to $(+)$
 by collation.
\end{proof}

\subsection{$\overline{\mathscr{P}}$}

In this subsection, we assume $\mathscr{F} \not \subseteq \mathscr{P}$.
Then there is a function $f \in \mathscr{F} - \mathscr{P}$.
Utilizing this $f$, we construct some function $h({\rm ++})$ having
the property
that its
 essential function $h \not \in \mathscr{P}$.
Formally we have the following lemma.

\begin{lemma}\label{lemma not product type}
Suppose $\mathscr{F} \not \subseteq \mathscr{P}$. Then we can construct a function $h{\rm ++} \in \overline{\mathscr{P}}{\rm ++}$, such that
\[\#{\rm CSP}_2^c(\{h{\rm ++}\} \cup \mathscr{F}) \leq_{\rm T}  \#{\rm CSP}_2^c(\mathscr{F}).\]
\end{lemma}

Part of the proof of Lemma \ref{lemma not product type} can be stated as the following arity reduction
lemma about $\mathscr{P}$.

Every function $g \in \mathscr{P}$ has a decomposition as a product of
functions over disjoint subsets of variables, where each factor has support
contained in a pair of antipodal points:
There exists a partition $X = \{x_1,\ldots, x_n\} = \bigcup_{j=1}^k X_j$, and functions
$g_j$ on $X_j$ such that
$g(X) = g(X_1, \ldots, X_k) = \prod_{j=1}^k g_j(X_j)$, and
for all $1 \le j \le k$,
${\rm supp}(g_j) \subseteq \{\alpha_j, \bar{\alpha_j}\}$ for some
$\alpha_j \in \{0, 1\}^{|X_j|}$.

\begin{lemma}\label{lemma arity reducing of not p}
If $f \not \in \mathscr{P}$, but $f^{2} \in \mathscr{P}$, then we can pin $f$
 to a rank 2 function $h$
such that its essential arity is 2
and its compressed function $\underline{h}$ is a binary function with
$\underline{h} \not \in \mathscr{P}$, and $\underline{h}^{2} \in \mathscr{P}$. Furthermore,
all 4 values of $\underline{h}$ are nonzero.
\end{lemma}

\begin{proof}
Since ${\rm supp}(f) = {\rm supp}(f^2)$, $f$ has affine support.
%%% JYC:
%%% this sentence is there because only then can we talk about rank for (subfunctions) of f
Let $g = f^{2} \in \mathscr{P}$, then there is a decomposition
$g(X) = \prod_{j=1}^k g_j(X_j)$, where each $g_j$ evaluates to zero except at
 possibly $\alpha_j$ and $\bar{\alpha_j}$.
A consequence of this decomposition is that
${\rm supp}(g)$ is a direct product of affine spaces $S_j$ with the
special property that, each $S_j$ has at most one
free variable, and if there is one free variable in $X_j$
then all variables $X_j$ are in the same bundle. So there are no
bundles that correspond to sums of two or more free variables.
%%% JYC: this statement is also true about f, because the notion of bundle
% etc are defined for the  supp(f) = {\rm supp}(g)

Clearly  $g_j(\alpha_j)$ and $g_j(\bar{\alpha_j})$ cannot
both be 0, for otherwise $g$ is identically 0, and so is $f$.
Then  $f \in \mathscr{P}$, a contradiction.
%Firstly, we handle the situation that some $g_j(\alpha_j)$ or $g_j(\bar{\alpha_j})$ is 0.
Suppose for some $j$, one of  $g_j(\alpha_j)$ or $g_j(\bar{\alpha_j})$ is 0.
Without loss of generality suppose $j=1$,
$g_1(\bar{\alpha_1}) =0$ and $g_1(\alpha_1) \not =0$.
 Then
the function $f'(X_2, \ldots, X_k) = f^{X_1 = \alpha_1}(X_2, \ldots, X_k)$
has the property that $f' \not \in \mathscr{P}$ but $(f')^2 \in \mathscr{P}$.
The latter claim is obvious since $(f')^2 = (f^2)^{X_1 = \alpha_1}
= g_1(\alpha_1) \cdot \prod_{2 \le \ell \le k} g_\ell$.
%shows that $(f')^2 \in \mathscr{P}$.
%is a desired product form. On the other hand,
For the former claim,
if $f' \in \mathscr{P}$, we can define $f_1(X_1) =1$
at $X_1 = \alpha_1$ and 0 otherwise,
then $f(X_1, \ldots, X_k) = f_1(X_1) \cdot f^{X_1 = \alpha_1} = f_1(X_1) \cdot f'$,
because $f$ is zero for all assignments unless $X_1 = \alpha_1$.
This shows that  $f \in \mathscr{P}$, a contradiction. Hence we can
continue the proof inductively on the function $f'$.

Therefore we can assume that
each $g_j$ has support ${\rm supp}(g_j) = \{\alpha_j, \bar{\alpha_j}\}$.

For each $1 \le j \le k$, define $f_j(X_j) = \sqrt{g_j(X_j)}$.
This is a pointwise definition by taking a square root value (of arbitrary sign).
Then $\left( \prod_{j=1}^k f_j \right)^2 = g = f^2$.

Now we define a sign function $S: \{0, 1\}^k \rightarrow \{+1, -1\}$.
For $(y_1, \ldots, y_k) \in \{0, 1\}^k$, let
\begin{equation}\label{eqn:defn-of-S}
S(y_1, \ldots, y_k) = \frac{f(X_1, \ldots, X_k)}{ \prod_{j=1}^k f_j(X_j)}
\end{equation}
where $X_j = \alpha_j$ if $y_j = 1$ and $X_j = \bar{\alpha_j}$ if $y_j = 0$.
Note that since ${\rm supp}(g_j) = \{\alpha_j, \bar{\alpha_j}\}$,
there is no division by zero and $S$ is well-defined.
Then
\begin{equation}\label{eqn:f-is-S-times-prod}
f(X_1, \ldots, X_k) = S(y_1(X_1), \ldots, y_k(X_k)) \cdot  \prod_{j=1}^k f_j(X_j)
\end{equation}
for all $X = (X_1, \ldots, X_k)$, where $y_j(\cdot)$ is
a function which is defined as $y_j(X_j) = 1$ if $X_j = \alpha_j$
and $y_j(X_j) = 0$ otherwise.
Note that in (\ref{eqn:f-is-S-times-prod})
both sides are zero unless for all $1 \le j \le k$,
$X_j = \alpha_j$ or $\bar{\alpha_j}$,
and in that case (\ref{eqn:f-is-S-times-prod})
follows from (\ref{eqn:defn-of-S}).

Because $S$ is a $\pm 1$ valued function, there is a
%(unique)
multilinear
polynomial $p(y_1, \ldots, y_k) \in \mathbb{Z}_2[y_1, \ldots, y_k]$
such that $S(y_1, \ldots, y_k) = (-1)^{p(y_1, \ldots, y_k)}$.
If $\deg(p) \le 1$ then $S$ is factorizable as functions on
each $y_j$ separately, and consequently $f \in \mathscr{P}$ by (\ref{eqn:f-is-S-times-prod}),
a contradiction.
Hence $\deg(p) \ge 2$.

Consider a monomial with minimum degree among all
monomials of degree at least 2. Without loss of generality let
it be $y_1 y_2 \ldots y_\ell$, where $\ell \ge 2$.
Now, for all $\ell < j \le k$,  pin $X_j$ to $\bar{\alpha_j}$,
which corresponds to setting  $y_j = 0$.
Any monomial in $p$ that has a factor $y_j$ for some $j >\ell$
is annihilated.
Any monomial that is a subproduct of $y_1 y_2 \ldots y_\ell$
(including  $y_1 y_2 \ldots y_\ell$ itself) is unaffected.
By the minimality of  $y_1 y_2 \ldots y_\ell$,
all other remaining monomials must have  degree at most 1.
Now, for all $2 < j \le \ell$,
 further pin $X_j$ to $\alpha_j$,
which corresponds to setting  $y_j = 1$,
we reduce $p$ to $c_0 + c_1y_1 + c_2 y_2 + y_1y_2$ for some
$c_0, c_1, c_2 \in \mathbb{Z}_2$.
The compressed function $\underline{h}$
of the corresponding rank 2 function $h$ obtained from $f$ by pinning
has matrix
 $\lambda \left ( \begin{matrix}1 & a\\
b & -ab \end{matrix} \right )$, with nonzero $\lambda, a, b$.
As noted earlier, the bundle named $x_1 + x_2$ is empty; this is
a property of ${\rm supp}(f) = {\rm supp}(g)$. Hence $h$ has
essential arity 2.
Clearly $\underline{h} \not \in \mathscr{P}$,
but $\underline{h}^2 \in \mathscr{P}$.
\end{proof}

\begin{proof}[Proof of Lemma \ref{lemma not product type}]
Starting from any $f \in \mathscr{F} - \mathscr{P}$,
if $f^2 \not \in \mathscr{P}$, then we can just realize $f^2{\rm ++}$ by Lemma \ref{lemma:f2++}.
%In any instance of
%$\#{\rm CSP}_2^c(\{f^2{\rm ++}\} \cup \mathscr{F})$
%any occurence of $f^2{\rm ++}$ can be simply replaced by two copies of
%$f$, producing an intance of $\#{\rm CSP}_2^c(\mathscr{F})$
%with the same value.

Now we assume $f^2 \in \mathscr{P}$.
By Lemma \ref{lemma arity reducing of not p},
we can get a function $h$ of rank 2 and essential arity 2, such that
 its compressed function $\underline{h} \not \in \mathscr{P}$,
and $\underline{h}^{2} \in \mathscr{P}$.

Ignore a nonzero constant
we may assume $\underline{h} = \left ( \begin{matrix}1 & a\\
b & c \end{matrix} \right )$, with nonzero $a,b,c$.
From the pointwise square
function $\underline{h}^{2} = \left ( \begin{matrix}1 & a^2 \\
b^2 & c^2 \end{matrix} \right ) \in \mathscr{P}$, we
get $c^2 = ab$, and thus $c= - ab$ because $\underline{h} \not \in \mathscr{P}$.
%%% c =\pm ab since a b c not 0, only way to be in P is degenerate.
% can't be + otherwise h itself is in P
%Hence, $\tilde{h}p = \left ( \begin{matrix}1 & a\\
%b & -ab \end{matrix} \right )$.
%%% JYC i don't know what is \tilde{h}p
% but in any case is not needed.
% in fact in we claimed a bit more in the last lm, this  -ab is  already proved.
%According to Lemma \ref{lemma arity reducing of not p}, $ab \neq 0$.

There are two cases.
One case is that $a^8 \neq 1$ or $b^8 \neq 1$.
Without loss of generality,
assume $a^8 \neq 1$. We can construct a gadget by taking two copies of
$h$, and connect their respective variables within the bundle $x_2$.
 A variable with a $(+)$ (respectively, a $(-)$) label is
connected to the corresponding variable in the other copy of $h$ with the
same  $(+)$ (respectively, $(-)$) label. This produces a function with
two bundles
(corresponding to the bundles both named $x_1$ in
each copy of $h$). This function is denoted as $g$.
The compressed function of $g$ is
%$\underline{h}\underline{h}^{\tt T}=
$\left ( \begin{matrix}1 & a\\
b & -ab \end{matrix} \right )\left ( \begin{matrix}1 & b\\
a & -ab \end{matrix} \right )=\left ( \begin{matrix}1+a^2 & b(1-a^2)\\
b(1-a^2) & b^2(1+a^2) \end{matrix} \right )$.
By Lemma~\ref{lemma:f2++} we
have $\#{\rm CSP}_2^c(\{g^2 {\rm ++}\} \cup \mathscr{F}) \leq_{\rm T}  \#{\rm CSP}_2^c(\mathscr{F})$.
The compressed function of $g^2 {\rm ++}$ is
$\left ( \begin{matrix}(1+a^2)^2 & b^2(1-a^2)^2\\
b^2(1-a^2)^2 & b^4(1+a^2)^2 \end{matrix} \right )$.
By checking its determinant we conclude that if $a^8 \neq 1$
then this binary function does not belong to $\mathscr{P}$.
%%% JYC det = 8 a^2 (1+a^4) b^4, also no zero entries by  a^2 \not =1, -1.
%We claim if $a^8 \neq 1$ then $(hh')^{2}{\rm ++}  \in
%\overline{\mathscr{P}}{\rm ++}$, since $(\tilde{h} \tilde{h}')^{2}$ is full of rank.
Hence $g^2{\rm ++}  \in
\overline{\mathscr{P}}{\rm ++}$.

The other case is that $a^8=b^8=1$. Suppose $h$ is $h(x(\sigma),y(\tau))$, where $\sigma$  (resp. $\tau$) is the type of bundle $x$ (resp. $y$).
 We pin the input bundle $x$ to get $[1,a](\tau)$
 (note that the bundle $x_1 + x_2$ is empty.)
We pin the input bundle $y$ to get $[1,b](\sigma)$.
Put them together we can get a function $s=s(x(\sigma),y(\tau))$,
with essential and compressed function $\tilde{s}=\underline{s}=[1,b] \otimes [1,a] = \left ( \begin{matrix}1 & a\\ b & ab \end{matrix} \right )$.

We put one copy of $h$ and 7
copies of $s$ together, to  get a function whose inputs are $x_j(\sigma),y_j(\tau)$, $j=1,2,\ldots,8$.
For each element in the type $\sigma$, say $+$, we connect the 8 variables $x_j$, $j=1,2,\ldots,8$, by an {\sc Equality}  $(=_{10})$ of arity 10. The 8 variables become 2 variables, and the eight $x_j$ bundles are merged into a bundle $x(\sigma \cup \sigma)$. We handle $y_j(\tau)$, $j=1,2,\ldots,8$, similarly.
At last we get a function in inputs $(x(\sigma \cup \sigma), y(\tau \cup \tau))$. It can be expressed as $t{\rm ++}$, where neither
$t(x(\sigma), y(\tau))$ nor its compressed function $\left ( \begin{matrix}1 & a^8 \\ b^8 & -a^8b^8 \end{matrix} \right )
=\left ( \begin{matrix}1 & 1 \\ 1 & -1 \end{matrix} \right )$ is in $\mathscr{P}$.
%%% JYC i think you are talking about two things: the gadget is using s(*)
%%% but the analysis is in terms of the compressed signature s
\end{proof}

%\begin{lemma} \label{lemma square must be in A}
%If $\mathscr{F} \not \subseteq \mathscr{P}$, and there is a function $F \in \mathscr{F}$ such that $F^2 \not \in \mathscr{A}$, then $\#{\rm CSP}_2^C(\mathscr{F})$ is \#P-hard.
%\end{lemma}
%
%\begin{proof}
%By Lemma \ref{lemma not product type}, there is a fucntion $H({\rm ++})$ whose base function $H \not \in \mathscr{P}$, such that $\#{\rm CSP}_2^C(\{H({\rm ++})\} \cup \mathscr{F}) \leq_{\rm T}  \#{\rm CSP}_2^C(\mathscr{F})$.
%
%From $F$, we construct $F^2({\rm ++})$.
%\[\#CSP(\{H,F^2\}) \leq_{\rm T} \#CSP_2^C(\{H({\rm ++}), F^2({\rm ++})\} \cup \mathscr{F}) \leq_{\rm T}  \#CSP_2^C(\mathscr{F}).\]
%
%Because $F^2 \not \in \mathscr{A}$ and $H \not \in \mathscr{P}$, by theorem \ref{thm old csp dichotomy}, $\#CSP(\{H,F^2\})$ is \#P-hard.
%\end{proof}

%\begin{lemma} \label{lemma square must be in A}
%If $\mathscr{F} \not \subseteq \mathscr{P}$, and there is a function $f({\rm ++}) \in \mathscr{F}$ such that $f \not \in \mathscr{A}$, then $\#CSP_2^c(\mathscr{F})$ is \#P-hard.
%\end{lemma}
%
%\begin{proof}
%By Lemma \ref{lemma not product type}, there is a fucntion $h({\rm ++})$ whose base function $h \not \in \mathscr{P}$, such that $\#CSP_2^C(\{h({\rm ++})\} \cup \mathscr{F}) \leq_{\rm T}  \#CSP_2^C(\mathscr{F})$.
%
%\[\#CSP(\{h,f\}) \leq_{\rm T} \#CSP_2^C(\{h({\rm ++}), f({\rm ++})\} \cup \mathscr{F}) \leq_{\rm T}  \#CSP_2^C(\mathscr{F}).\]
%
%Because $f \not \in \mathscr{A}$ and $h \not \in \mathscr{P}$, by theorem \ref{thm old csp dichotomy}, $\#CSP(\{h,f\})$ is \#P-hard.
%\end{proof}

\subsection{$\overline{\mathscr{A}}$}
\begin{lemma}\label{lemma not A pin to 321}
Let $f \not \in \mathscr{A}$. In $\#{\rm CSP}_2^c(\{f\})$, we can realize a
 function in one of the following sets:
\begin{itemize}
  \item $\overline{\mathscr{A}}{\rm ++}$,
  \item $F_1^\alpha$, $F_1^\alpha(+-)$,
  \item $F_3^\alpha(+)$, $F_3^\alpha(+-)$,
  \item $F_7^\alpha(+)$, $F_7^\alpha(+-)$.
\end{itemize}
\end{lemma}

\begin{proof}
If $f^2 \not \in \mathscr{A}$, we can realize $f^2{\rm ++}$, which is a function in $\overline{\mathscr{A}}{\rm ++}$.

Now suppose $f^2 \in \mathscr{A}$.
Ignoring a nonzero constant factor, we can write $f$ more explicitly as
\begin{equation}\label{generaic-form-f-sq-in-A-detailed}
f= {\rm supp}(f) \cdot  \alpha^{L(x) + 2 Q(x) + 4 H(x)}
\end{equation}
where  $L(x) = \sum_{j=1}^r c_j x_j$ is a linear function,
$Q(x) = \sum_{1 \le j < k \le r} c_{jk} x_j x_k$
 is a quadratic (multilinear)  polynomial,
and $H(x) = \sum_{1 \le j < k < \ell \le r} c_{jk\ell}  x_j x_k x_{\ell}
+ \cdots$
is a  (multilinear)  polynomial with all monomials  of
degree at least 3, where $r$ is its rank.

By the uniqueness (in the sense that all
coefficients $c_j, c_{jk}$ and $c_{jk\ell}$ are integers mod 2)
 of this polynomial expression in the
exponent of $\alpha$, any $f$ defined by the expression in
(\ref{generaic-form-f-sq-in-A-detailed}) is in $\mathscr{A}$
 iff all the coefficients of the 3 polynomials $L,Q$ and $H$ are even.
Because $f \not \in \mathscr{A}$, there is an odd coefficient.

If one coefficient of $L$ is odd, say $c_1$, we pin all free variables to 0 except $x_1$, to get a rank 1 function $\ell$ not in $\mathscr{A}$. It is not 
hard to see, after a collation if necessary,
 $\ell \in F_1^\alpha \cup F_1^\alpha(+-) \cup F_1^\alpha(cc)$, 
while  $F_1^\alpha(cc) \subseteq \overline{\mathscr{A}}{\rm ++}$.

If one coefficient of $Q$ is odd, say $c_{12}$, we pin all free variables to 0 except $x_1,x_2$, to get a rank 2 function $q_1$ not in $\mathscr{A}$.
If we can pin $q_1$ to get a rank 1 function not in $\mathscr{A}$, we fall into the previous case. Hence, we can assume the conditions of Lemma \ref{lemma r2 useful even odd}, \ref{lemma r2 useful  consistent opposite} and \ref{lemma r2 useful odd to +} are satisfied.
By Lemma \ref{lemma r2 useful even odd}, from $q_1$ we can construct a function $q_2$, all bundles of $q_2$ are even or all bundles of $q_2$ are  odd.

If all bundles of $q_2$ are odd, by Lemma \ref{lemma r2 useful odd to +}, we 
go on to get a function $q(+) \not \in \mathscr{A}$.
If the linear terms $L_{q}$ in the
corresponding polynomial for $q$ of $q(+)$ contains an odd coefficient, we can pin to get a rank 1 function and fall into lower rank case. Hence, we assume all coefficients of $L_q$ are even.
Suppose the compressed function $\underline{q}$ of $q(+)$ is $\alpha^{c_1 x_1 +c_2 x_2 + 2c_{12} x_1x_2}$, Then $c_1 \equiv c_2 \equiv 0$ and $c_{12} \equiv 1$.
If we apply a $M_\alpha$ transformation to $q(+)$, the compressed function becomes
$\alpha ^{x_1+x_2+(x_1+x_2)^3}  \cdot 
\alpha^{c_1 x_1 +c_2 x_2 + 2c_{12} x_1x_2}=
 \alpha^{(c_1+2) x_1 +(c_2+2) x_2 + 2(c_{12}-1) x_1x_2}$, which is a function in $\mathscr{A}$. Note that,
%%% JYC i corrected that term to two times -1, (really 6 = -2)
for $x_i = 0, 1 \in \mathbb{Z}$, the value $x_1 + x_2 \bmod 2$ 
cannot be calculated as a linear term on the exponent of $\alpha$,
but can be calculated mod 8, and thus
we used the expression $(x_1 + x_2)^3$, since
$x_1 + x_2 \equiv 0, 1 \bmod 2$ iff $(x_1 + x_2)^3 \equiv 0, 1 \bmod 8$.
Hence, $q(+)$ belongs to the set $F_3^{\alpha}$.

If all bundles of $q_2$ are even, we apply Lemma \ref{lemma r2 useful  consistent opposite} to make all bundles either consistent or opposite.
If all bundles are opposite, by the same analysis of the compressed function,  we get a $q(+-) \in F_3^{\alpha}(+-)$.
If all bundles are consistent,  by the same analysis of the compressed function, we get a $q(cc)$ of essential arity 2 or 3, with $q \not \in \mathscr{A}$. Hence, $q(cc) \in \overline{\mathscr{A}}{\rm ++}$.

The last case is that there is an odd coefficient in $H$
in (\ref{generaic-form-f-sq-in-A-detailed}).
 Suppose the monomial $M$ has the minimum degree, among all
 monomials in $H$ with  odd coefficient. We pin all free variables which are not in $M$ to 0, and pin the variables in $M$ to 1, except 3 of them, to get a rank 3 function $h_1$.

Similarly, by Lemma \ref{lemma r3 useful even odd}, \ref{lemma r3 useful odd to +} and \ref{lemma r3 useful opposite consistent}, from $h_1$,
either we get functions not in $\mathscr{A}$ of smaller rank and fall into the solved two cases, or we get one of the following rank 3 functions:
 an essential arity 7 function $h(+)$, or an essential arity 7 function $h(+-)$, or an essential arity 3 function $h(cc)$, or  an essential arity 7 function $h(cc)$.
For all cases the analysis of the compressed function
$\underline{h}=\alpha^{c_1 x_1 +c_2 x_2 + c_3 x_3 + 2c_{12} x_1x_2+ 2c_{13}x_1x_3 + 2 c_{23} x_2 x_3 + 4 c_{123}x_1x_2x_3}$ is the same.
(The fact that $\underline{h}$ has such an expression,
namely the coefficients of degree 2 terms are all even and
the coefficient of degree 3 terms is divisible by 4 ultimately follows
from the expression (\ref{generaic-form-f-sq-in-A-detailed}) for $f$.)
 Similar to the above proof, we can assume all coefficients in $\underline{h}$ are even, except that $c_{123}$ is odd.
$\underline{h} \not \in \mathscr{A}$, so in the last two cases, we get a function in $\overline{\mathscr{A}}{\rm ++}$.
For the first two cases, we only need to prove a $M_{\alpha }$ transformation applied to the essential function $h$, will change $h$ to a function in $F_7^{\mathscr{A}}$.
The compressed function of $M_{\alpha }^{\otimes 7} \cdot \tilde{h}$ is
\begin{eqnarray*}
&&\alpha^{x_1+x_2+x_3+(x_1+x_2)^3+(x_1+x_3)^3+(x_2+x_3)^3+(x_1+x_2+x_3)^4} \cdot \underline{h}\\
&=&\alpha^{(c_1+4) x_1 +(c_2+4) x_2 + (c_3+4) x_3 + 2 (c_{12}+2) x_1x_2+ 2(c_{13}+2) x_1x_3 + 2 (c_{23}+2) x_2 x_3 + 4 (c_{123}+1)x_1x_2x_3},
\end{eqnarray*}
which is a function in $\mathscr{A}$, since the exponent
has the form $L' + 2Q' + 4H'$,  and all coefficients of $L', Q', H'$
are  even.
(As in these two case, the essential arity 7 function is either $h(+)$
or $h(+-)$, the holographic transformation is $M_{\alpha }^{\otimes 7}$
on the essential function $\tilde{h}$.)
\end{proof}
\noindent
{\bf Remark:}
We remark that this form for 
the  compressed function of $M_{\alpha }^{\otimes 7} \cdot \tilde{h}$
can be derived as follows: The values of $x_1, x_2, x_3$ are all
0-1 integers, and the transformation produces a factor $\alpha$ 
for each variable iff
the variable takes value 1.
For a variable such as $x_1+x_2$ or $x_1+x_2+x_3$, one has to be careful
to remember that such a linear expression is in the sense of $\mathbb{Z}_2$;
it is illegitimate to simply substitute the  linear expression
on the exponent of $\alpha$, which can only be computed
as an integer mod 8. For an expression such as $x_1+x_2$
the integer value can be only 0, 1 or 2, in which case
$(x_1+x_2)^3 \bmod 8$ keeps the meaning of
the 0-1 value of $x_1+x_2 \bmod 2$.
For $x_1+x_2+x_3$, the value could be
0, 1, 2 or 3, then
we must use the expression
$(x_1+x_2+x_3)^4 \bmod 8$ to keep the
 meaning of
the 0-1 value of $x_1+x_2 + x_3 \bmod 2$.
One can calculate the end result such as the
coefficients of $x_1x_2$ or $x_1x_2x_3$
by noticing that the modifier expression is
symmetric in $x_1, x_2, x_3$
and, e.g., the modifier coefficient for $x_1x_2x_3$
is ${4 \choose 2} 3! = 36 \equiv 4 \bmod 8$.
% {4 \choose 2} is those 2 places where the double like x_1^2 occurs.
% then pick which variable as the double occurence etc.
%

\subsection{$\overline{\mathscr{A}^\alpha}$}

By definition,  a function $f \not \in \mathscr{A}^\alpha$
iff $M_{\alpha^{-1}} ^ {\otimes n} \cdot f \not \in \mathscr{A}$,
 where $n = n(f)$ is the arity of $f$.
 Let $f'=M_{\alpha^{-1}} ^ {\otimes  n} \cdot f$.
To derive the corresponding results for
$\overline{\mathscr{A}^\alpha}$ in Lemma~\ref{lemma not A-alpha  pin to 321}
 we use $f'$ to repeat the proof of Lemma~\ref{lemma not A pin to 321}.
However we should be careful in justifying
the steps in gadget constructions. 

One primitive of gadget construction
 is pinning. Because $M_{\alpha}$ is diagonal,
  $f'^{x_j=\epsilon} \in \mathscr{A}$ iff $f^{x_j=\epsilon} 
\in \mathscr{A}^\alpha$, for $\epsilon =0, 1$.
Thus pinning to $f'$ can be replaced by pinning directly to $f$.

The other primitive of gadget construction
 is merging two bundles, where the basic operation is 
to connect by $(=_4)$
 two inputs (one is from say $f'$ and the other is from some $q'$
 which may be obtained from $f'$ by pinning).
In this gadget, $(=_4)$ is separated by two $M_{\alpha^{-1}}$ 
from touching $f$ and $q$ directly. But if we consider
there is a $(=_2)$ on the edge, then the new function
on the edge is a symmetric binary function with matrix
 $(M_{\alpha^{-1}})^{\tt T} M_{\alpha^{-1}}$ which represents
a function in $\mathscr{A}$.
 Therefore we can directly argue
whether the gadget using transformed function $f'$
results in a function in $\mathscr{A}^\alpha$
iff the same  gadget using the untransformed function $f$
results in a function in $\mathscr{A}$.
%   move the  $M_{\alpha^{-1}}$  to the other two edges 
%of $=_4$ without changing the gadget function.
%In effect we have a new gadget composed of $f$ and $q$,
% which is equal to the original gadget if surrounded by  $M_{\alpha^{-1}}$.  
%The original gadget is in $\mathscr{A}$ 
%iff the new gadget is in $\mathscr{A}^\alpha$.

We conclude that using $f'$  to repeat the proof of Lemma \ref{lemma not A pin to 321}, we get 
%JYC 
%the same ???
some function through some gadgets composed of $f'$, $(=_4)$ 
and pinning functions. 
%JYC I don't see how to adjust the gadget. but the argemtn is adjusted.
% like adding a diag [1 0\\ 0 i] which is the square of [1 0 \\ 0 alpha]
%After some adjustments of these gadgets, 
In the end
we get some equivalent gadgets which are new gadgets composed of 
$f$, $(=_4)$ and pinning functions, which are transfored versions
under  $M_{\alpha^{-1}}$.
%According to Lemma \ref{lemma not A pin to 321}, these new gadgets transformed by  $M_{\alpha^{-1}}$ are not in $\mathscr{A}$. Hence this new gadget are not in $\mathscr{A}^\alpha$.
To get the form of these functions of the new gadgets, 
we just do $M_{\alpha}$ transformations to the outcomes of Lemma \ref{lemma not A pin to 321}.

%%% JYC: actually i am not convinced that the above argument works for
% the first item that the constructed function is not in A.
%%% do you mean the item 1 in lm 4.11 is for \overline{A^alpha} ++ ??

\begin{lemma}\label{lemma not A-alpha  pin to 321}
Let $f \not \in \mathscr{A}^\alpha$. In $\#{\rm CSP}_2^c(\{f\})$, we can realize a function in one of the following sets:
\begin{itemize}
  \item $\overline{\mathscr{A}}{\rm ++}$,
  \item $F_1^\mathscr{A}$, $F_1^\alpha(+-)$,
  \item $F_3^\mathscr{A}(+)$, $F_3^\alpha(+-)$,
  \item $F_7^\mathscr{A}(+)$, $F_7^\alpha(+-)$.
\end{itemize}
\end{lemma}

The the expressions in Lemma~\ref{lemma not A-alpha  pin to 321}
are those expressions in Lemma~\ref{lemma not A pin to 321}
under the transformation by $M_\alpha$. 
For $\overline{\mathscr{A}}{\rm ++}$, the two copies of $M_\alpha$
produce a modification by a factor 1 or $\alpha^2 = i$
 (if the variable for a bundle name is 0 or 1, which
appears twice as equal variables in the bundle by the $(++)$ type).
The unary function $[1,i]$ is in $\mathscr{A}$.
Therefore this modification does not affect the (non)membership 
for its essential function in
$\mathscr{A}$.
As well,
the expressions 
$F_1^\alpha(+-)$,  $F_3^\alpha(+-)$ and $F_7^\alpha(+-)$
from Lemma~\ref{lemma not A pin to 321} are not changed 
to the corrresponding expressions in Lemma~\ref{lemma not A-alpha  pin to 321}
because for the $(+-)$ type the aggregate modification on the two variables
in a bundle is always $\alpha$, which becomes a constant factor.
The expressions 
$F_1^\alpha$,  $F_3^\alpha(+)$ and  $F_7^\alpha(+)$
do get changed to
$F_1^\mathscr{A}$, $F_3^\mathscr{A}(+)$ and $F_7^\mathscr{A}(+)$
respectively.

\subsection{$\overline{\mathscr{L}}$}

\begin{lemma} \label{lemma L 1}
If we have a rank 1 function $f \not \in \mathscr{L}$, in $\#{\rm CSP}_2^c(\{f\})$ we can realize a function in one of the following sets: $\overline{\mathscr{A}}{\rm ++}$, $F_1^\alpha$, $F_1^\mathscr{A}$, $F_1^\mathscr{A}(+-)$.
\end{lemma}

\begin{proof}
If $f^2 \not \in \mathscr{A}$, we can realize $f^2{\rm ++}$, which is a function in $\overline{\mathscr{A}}{\rm ++}$.
Now assume $f^2 \in \mathscr{A}$.

We discuss the cases according to the type of the
unique bundle of $f$ being odd, or consistent, or opposite.

Suppose the unique bundle of $f$ is odd, (by equation (\ref{eqn-homogeneous-deg-1}), $f \not \in \mathscr{L}$ no matter what is the integer coefficient $c_1$), we change this bundle type to a singleton by the collation operation, to get a $[1, \alpha^{c_1}] \in F_1^\alpha \cup F_1^\mathscr{A}$.

Suppose the unique bundle is consistent. By being of rank 1,
the bundle is named for a free variable and thus non-empty.
Because  $f \not \in \mathscr{L}$, and equation (\ref{eqn-homogeneous-deg-1})
is satisfied, it follows that 
 equation (\ref{eqn-inhomogeneous-deg-1}) must have been violated.
Being consistent, the left hand side
of equation (\ref{eqn-inhomogeneous-deg-1})
 is $0 \bmod 2$. 
Hence, $c_1 \equiv 1$. In this case, after some collation operations we get $g=[1, \alpha^{c_1}](cc)$ and $c_1 \equiv 1$.
Because $[1, \alpha^{c_1}] \not \in \mathscr{A}$, $g \in \overline{\mathscr{A}}{\rm ++}$.

Suppose the unique bundle is opposite.
In particular the bundle is even and
equation (\ref{eqn-homogeneous-deg-1})
is satisfied. 
%JYC opposite --> odd many b_i=1, others (odd many) are 0 for b_i's in (9)
Because  $f \not \in \mathscr{L}$, 
equation (\ref{eqn-inhomogeneous-deg-1}) must have been violated.
Being opposite, the left hand side is $1$.
Hence,  $c_1 \equiv 0$. In this case, we get $[1, \alpha^{c_1}](+-)$.
Since $c_1 \equiv 0$, we have  $[1, \alpha^{c_1}](+-) \in F_1^\mathscr{A}(+-)$.
\end{proof}

\begin{lemma} \label{lemma L 2}
If we have a rank 2 function $f \not \in \mathscr{L}$,
in $\#{\rm CSP}_2^c(\{f\})$, either we can pin to get a rank 1 function not in $\mathscr{L}$, or
we can realize a function in  the  sets
$\overline{\mathscr{A}}{\rm ++}$ or $F_3 ^ \mathscr{A} (+-)$.
\end{lemma}

\begin{proof}
If $f^2 \not \in \mathscr{A}$, we can realize $f^2{\rm ++}$, which is a function in $\overline{\mathscr{A}}{\rm ++}$.
Now we may assume $f^2 \in \mathscr{A}$.
If any pinning of $f$ always gives a function in $\mathscr{L}$, we can apply Lemma \ref{lemma r2 useful even odd}, \ref{lemma r2 useful odd to +} and \ref{lemma r2 useful  consistent opposite}, to get a rank 2 function
not in  $\mathscr{L}$, such that all 3 bundles are $(+)$, or all 3 bundles are opposite, or all 3 non-empty bundles are consistent,  or there are exactly  2 non-empty bundles which are consistent.

Let
the compressed function be $\alpha^{c_1 x_1 +c_2 x_2 + 2c_{12} x_1x_2}$.
We consider the following cases.

\begin{enumerate}
\item
All 3 bundles are $(+)$.

The matrix $(a_{ij})$ in equations 
 (\ref{eqn-homogeneous-combined})
and (\ref{eqn-inhomogeneous-combined})
is $\left[\begin{smallmatrix}
1 & 0 \\
0 & 1 \\
1 & 1  \end{smallmatrix}\right]$, and all $b_i =0$ since they are all of 
type  $(+)$.
If $c_1$ is odd, then we can pin $x_2 =0$.
The resulting rank 1 function violates equation (\ref{eqn-inhomogeneous-deg-1})
since 
the corresponding matrix is just $\left[\begin{smallmatrix}
1 \\ 1   \end{smallmatrix}\right]$, and it has bundle type $(++)$
and so both $b_i =0$.
If  $c_1$ is even, then we can pin $x_2 =1$.
The resulting rank 1 function also 
violates equation (\ref{eqn-inhomogeneous-deg-1})
since 
it has the same matrix $\left[\begin{smallmatrix}
1 \\ 1   \end{smallmatrix}\right]$,  but  it has bundle type $(+-)$
and so the $b$ vector is $(0, 1)^{\tt T}$. 
Hence we get  a rank 1 function not in $\mathscr{L}$.
 
%%%%%%%%%%%%%%%%%%%%%%%%%%%%%%%%%%%%%%%%%%%%%the following pf has a bug
%For $j=1,2$, equations (\ref{eqn-homogeneous-deg-1})
%are satisfied, and for (\ref{eqn-inhomogeneous-deg-1})
%the left hand sides are all 0:
 %$\sum_i a_{ij} b_i=0$.
%% since all $b_i =0$ by being
%%of type $(+)$.   
%If one of $c_1,c_2$  is odd,
%then we can pin either $x_1$ or $x_2$ 
%% to  0 , so the new bundle is (++)
%so that  equation (\ref{eqn-inhomogeneous-deg-1}) is not satisfied 
%for the resulting function.
%%%%%%%%%%%%%%%%%%%%%%%old text
%%%% say c1 =1, 
%%then equation (\ref{eqn-inhomogeneous-deg-1}) is not satisfied  and 
%%we can pin, keep this equation unsatisfied,
%%%%%%%%%%%%%%%%%%%%%%%%%%%%%%%%%%
%Thus we get a rank 1 function not in $\mathscr{L}$.
%Now suppose both $c_1,c_2$ are even. Then 
%equation (\ref{eqn-inhomogeneous-deg-1}) is satisfied for $f$,
%and therefore equation (\ref{eqn-inhomogeneous-deg-2}) is not satisfied.
%$c_{12}$ is odd, because equation (\ref{eqn-inhomogeneous-deg-2}) is not satisfied.
%We pin $x_2$ to $1$, to get a function of the form $F_1 ^ \mathscr{A} (+-)$, which is a rank 1 function not in $\mathscr{L}$.

\item
All 3 bundles are opposite.

The matrix $(a_{ij})$ in equations  (\ref{eqn-homogeneous-combined})
and (\ref{eqn-inhomogeneous-combined})
is $\left[\begin{smallmatrix}
1 & 0 \\ 
1 & 0 \\
0 & 1 \\
0 & 1 \\  
1 & 1 \\
1 & 1  \end{smallmatrix}\right]$. So 
%for $j=1,2$, equations (\ref{eqn-homogeneous-deg-1})
%JYC also need to say eqn (\ref{eqn-homogeneous-deg-2})
equations (\ref{eqn-homogeneous-deg-1}) and (\ref{eqn-homogeneous-deg-2})
are satisfied. The six $b_i$ are alternately 0's and 1's due to type $(+-)$.
If  $c_1$ or $c_2 \equiv 1$, then we 
can pin to get a function not in $\mathscr{L}$.
%%%JYC the matrix becomes 4 by 1, all 1's
% b's are +-+-, so two 1's ans two 0's.
So suppose  $c_1 \equiv c_2 \equiv 0$. Then
 equation (\ref{eqn-inhomogeneous-deg-1}) holds. 
Since $f  \not \in \mathscr{L}$, equation (\ref{eqn-inhomogeneous-deg-2})
must have been violated, and we get $c_{12} \equiv 0$.
This means that we have a function in $F_3 ^ \mathscr{A} (+-)$.
%
%Either we can pin to get a function not in $\mathscr{L}$, or equation (\ref{eqn-inhomogeneous-deg-1}) holds. This tells us
%that $c_1 \equiv c_2 \equiv 0$, because of the type $(+-)$.
%%%% every other b_i is 0, the other is 1.
%Because $f \not \in \mathscr{L}$, equation (\ref{eqn-inhomogeneous-deg-2})
%must be violated, then $c_{12} \equiv 0$.
% We get a function in $F_3 ^ \mathscr{A} (+-)$.

\item
All the bundles are consistent.

We have a function $q(cc)$,
where $q \not \in \mathscr{L}$.
 By a similar analysis of  the compressed function of $q(cc)$, we know
that $q \not \in \mathscr{A}$, and so we get $q(cc) \in \overline{\mathscr{A}}{\rm ++}$.
Indeed, if $q \in \mathscr{A}$, then $c_1 \equiv c_2 \equiv c_{12} \equiv 0$.
All left hand sides of  (\ref{eqn-homogeneous-deg-1}),
(\ref{eqn-homogeneous-deg-2}),
(\ref{eqn-inhomogeneous-deg-1}) and
(\ref{eqn-inhomogeneous-deg-2}) are 0, and this would imply that
$q \in \mathscr{L}$.

\end{enumerate}

\end{proof}

\begin{lemma} \label{lemma L 3}
If we have a rank 3 function $f \not \in \mathscr{L}$,
in $\#{\rm CSP}_2^c(\{f\})$, either we can pin to get a rank 2 function not in $\mathscr{L}$, or
we can realize a function in  the  sets
$\overline{\mathscr{A}}{\rm ++}$ or $F_7 ^ \mathscr{A} (+-)$.
\end{lemma}

\begin{proof}
If $f^2 \not \in \mathscr{A}$, we can realize $f^2{\rm ++}$, which is a function in $\overline{\mathscr{A}}{\rm ++}$.
Now we assume $f^2 \in \mathscr{A}$.
If any pinning of $f$ always gives a function in $\mathscr{L}$, we can apply Lemma \ref{lemma r3 useful even odd}, \ref{lemma r3 useful odd to +} and \ref{lemma r3 useful opposite consistent}, to get a rank 3 function
not in $\mathscr{L}$, such that all 7 bundles are $(+)$, or all 7 bundles are opposite, or all 7 bundles are non-empty consistent,  or there are exactly 3 non-empty bundles and 
they are all non-empty consistent.

We consider the following cases.

\begin{enumerate}
\item
All 7 bundles are $(+)$.

Suppose the rank 3 function is $h(+)$.
The
matrix $(a_{ij})$  in  equations  (\ref{eqn-homogeneous-combined})
and (\ref{eqn-inhomogeneous-combined})
is $\left[\begin{smallmatrix}
1 & 0 & 0 \\
0 & 1 & 0 \\
0 & 0 & 1 \\
1 & 1 & 0 \\
1 & 0 & 1 \\
0 & 1 & 1 \\
1 & 1 & 1  \end{smallmatrix}\right]$, with all $b_i =0$.

Let the compressed function of $h$ be
 \[\alpha^{c_1 x_1 +c_2 x_2 + c_3 x_3+ 2c_{12} x_1x_2 + 2c_{13} x_1x_3 + 2c_{23} x_2x_3 + 4c_{123} x_1x_2x_3}.\]

We show that in this case we can pin $x_3$ to get a rank 2 function not
in $\mathscr{L}$.

If we pin $x_3=0$, the bundle $x_3$ disappears, and the remaining
six bundles are merged into three bundles of type $(++)$, 
and the new matrix $(a_{ij})$
for the rank 2 function  in  equations  (\ref{eqn-homogeneous-combined})
and (\ref{eqn-inhomogeneous-combined})
is $\left[\begin{smallmatrix}
1 & 0  \\
1 & 0 \\
0 & 1 \\
0 & 1 \\
1 & 1 \\
1 & 1  \end{smallmatrix}\right]$, with all six $b_i =0$.
The new expression for the compressed function is
$\alpha^{c_1 x_1 +c_2 x_2 + 2c_{12} x_1x_2}$.
If any of $c_1, c_2, c_{12}$ is odd, then
some equation in (\ref{eqn-inhomogeneous-deg-1}) 
or (\ref{eqn-inhomogeneous-deg-2}) is violated, thus
we get a rank 2 function  not in $\mathscr{L}$.

Suppose $c_1 \equiv c_2 \equiv c_{12} \equiv 0$.
Now we pin $x_3=1$, the bundle $x_3$ disappears, and the remaining
six bundles are merged into three bundles of type $(+-)$, 
and the new matrix $(a_{ij})$
for the rank 2 function is the same as above, but
the new vector $b = (0, 1, 0, 1, 0, 1)^{\tt T}$.
The new expression for the compressed function is
$\alpha^{(c_1 + 2 c_{13}) x_1 + (c_2 + 2 c_{23}) x_2 + 2(c_{12} +
2c_{123}) x_1x_2}$. 
%Thus mod 2 the changes are irrelevant
Now  equation (\ref{eqn-inhomogeneous-deg-2}) is violated,
% only one term remains, last 1*1 * (b_6=1) = 1 on LHS, but RHS = c_{{12}  =0
and we get a rank 2 function  not in $\mathscr{L}$.
%
%
%
%Because we ca enot pin $h(+)$ to get a function not in $\mathscr{L}$, equations (\ref{eqn-inhomogeneous-deg-1}) and (\ref{eqn-inhomogeneous-deg-2}) hold. We get the $c_i$ and $c_{ij}$ are even. Because $h(+) \not \in \mathscr{L}$, equation (\ref{eqn-inhomogeneous-deg-3}) does not hold. Hence $c_{123}$ is odd.
%We get $h \in \mathscr{A}^\alpha$. We can pin $x_3$ to $1$, to get a function in $F_3 ^ \mathscr{A} (+-)$, which is a rank 2 function not in $\mathscr{A} $.

\item
All 7 bundles are opposite.

Suppose the rank 3 function is $h(+-)$. 
Let its compressed function be
\[ \alpha^{c_1 x_1 +c_2 x_2 + c_3 x_3+ 2c_{12} x_1x_2 + 2c_{13} x_1x_3 + 2c_{23} x_2x_3 + 4c_{123} x_1x_2x_3}.\]
Similar to the proof above,
if it is not the case that  
$c_1 \equiv c_2 \equiv c_3 \equiv c_{12} \equiv c_{13}\equiv c_{23} \equiv 0$, 
we can pin to get a rank 2 function not in $\mathscr{L}$.
%
%Either we can pin to get a function not in $\mathscr{L}$, or equation (\ref{eqn-inhomogeneous-deg-1}) and  (\ref{eqn-inhomogeneous-deg-2}) hold. This tells
%us that $c_1 \equiv c_2 \equiv c_3 \equiv c_{12} \equiv c_{13}\equiv c_{23} \equiv 0$.
But if all these values are even, then
because $f \not \in \mathscr{L}$, by equation (\ref{eqn-inhomogeneous-deg-3}),
we get  $c_{123} \equiv 0$. So we have
 a function $h(+-) \in F_7 ^ \mathscr{A} (+-)$.

\item
All the bundles are consistent.

Suppose the rank 3 function is $h(cc)$.
A similar analysis about the compressed function gives
$c_1 \equiv c_2 \equiv c_3 \equiv c_{12} \equiv c_{13}\equiv c_{23} \equiv 0$ and $c_{123} \equiv 1$. This tells us that
 the essential function $h$ is not in  $\mathscr{A}$, regardless of whether
 the essential arity is 3 or 7.
\end{enumerate}
\end{proof}

\begin{lemma}\label{lemma: rank>3 to at most 3}
Suppose $f \not \in \mathscr{L}$. In $\#{\rm CSP}_2^c(\{f\})$, either we can get a
function in $\overline{\mathscr{A}}{\rm ++}$, or we can realize a function of rank at most 3 not in $\mathscr{L}$.
\end{lemma}

\begin{proof}
If $f^2 \not \in \mathscr{A}$, we can realize $f^2{\rm ++}$, which is a function in $\overline{\mathscr{A}}{\rm ++}$.
Now assume $f^2 \in \mathscr{A}$.
Hence, $f$ has the form (\ref{generaic-form-f-sq-in-A-detailed}). Suppose $f={\rm supp}(f) \cdot  \alpha^{L(x) + 2 Q(x) + 4 H(x)}$. If a coefficient of $H$ is even, we remove the corresponding monomial since $\alpha^8=1$.  According to Theorem \ref{thm: Local affine}, we need to consider two cases.

The first case is that $H$ is not homogeneous of degree 3.
Let $M$ be  a monomial having the  minimum degree among all
monomials in $H$ of degree at least $4$.
% Let $T$ denote the set of monomials of $H$ with degree at least $4$. Suppose $M$ is a monomial in $T$ which achieves the minimum degree of $T$.
Of course the degree of $M$ is at least $4$. We pin the free variables of $f$ which are outside of $M$ to $0$, and pin the variables in $M$ to $1$ except $4$ of them, to get a new function of rank 4.
 The new $H$ polynomial has a unique monomial $c_{1234}x_1x_2x_3x_4$ of degree $4$ where $c_{1234} \equiv 1$.
By Theorem \ref{thm: Local affine}, it is still not in $\mathscr{L}$. We denote this new function still by $f$.

From $f$ we construct 3 functions of rank $3$; if they are all in $\mathscr{L}$, we will get a contradiction.

Consider $f^{x_4=0}$. Substitute $x_4=0$ into the compressed function of $f$, we find in the
$H$ polynomial of $f^{x_4=0}$, the coefficient of $x_1x_2x_3$ is $c_{123}$.
If $f^{x_4=0} \in \mathscr{L}$,
 according to condition (\ref{eqn-inhomogeneous-deg-3}),
this coefficient has the same parity as
 the  number of variables in the bundle $x_1+x_2+x_3$ in  $f^{x_4=0}$
that are labeled as $(-)$, i.e., those variables that are
equal to $x_1+x_2+x_3+1$ on the support.
These variables come from the union of two sets
of variables of $f$.
We denote by $v_{123}$ the number of variables 
in the bundle $x_1+x_2+x_3$ in  $f$  that are  labeled as $(-)$, i.e., 
those variables that are
equal to $x_1+x_2+x_3+1$ on the support. Similarly denote
$v_{124}$ and  $v_{1234}$ the numbers of
variables that are
equal to $x_1+x_2+x_4+1$ and $x_1+x_2+x_3+x_4+1$,
respectively, in $f$ on its support.
%These variables are the union of original $x_1+x_2+x_3+1$ and $x_1+x_2+x_3+x_4+1$ variables of $f$.
We have $c_{123} \equiv v_{123} + v_{1234}$.

If we  similarly consider $f^{x_3=0}$, we get 
$c_{124} \equiv v_{124} + v_{1234}$.

The third rank 3  function we construct is
 $f^{x_3=x_4}$.  To do so,
we connect one $x_3$ variable and one $x_4$ variable by a $(=_4)$ function,
and pin any variable labeled $x_3+x_4(+)$ to 0
and any variable labeled $x_3+x_4(-)$ to 1.
We get one extra condition that $x_3=x_4$, which narrows the support. That is, we get $\chi_{x_3=x_4} \cdot f$.  The 15 bundles of $f$,
namely $x_1, x_2, \ldots, x_1+x_2+x_3+x_4$, turn into 7 bundles. The new bundle $x_1+x_2+x_3$ is the union of the original bundles $x_1+x_2+x_3$ and $x_1+x_2+x_4$. 
Hence, if $f^{x_3=x_4} \in \mathscr{L}$,
its corresponding coefficient $c'_{123} \equiv v_{123} + v_{124}$,
 according to condition (\ref{eqn-inhomogeneous-deg-3})
on $f^{x_3=x_4}$.
Hence $c'_{123} \equiv c_{123} + c_{124}$.

Substitute $x_3=x_4$ into the compressed function of $f$.
We get $c'_{123} \equiv c_{123}+c_{124}+c_{1234}$. 
Thus we reach a contradiction $c_{1234} \equiv 0$. 
 %Hence, $c_{123}+c_{124}+c_{1234} \equiv \#_{x_1+x_2+x_3+1}+\#_{x_1+x_2+x_4+1}$.

%Put all the 3 equations together, we get $c_{1234} \equiv 0$. A contradiction. 

We conclude that
 at least one of the three rank 3 functions $f^{x_4=0}$, $f^{x_3=0}$ and $f^{x_3=x_4}$, is not in $\mathscr{L}$.

Now, we can assume that $H$ is homogeneous of degree 3, 
and consider the second case that one of the equations 
(\ref{eqn-homogeneous-deg-1}), 
(\ref{eqn-homogeneous-deg-2}), 
(\ref{eqn-homogeneous-deg-3}), 
(\ref{eqn-homogeneous-deg-4}), 
(\ref{eqn-inhomogeneous-deg-1}), 
(\ref{eqn-inhomogeneous-deg-2}) and 
(\ref{eqn-inhomogeneous-deg-3}) does not hold.
If one of the 6 equations 
(\ref{eqn-homogeneous-deg-1}), 
(\ref{eqn-homogeneous-deg-2}), 
(\ref{eqn-homogeneous-deg-3}), 
(\ref{eqn-inhomogeneous-deg-1}), 
(\ref{eqn-inhomogeneous-deg-2}) and 
(\ref{eqn-inhomogeneous-deg-3}) does not hold,
then we can pin to get a function of rank at most 3 not in $\mathscr{L}$.
 For example, if equation (\ref{eqn-inhomogeneous-deg-2}) 
does not hold for $j<k$, we can keep $x_j$ and $x_k$, 
and pin other free variables to 0. 
Now, we can assume all the 6 equations 
(\ref{eqn-homogeneous-deg-1}), 
(\ref{eqn-homogeneous-deg-2}), 
(\ref{eqn-homogeneous-deg-3}), 
(\ref{eqn-inhomogeneous-deg-1}), 
(\ref{eqn-inhomogeneous-deg-2}) and 
(\ref{eqn-inhomogeneous-deg-3}) hold,
 and equation (\ref{eqn-homogeneous-deg-4}) does not hold for $j<k<l<m$.

Firstly, keep $x_j,x_k,x_l$ and $x_m$ and pin other free variables to 0, 
to get a function $h$.
The left hand side of  (\ref{eqn-homogeneous-deg-4})
is the number of variables in all bundles in $f$ with a name that contains
$x_j, x_k, x_l, x_m$, i.e., any name that is of the form
$x_j+x_k+x_l+x_m+ \mbox{any affine linear form of other $x$'s}$.
This number is precisely the number of variables
in the bundle named $x_j+x_k+x_l+x_m$ in $h$, 
i.e., the number of variables named
$x_j+x_k+x_l+x_m$ or $x_j+x_k+x_l+x_m+1$ in $h$.
Because $f$ fails (\ref{eqn-homogeneous-deg-4}), this number is odd,
and so $h$ still fails equation (\ref{eqn-homogeneous-deg-4}).
% $\#_{x_j+x_k+x_l}+\#_{x_j+x_k+x_l+1}$ is even, and $h$ still satisfies equation (\ref{eqn-homogeneous-deg-3}).
Consider $h^{x_m=0}$ and $h^{x_m=1}$. Because the $H$ polynomial of $f$
 is homogeneous of degree 3, when we set $x_m=0$ or $x_m=1$,
there are no new cubic terms formed, and thus
 the $c_{jkl}$ coefficients of $h^{x_m=0}$ and $h^{x_m=1}$
are the same as the $c_{jkl}$ coefficient of $f$.
 The $x_j+x_k+x_l+1$ variables of $h^{x_m=0}$ come
 from the $x_j+x_k+x_l+1$ variables and
the $x_j+x_k+x_l+x_m+1$ variables of $h$.
Meanwhile, the $x_j+x_k+x_l+1$ variables of $h^{x_m=1}$ come 
from the $x_j+x_k+x_l+1$ variables and
the $x_j+x_k+x_l+x_m$ variables of $h$. 
Hence, they have opposite parities,
as their sum is odd. But they are respectively
the left hand sides of the equation (\ref{eqn-inhomogeneous-deg-3})
for  $h^{x_m=0}$ and $h^{x_m=1}$, whose right hand sides are 
the same $c_{jkl}$. It follows that
one of the two rank 3 functions $h^{x_m=0}$ and $h^{x_m=1}$ must
fail equation (\ref{eqn-inhomogeneous-deg-3}), and thus
 not in $\mathscr{L}$.
\end{proof}

Putting Lemma \ref{lemma: rank>3 to at most 3}, \ref{lemma L 3}, \ref{lemma L 2} and \ref{lemma L 1} together, we get the following lemma.

\begin{lemma} \label{lemma reduce L}
If we have a  function $f \not \in \mathscr{L}$, then in $\#{\rm CSP}_2^c(\{f\})$, we can realize a function of the form:
\begin{itemize}
  \item $\overline{\mathscr{A}}{\rm ++}$,
\item  $F_1^\alpha$, $F_1^\mathscr{A}$,  $F_1^\mathscr{A}(+-)$,
\item $F_3 ^ \mathscr{A} (+-)$,
\item $F_7 ^ \mathscr{A} (+-)$.
\end{itemize}
\end{lemma}

\subsection{Putting Things Together}

\begin{lemma} \label{lemma: together two +-}
For any $j,k \in \{1,3,7\}$, there is some  $s \in \{1,3,7\}$,  
such that we can realize in the setting
\#{\rm CSP}$_2^c (F_j^\mathscr{A}(+-),F_k^\alpha(+-))$,
 a function in the set
 $F_{s}^\alpha({\rm ++})$.
\end{lemma}

\begin{proof}
Let ${\rm Min}=\min \{j,k\}$ and ${\rm Max}=\max \{j,k\}$.

If ${\rm Min}={\rm Max}$, we overlay two functions 
from $F_k^\mathscr{A}(+-)$ and $F_k^\alpha(+-)$ by bundles, and connect the
 variable labeled $(-)$ in each bundle of one function
to the  variable labeled $(-)$  in the
 corresponding bundle  of the other function, to get
a function in  $F_k^\alpha(++)$.

If ${\rm Min}=1$, say ${\rm Min}=j$.
Then for each bundle of
a function in $F_{k}^\alpha(+-)$, we merge a $F_j^\mathscr{A}(+-)$ 
function with it by connecting the corresponding variables 
 labeled $(-)$, to get
a function in $F_{k}^\alpha({\rm ++})$.
If ${\rm Min}=k$ just switch $j$ and $k$.

The remaining case is ${\rm Min}=3$ and ${\rm Max}=7$. Suppose we have $g \in F_3^\mathscr{A}(+-), h \in F_7^\alpha(+-)$.

We take one copy of $\tilde{h}(x_1,x_2,x_3, x_2+x_3, x_1+x_3, x_1+x_2, x_1+x_2+x_3)(+-)$, and three copies of $g$: $\tilde{g}(u_1,u_2,u_1+u_2)(+-)$, $\tilde{g}(v_1,v_2,v_1+v_2)(+-)$, $\tilde{g}(w_1,w_2,w_1+w_2)(+-)$ to 
construct a function realizing a function
 $f$ in $(x_1, u_1,  x_2, v_1,  x_3, w_1, 
x_2+x_3,u_2, x_1+x_3, v_2,  x_1+x_2, w_2, x_1+x_2+x_3, w_1+w_2)$.
(See Figure~\ref{Fig f3f7} for an illustration.). We merge the $u_1$ bundle of
one copy of $g$ with the $x_1$ bundle of $h$, by equating the 
variables labeled  $u_1(-)$ and d $x_1(-)$. 
 That is, set $u_1+1=x_1+1$.
 Similarly, we merge the bundles $u_2$ with $x_2+x_3$, merge $v_1$ with $x_2$, merge $v_2$ with $x_1+x_3$, merge $w_1$ with $x_3$ and merge $w_2$ with $x_1+x_2$. These 5 merging operations are accomplished by
similarly connecting 5 pairs of variables labeled $(-)$
as illustrated in Figure~\ref{Fig f3f7}. 
After these mergings, the remaining
 four bundles $x_1+x_2+x_3$, $u_1+u_2$, $v_1+v_2$, $w_1+w_2$ are already merged into one bundle of type $(++++----)$ automatically, which can become $(++)$ by 
3 collations. Including these 3 collations there are a total of
9 pair of equating variables all labeled $(-)$
except the pair $u_1 + u_2(+)$ and $v_1 + v_2(+)$.
The 3 collations are algebraically
 $x_1+x_2+x_3+1=u_1+u_2+1$, $u_1+u_2=v_1+v_2$, $v_1+v_2+1=w_1+w_2+1$, and
we leave $x_1+x_2+x_3$ and $w_1+w_2$ in this bundle.
The equations from these  3 collations are algebraic consequences
of the previous 6 merging operations.
To summarize the above description,  the 18 variables among 32 variables
 of the 4 functions are matched by 9 edges in this gadget,
 we list them by the following equations.
\begin{displaymath}
\left\{ \begin{array}{rcl}
u_1+1& \equiv &x_1+1   \\
u_2+1& \equiv &x_2+x_3+1  \\
 v_1+1& \equiv &x_2+1   \\
v_2+1& \equiv &x_1+x_3+1  \\
w_1+1& \equiv &x_3+1   \\
w_2+1& \equiv &x_1+x_2+1  \\
u_1+u_2+1 &  \equiv & x_1+x_2+x_3+1 \\
u_1+u_2& \equiv &v_1+v_2 \\
v_1+v_2+1& \equiv &w_1+w_2+1
\end{array} \right.
\end{displaymath}

Removing algebraic redundancy, this  system of equations is
equivalent to
% This system of equations is equal to
\begin{displaymath}
\left\{ \begin{array}{rcl}
u_1&  \equiv  &x_1   \\
u_2&  \equiv  &x_2+x_3  \\
 v_1&  \equiv  &x_2   \\
v_2&  \equiv  &x_1+x_3  \\
w_1&  \equiv  &x_3   \\
w_2&  \equiv  &x_1+x_2  
%\\
%w_1+w_2& = &x_1+x_2+x_3
\end{array} \right.  .
\end{displaymath}
with 
consequences $u_1+u_2  \equiv v_1+v_2  \equiv w_1+w_2  \equiv x_1+x_2+x_3$.

The external variables of this gadget
has 7 bundles $x_1, x_2, x_3, x_1+x_2, x_1+x_3, x_2+x_3$ and
$x_1+x_2+x_3$ and are all of the type $(++)$.
%JYC i don't think this is true, can't give arbitrary 14 values
%Given an input to the 14 external variables of the gadget $f$, there is a unique value assignment to the $9$ internal edges such that all the 4 function give no zero values.
%%%%
It is not hard to verify that, the above system of linear equations are all the new introduced linear constraints on the 14 external variables, besides the natural linear constraints of the support of $h$ already shown by the names of variables. So $f$ has 7 bundles, such that it 
has the form $\tilde{f}(x_1, x_2,x_3,x_2+x_3, x_1+x_3,x_1+x_2, x_1+x_2+x_3)(++)$.  Denote the input variables
 of $\tilde{f}$ by $X$.
 We calculate $f$ on a general input $X(++)$ on the support. Every
such assignment has a unique extension to the  $9$ internal edges
so that the 4 functions give no zero values. We
  get $\tilde{f}(X)=\tilde{h}(X)
\tilde{g}(x_1, x_2+x_3, x_1+x_2+x_3)
\tilde{g}(x_2, x_1+x_3, x_1+x_2+x_3)
\tilde{g}(x_3, x_1+x_2, x_1+x_2+x_3)$. 
Let $\tilde{h} = M_\alpha^{\otimes 7}  q$, where $q \in F_7^\mathscr{A}$,
since $h \in F_7^\alpha(+-)$ by assumption.
 We have $\tilde{h}(X)=M_\alpha^{\otimes 7} (X) \cdot q(X)$,
where we view $M_\alpha$ as a generalized binary equality function
(which modifies each external variable). 
 Now, we see that $\tilde{f}$ is a product of $M_\alpha^{\otimes 7}$ with 
4 functions in $\mathscr{A}$: $q$, $\tilde{g}$, $\tilde{g}$ and $\tilde{g}$.
The product of 4 functions in $\mathscr{A}$ is in
$\mathscr{A}$.  Hence, $\tilde{f} \in F_7^\alpha$ and $f \in  F_7^\alpha(++)$.

Suppose we have $g \in F_3^\alpha(+-), h \in F_7^\mathscr{A}(+-)$.
 The construction of the gadget $f$ and the analysis of the support 
of $f$ are the same.
In the last step, we calculate $\tilde{f}$. Let $\tilde{g} = M_\alpha^{\otimes 3}  p$, where $p \in F_3^\mathscr{A}$,
since $g \in F_3^\alpha(+-)$ by assumption.
 Then we have  
\begin{eqnarray*}
\tilde{f}(X)
&=&
\tilde{h}(X)\tilde{g}(x_1, x_2+x_3, x_1+x_2+x_3)\tilde{g}(x_2, x_1+x_3, x_1+x_2+x_3)\tilde{g}(x_3, x_1+x_2, x_1+x_2+x_3)\\
&=&
\tilde{h}(X)p(x_1, x_2+x_3, x_1+x_2+x_3)p(x_2, x_1+x_3, x_1+x_2+x_3)p(x_3, x_1+x_2, x_1+x_2+x_3) \\
& & \cdot M_\alpha^{\otimes 7}(X) 
M_\alpha(x_1+x_2+x_3)M_\alpha(x_1+x_2+x_3).
\end{eqnarray*}
The second equality holds, because from each $\tilde{g}$ we get a $p$ 
and a $M_\alpha^{\otimes 3}$ applied to the 3 variables.
The modifier factor $M_\alpha^{\otimes 7}(X)$ is obtained by
collecting one factor $M_\alpha$ on each of the inputs
$x_1, x_2, x_3,  x_1+x_2, x_1+x_3, x_2+x_3, x_1+x_2+x_3$,
and there are still two extra factors of  $M_\alpha$ on 
$x_1+x_2+x_3$.

Because $M^2_\alpha ( = M_{\alpha^2} = M_i )$ is in $\mathscr{A}$,
 $\tilde{h}$ and $p$ both belong to $\mathscr{A}$,
and $\tilde{f} \in F_7^\alpha$,
we have  $f \in  F_7^\alpha(++)$.
\end{proof}

\begin{lemma} \label{lemma: together two +}
For any $j,k \in \{1,3,7\}$, there is some  $s \in \{1,3,7\}$, 
such that  we can realize in the setting
 \#{\rm CSP}$_2^c (F_j^\mathscr{A}(+),F_k^\alpha(+))$,
a function in the set
 $F_{s}^\alpha({\rm ++})$.
\end{lemma}

\begin{proof}
The proof is similar to the proof of
Lemma~\ref{lemma: together two +-}. 
The only difference is a slight modification in the gadget construction
when ${\rm Min}=3$ and ${\rm Max}=7$.
Now each bundle of the constituent functions from $F_j^\mathscr{A}(+)$ and
$F_k^\alpha(+)$ has a single variable labeled $(+)$.
  When we merge two bundles, we connect these two variables 
by a copy of  $(=_4)$.
After the six
connection steps have been made,
 the $x_1+x_2+x_3$ bundle includes
$u_1+u_2$, $v_1+v_2$ and $w_1+w_2$, and has type
$({\rm ++}{\rm ++})$.
 We turn it to $({\rm ++})$ type using collation.
\end{proof}

\begin{lemma} \label{lemma: together one alpha +- one unary}
For any $j \in \{1,3,7\}$, we can realize a
function in the set $F_j^\alpha({\rm ++})$, 
from either \#{\rm CSP}$_2^c (F_j^\alpha(+-),F_1^\mathscr{A})$ 
or \#{\rm CSP}$_2^c (F_j^\alpha(+-),F_1^\alpha)$.
\end{lemma}

\begin{proof}
Using $(=_4)$ we change $F_j^\alpha(+-)$ to $F_j^\alpha({\rm +++-})$.
 For each bundle, we can apply a function in
 $F_1^\mathscr{A}$ to one variable labeled $(+)$ and one
variable labeled $(-)$ to get a function in  $F_j^\alpha({\rm ++})$.
We can also apply a function in 
 $F_1^\alpha$ to one variable labeled $(+)$ and one
variable labeled $(-)$   in each bundle 
to get a $F_j^\alpha({\rm ++})$ function.
Note that in the latter case, the pair of variables labeled $(+)$ and
$(-)$ in a single bundle will always take opposite values
on the support and therefore the aggregate modifcation
by $F_1^\alpha$ is a constant factor, thus it does not change
the membership of its essential function in $F_j^\alpha$.
\end{proof}

\begin{lemma} \label{lem: remove cc}
For any set of constraint functions $\mathscr{F}$,
let $\mathscr{F}(++) = \{f(++) \mid f \in \mathscr{F}\}$. Then
\[
\#{\rm CSP}(\mathscr{F}) \leq_{\rm T} \#{\rm CSP}_2(\mathscr{F}(++))
 \leq_{\rm T}  \#{\rm CSP}_2^c(\mathscr{F}(++)).\]
\end{lemma}

\begin{proof}
In an instance of \#{\rm CSP}$(\mathscr{F})$
if we  replace each edge by two parallel edges,
and replace each occurrence of any $f \in \mathscr{F}$
by $f(++) \in \mathscr{F}(++)$,
 we get an instance in  \#{\rm CSP}$_2(\mathscr{F}(++))$,
 and they have the same value.
\end{proof}

\vspace{.1in}

\begin{proof}[Proof of \#P-hardness part of Theorem \ref{thm csp2 theorem}:]~\\

%\noindent
%~\\
We have $\mathscr{F} \not \subseteq \mathscr{P}$, $\mathscr{F} \not \subseteq \mathscr{A}$, $\mathscr{F} \not \subseteq  \mathscr{A}^\alpha$ and $\mathscr{F} \not \subseteq \mathscr{L}$.
By Lemma \ref{lemma not product type}, 
we can realize a function  $p{\rm ++}$ in 
  the setting 
\#{\rm CSP}$_2^c (\mathscr{F})$, such that $p \not \in \mathscr{P}$.
Now the idea is to obtain a function from $\overline{\mathscr{A}}{\rm ++}$,
then we can apply Lemma~\ref{lem: remove cc}. This will allow us to apply
Theorem~\ref{thm old csp dichotomy} to prove \#P-hardness.

 Lemmata~\ref{lemma not A pin to 321}, \ref{lemma not A-alpha  pin to 321} and \ref{lemma reduce L}  tell us  respectively
 what we can get from  
$\mathscr{F} \not \subseteq \mathscr{A}$, 
$\mathscr{F} \not \subseteq  \mathscr{A}^\alpha$ and 
$\mathscr{F} \not \subseteq \mathscr{L}$.
 If one of the lemmas brings us as one of the  direct outcomes
 a function $f{\rm ++}$ in 
 $\overline{\mathscr{A}}{\rm ++}$, together with $p{\rm ++}$, 
we have \#{\rm CSP}$(\{p,f\}) \leq_{\rm T}$ \#{\rm CSP}$_2^c (\mathscr{F})$ 
by Lemma \ref{lem: remove cc}. Then by Theorem~\ref{thm old csp dichotomy}, 
we have proved that \#{\rm CSP}$_2^c (\mathscr{F})$ is $\#$P-hard.

So we may assume 
the direct outcomes of the three lemmas contain
 no function in $\overline{\mathscr{A}}{\rm ++}$.
 We analyze the possible combinations of outcomes,
and still construct a function in  $\overline{\mathscr{A}}{\rm ++}$ 
 to finish the proof.

If the outcomes contain no functions
belonging to some $F_s^\alpha(+-)$ ($s \in \{1,3,7\}$), 
then by the outcomes of Lemma~\ref{lemma not A pin to 321} and
Lemma~\ref{lemma not A-alpha  pin to 321} for the cases of
reducing $\overline{\mathscr{A}}$ and 
$\overline{\mathscr{A}^\alpha}$ respectively,
 there must be both a function in
$F_k^\alpha(+)$ ($k \in \{1,3,7\}$) and a function in
 $F_j^\mathscr{A}(+)$ ($j \in \{1,3,7\}$).
By Lemma~\ref{lemma: together two +}, we can realize some
function in  $F_{s}^\alpha({\rm ++})$ ($s \in \{1,3,7\}$).
But any function from  $F_{s}^\alpha({\rm ++})$ ($s \in \{1,3,7\}$)
is from  $\overline{\mathscr{A}}{\rm ++}$, and so
\#{\rm CSP}$_2^c (\mathscr{F})$ is \#P-hard.

Now suppose the outcomes of
 Lemma~\ref{lemma not A pin to 321} and
Lemma~\ref{lemma not A-alpha  pin to 321} 
 contain a function in $F_j^\alpha(+-)$ ($j \in \{1,3,7\}$).
 By Lemma~\ref{lemma reduce L}, either we have a 
function in $F_k^\mathscr{A}(+-)$ ($k \in \{1,3,7\}$), 
or we have a function in  $F_1^\mathscr{A}$, or
  we have a function in  $F_1^\alpha$.
In the first case, by Lemma~\ref{lemma: together two +-} 
we realize some function in $F_{s}^\alpha({\rm ++})$ ($s \in \{1,3,7\}$).
In the second and third cases,
 by Lemma~\ref{lemma: together one alpha +- one unary}
we can also realize some function in $F_{s}^\alpha({\rm ++})$ ($s \in \{1,3,7\}$).
As noted above, any function in $F_{s}^\alpha({\rm ++})$ 
is from  $\overline{\mathscr{A}}{\rm ++}$.
Therefore \#{\rm CSP}$_2^c (\mathscr{F})$ is \#P-hard.
This copletes the proof of Theorem~\ref{thm csp2 theorem}.
%
%By Lemma \ref{lemma: together two +-} and \ref{lemma: together one alpha +- one unary}, we can realize some $F_{s}^\alpha({\rm ++})$ ($s \in \{1,3,7\}$) type function.
%
%Because $F_{s}^\alpha({\rm ++})$ ($s \in \{1,3,7\}$) functions are  $\overline{\mathscr{A}}{\rm ++}$ function, \#{\rm CSP}$_2^c (\mathscr{F})$ is \#P-hard.
\end{proof}

\section{Complexity dichotomy theorem of Holant$^c$}

We use a  $2\times 4$ matrix  to denote a function of arity $3$,
with rows indexed by $x_1 =0, 1$ and columns indexed by $x_2x_3 =00,01,10,11$,
thus
$f= \begin{pmatrix} f^{000} &  f^{001} & f^{010} & f^{011} \\
                 f^{100} &  f^{101} & f^{110} & f^{111}
\end{pmatrix}$.

We say a function is a generalized {\sc Equality} if its support
is a pair of antipodal points $\{{\bf x}, \overline{{\bf x}}\}$.
A binary  function is a  generalized {\sc Disequality}
if it has support $\{01, 10\}$, i.e.,  $f = (0,a,b,0)$ with $ab \not =0$.
\begin{lemma}\label{lemma-odd-equality}
Let $f\in \mathscr{F}$ be a  generalized {\sc Equality} of arity $3$.  Then {\rm Holant}$^c(\mathscr{F})$ is \#P-hard unless $\mathscr{F} \subseteq \mathscr{A}$, $\mathscr{F} \subseteq \mathscr{A}^\alpha$ or $\mathscr{F} \subseteq \mathscr{P}$. In
all three exceptional cases, the problem is in P (and belongs to the tractable families for \#{\rm CSP}$_2^c$).
%%% JYC shouldn't it be CSP_2^c?
\end{lemma}

\begin{proof}
Let ${\rm supp}(f)$  be $\{(a_1,a_2,a_3), (1-a_1,1-a_2,1-a_3)\}$.
 If they are $000$ and $111$, then $f$ is a symmetric function $[a,0,0,b]$,
with $ab \not =0$. Otherwise, there are both $0$ and $1$ among $(a_1,a_2,a_3)$.
By renaming variables, without loss of generality, we assume that
they are $001$ and $110$.
By connecting $x_1$ and $x_2$ with a self loop, we get a unary function $[a,b]$ with $ab \neq 0$.
By connecting this unary function to $x_1$ of $f$, we get a
 generalized {\sc Disequality} function $(0,c,d,0)$,
with $cd \not =0$. By connecting this generalized {\sc Disequality}
 to $x_3$ of $f$, we obtain a symmetric generalized {\sc Equality}.
%also get a symmetric function of form $[a',0,0,b']$.
After a scaling,  in both cases, we may assume to have $[1,0,0,b]$.
%After a scaling, we assume that it is  $[1,0,0,b]$

Taking a self loop on $[1,0,0,b]$ we get $[1,b]$.
Connecting one $[1,b]$ back to $[1,0,0,b]$, we get  $[1,0,b^2]$.
Connection one $[1,0,b^2]$ back to $[1,0,0,b]$, we get $[1,0,0,b^3]$.
Then we have
%Having $[1,b]$,  $[1,0,b^2]$,  $[1,0,0,b^3]$, we have
\[
{\rm Holant}([1,b],[1,0,b^2], [1,0,0,b^3] ~\mid~ \mathscr{F} \cup \{[1,0,1]\})\leq_{\rm T}
  {\rm Holant}^c( \mathscr{F}).\]
After a holographic reduction by
$T=\begin{pmatrix} 1 &  0  \\
                0  & b
\end{pmatrix}$, the left hand side becomes $\{(=_1), (=_2), (=_3)\}$
 and the right hand side becomes $T \mathscr{F}  \cup \{[1,0,b^2]\}$.
By the dichotomy theorem for \#CSP with each variable appearing
 at most three times~\cite{CaiLX14},
we know that the problem is   \#P-hard unless $T\mathscr{F}
  \cup \{[1,0,b^2]\} \subseteq \mathscr{P}$
 or $T\mathscr{F}  \cup \{[1,0,b^2]\} \subseteq \mathscr{A}$.
 A diagonal holographic reduction keeps the class  $\mathscr{P}$
invariant, so  $T\mathscr{F}  \cup \{[1,0,b^2]\} \subseteq \mathscr{P}$
  iff $\mathscr{F} \subseteq \mathscr{P}$, as $[1,0,b^2]\in \mathscr{P}$.
If $T\mathscr{F}  \cup \{[1,0,b^2]\} \subseteq \mathscr{A}$, we have $b^2 \in \{\pm 1, \pm i \}$.
If $b^2 =\pm 1 $, the holographic transformation $T$
 also keeps the class   $\mathscr{A}$ invariant,
and so  $T\mathscr{F} \cup \{[1,0,b^2]\} \subseteq \mathscr{A}$
 iff $\mathscr{F} \subseteq \mathscr{A}$.
 If $b^2 =\pm i $, the holographic transformation is the
 $\alpha$ transformation (followed by a transformation that keeps
$\mathscr{A}$ invariant),  and so  $T\mathscr{F}  \cup \{[1,0,b^2]\}
 \subseteq \mathscr{A}$  iff $\mathscr{F} \subseteq \mathscr{A}^\alpha$.
 This completes the proof.
\end{proof}

%%%%%%% JYC
%%%%%%% the following sentence is vague. what does it mean by "no longer need"
%%%%%%% note that technically in the proof of this past lemma, we did
%%%%%%% not use pinning.

In the above proof, once we have a symmetric generalized {\sc Equality}  $[1,0,0,b]$, we no longer need the two unary pinning functions
$\Delta_0$ and $\Delta_1$. We will use this fact later.
% proof.

%%% JYC: isn't all following CSP^2 should be CSP_2^c ?

\begin{lemma}\label{lemma-even-equality}
Suppose $\mathscr{F}$ contains a generalized {\sc Equality} $f$ of arity $4$,
 then {\rm Holant}$^c(\mathscr{F})\equiv_{\rm T}${\rm CSP}$^2(\mathscr{F})$.
\end{lemma}
\begin{proof}
Let ${\rm supp}(f)$  be $\{(a_1,a_2,a_3,a_4), (1-a_1,1-a_2,1-a_3, 1-a_4)\}$.
 By renaming variables,  we only need to consider
three possibilities for  $(a_1,a_2,a_3,a_4)$ : $0000,0001$ or $0011$.
 For $0000$, the function is already symmetric $[a,0,0,0, b]$,
with $ab \neq 0$. For $0001$, by connecting $x_1$ and $x_2$ with a self loop,
  we get a generalized {\sc Disequality} function.
By connecting this generalized {\sc Disequality} to $x_4$ of $f$,
we get a symmetric function of the form $[a,0,0,0, b]$ ($ab \neq 0$).
  For  $0011$, we take two copies of $f$ and connect their
variables $x_3$ and $x_4$ respectively. From this, we
 also get a symmetric function of the form $[a,0,0,0, b]$  ($ab \neq 0$).
So, after a scaling,  in all cases we may assume to have $[1,0,0,0,b]$
 ($b \neq 0$).
 By a self loop we get $[1,0,b]$. Using $k-1$ copies of $[1,0,b]$
 to connect back to  $[1,0,0,0,b]$ we get  $[1,0,0,0,b^k]$.
If  $b$ is a root of unity, we can directly
 realize $[1,0,0,0,1]$; otherwise, we can interpolate $[1,0,0,0,1]$.
 From that we can get all {\sc Equalities} of even arity.
This completes the proof.
\end{proof}

\begin{lemma}\label{lemma-pairty}
Let $f\in \mathscr{F}$ be a non-decomposable function of arity $3$
satisfying the parity condition, namely it has
 the form
$\begin{pmatrix} a &  0 & 0 & b \\
                 0 &  c & d & 0
\end{pmatrix}$ or
$\begin{pmatrix} 0 & a & b & 0 \\
                 c & 0 & 0 & d
\end{pmatrix}$.  Then {\rm Holant}$^c(\mathscr{F})$ is \#P-hard 
unless {\rm Holant}$^*(\mathscr{F})$ is tractable or
\#{\rm CSP}$_2^c(\mathscr{F})$ is tractable.
% $\mathscr{F} \subseteq \mathscr{A}$ or $\mathscr{F} \subseteq \mathscr{A}^\alpha$.
In both exceptional cases, the problem 
 {\rm Holant}$^c(\mathscr{F})$  is in P.
% (and belongs to the tractable families for \#{\rm CSP}$_2^c$).
\end{lemma}
\begin{proof}
Because any pair of antipodal points in $\{0, 1\}^3$ has opposite
parity, if there are at most two nonzeros among  $a,b,c,d$,
they would belong to a same subcube $\{0, 1\}^2$, thus $f$ would
be decomposable.
Since $f$ is non-decomposable, at least three of $a,b,c,d$ are non-zero.
By the parity condition,
some subcube $\{0, 1\}^2$ has exactly two nonzero values,
Without loss of generality suppose it is the subcube $x_3 = \epsilon$
(where $\epsilon  \in \{0,1\}$). If the two nonzero values have indices
$01\epsilon$ and $10\epsilon$, then we have
 a  generalized {\sc Disequality} by pinning $x_3 = \epsilon$.
If  the two nonzero values have indices $00\epsilon$ and $11\epsilon$,
then either $01\overline{\epsilon}$ or $10\overline{\epsilon}$ have
nonzero values. Then pinning $x_2 = 1$ or $x_1 = 1$ respectively
produces a generalized {\sc Disequality}.
%%% JYC: is there an easier way to argue this?
%%and so we can get a  generalized {\sc Disequality} by pinning.
Use the {\sc Disequality} to flip bits, we can change $f$
 to the form $\begin{pmatrix} a &  0 & 0 & b \\
                 0 &  c & d & 0
\end{pmatrix}$   with $bcd\neq 0$.
Using the triangle gadget, we can get three symmetric functions:
\[[a^3+b^3, 0, bcd+acd, 0], [a^3+c^3, 0, bcd+abd, 0], [a^3+d^3, 0, bcd+abc, 0]. \]
Note that by labeling in three different and cyclically symmetric ways
in the triangle gadget
we get these three functions (on the Boolean domain,
a cyclically symmetric ternary function is symmetric).
If at least one of the three values
 $bcd+acd, bcd+abd, bcd+abc$ is nonzero we get a
symmetric function of the form $[z,0,1,0]$, for some $z \in \mathbb{C}$.
 Otherwise, we have $b=c=d=-a$, in which case the original function
$f$ is already in this form $[z,0,1,0]$ after a nonzero scaling.

By dichotomy theorem for symmetric Holant$^c$~\cite{bib:holantc},  we know that Holant$^c([z,0,1,0])$ is \#P-hard unless $z^4= 1$.
Now we assume that $z^4= 1$. After pinning we get
the  binary $[z,0,1]$. Connection three copies we get $[z^3,0,1]$.
 Let $T=\begin{pmatrix} \sqrt{z} &  \sqrt{z}\\
                1  & -1
\end{pmatrix}$, we have $[z^3,0,1] T^{\otimes 2} = [1,0,1]$,
 and  $[z,0,1,0] = T^{\otimes 3} [1,0,0,1] $ up to a nonzero factor
$2 \sqrt{z}$.
%really  T^{\otimes 3} [1,0,0,1] = 2 \sqrt{z} [z, 0,1,0]
 So, we have the following reduction
\begin{eqnarray*}
& & {\rm Holant}( \{[1, 0, 0, 1], [z^{-1}+1,z^{-1}-1, z^{-1}+1]\} \cup T^{-1}\mathscr{F})\\
&\equiv_{\rm T} & {\rm Holant}( [1, 0, 1] \mid \{[1, 0, 0, 1], [z^{-1}+1,z^{-1}-1, z^{-1}+1] \} \cup T^{-1}\mathscr{F}) \\
&\equiv_{\rm T} & {\rm Holant}( [z^3,0,1] | \{[z,0,1, 0], [1,0,1]\} \cup \mathscr{F}) \\
& \leq_{\rm T} & {\rm Holant}^c(\mathscr{F}).
\end{eqnarray*}
Having the arity 3 {\sc Equality} $(=_3) = [1, 0, 0, 1]$ in a Holant
problem allows us to get {\sc Equality} of all arities, and thus
we can apply the \#CSP dichotomy on $T^{-1}\mathscr{F}\cup \{ [z^{-1}+1,z^{-1}-1, z^{-1}+1] \}$.
\begin{itemize}
  \item Case $z=1$. Then $T=\begin{pmatrix} 1 &  1\\
                1  & -1
\end{pmatrix}$, an orthogonal matrix
(up to a scalar $1/\sqrt{2}$)
 that belongs to the stabilizer group
 of  $\mathscr{A}$. If $T^{-1}\mathscr{F} \subseteq \mathscr{A} $, 
then  $\mathscr{F} \subseteq \mathscr{A}$. If $T^{-1}\mathscr{F} \subseteq \mathscr{P} $, then  $\mathscr{F} \subseteq T \mathscr{P}$, a tractable family for Holant$^*$. For all other $\mathscr{F}$, it is \#P-hard.
  \item Case $z=-1$. Then $T=\begin{pmatrix} i &  i\\
                1  & -1
\end{pmatrix} =
i \begin{pmatrix} 1 &  1\\
                i  & -i
\end{pmatrix} 
\begin{pmatrix} 0 &  1\\
                1  & 0
\end{pmatrix}$ is essentially the $Z$ transformation.
Note that $\begin{pmatrix} 0 &  1\\
                1  & 0
\end{pmatrix}$ belongs to the stabilizer groups of both  
$\mathscr{A}$ and $\mathscr{P}$.
If $T^{-1} \mathscr{F} \subseteq \mathscr{A}$,
then $\mathscr{F} \subseteq Z \mathscr{A} = \mathscr{A}$.
Hence \#{\rm CSP}$(\mathscr{F})$ is tractable, in particular,
\#{\rm CSP}$_2^c(\mathscr{F})$ is tractable.
  If $T^{-1}\mathscr{F} \subseteq \mathscr{P} $, then  $\mathscr{F} \subseteq
 Z \mathscr{P}$, a tractable family for Holant$^*$. In all other cases, it is \#P-hard.
  \item Case $z=\pm i$. Then $T=\begin{pmatrix} \alpha^c &  \alpha^c\\
                1  & -1
\end{pmatrix}$ for some odd $c$. In this case, $[z^{-1}+1,z^{-1}-1, z^{-1}+1]  \not \in \mathscr{P}$, so the only possible tractable case is
   $T^{-1}\mathscr{F} \subseteq \mathscr{A} $. Then it is easy to see that
   $\mathscr{F} \subseteq \mathscr{A}^\alpha$,
 a tractable family for \#CSP$_2^c$.
In all other cases, it is \#P-hard.
\end{itemize}
This completes the proof.
%
%Because $f  \not \in \mathscr{P}$ and
% $T$ is diagonal, we have $(T^{-1})^{\otimes 3}f \not \in \mathscr{P}$.
%Note that
%$\begin{pmatrix} 1 &  1\\
%                1  & -1
%\end{pmatrix}$
%and
%$\begin{pmatrix}  i & 0\\
%                0  & 1
%\end{pmatrix}$ keep $\mathscr{A}$ invariant.
%Therefore, Holant$^c(\mathscr{F})$
% is \#P-hard unless $\mathscr{F} \subseteq \mathscr{A}$
%or $\mathscr{F} \subseteq \mathscr{A}^\alpha$, in which case
% Holant$^c(\mathscr{F})$ is tractable in P.
\end{proof}

In the above three lemmas, we stated and proved them for general complex valued functions.
%So, they can be used in the proof of the dichotomy for CSP$_2$.
In the following lemmas, functions are real valued,
 which is important for our interpolation to succeed.
We first define the following notion of non-interpolatable.

\begin{definition}
Let $ab\neq 0$ be two real numbers. A binary function is called non-interpolatable if it is of the form
$\begin{pmatrix} a &  b  \\
                -b  & a
\end{pmatrix}$ or
$\begin{pmatrix} a &  b  \\
                b  & -a
\end{pmatrix}$.
\end{definition}

Non-interpolatable $2 \times 2$ matrices are just
nonzero multiples of orthogonal matrices with nonzero entries.

\begin{lemma}\label{lemma-interpolation}
Let $\begin{pmatrix} a &  b  \\
                c  & d
\end{pmatrix} \in \mathscr{F}$ be
a real valued binary function with
$a b\neq 0$ and $ad\neq bc$ (non-degenerate).
Unless it is non-interpolatable,
we have
{\rm Holant}$^c(\mathscr{F})\equiv_{\rm T}${\rm Holant}$^*(\mathscr{F})$.
\end{lemma}

\begin{proof}
By a Lemma 5.3 of \cite{CaiLX11-holant}, we can use a non-degenerate symmetric real
valued binary function $\begin{pmatrix} x &  y  \\
                y  & z
\end{pmatrix}$ and two unary $[0,1], [1,0]$ to interpolate all unary functions unless $y=0$ or $x+z=0$ (the conditions guarantee that the two
eigenvalues  have different  nonzero norm, and  $[0,1], [1,0]$
are not two eigenvectors).
From the binary function $\begin{pmatrix} a &  b  \\
                c  & d
\end{pmatrix}$, we can get two  non-degenerate symmetric real
valued binary functions
$\begin{pmatrix} a^2+b^2 &  ac+bd  \\
                ac+ bd  & c^2+d^2
\end{pmatrix}$ and
$\begin{pmatrix} a^2+c^2 &  ab+cd  \\
                ab+ cd  & b^2+d^2
\end{pmatrix}$.
%%%by connecting haed to head or tail to tail
 Since $a^2+b^2+c^2+d^2\neq 0$, we are done unless $ac+bd=0$ and  $ab+cd =0$.
Since $ab\neq 0$, this implies that $c=b, d=-a$ or
$c=-b, d=a$, which are non-interpolatable.
\end{proof}

\begin{lemma}\label{lemma-ab-a-b}
Let $f\in \mathscr{F}$ be a real valued
 function of arity $3$ such that each  of six pinnings $f^{x_i = 0}$,
 $f^{x_i = 1}$ ($1 \le i \le 3$)
produces a   non-interpolatable binary function.
 Then {\rm Holant}$^c(\mathscr{F})$ is \#P-hard
 unless $\mathscr{F}$ is  a tractable family for {\rm Holant}$^*$ or \#{\rm CSP}
(the latter condition certainly implies tractablility for \#{\rm CSP}$_2^c$).
%%% JYC: I changed this to straight #CSP.
%%% of course it also implies that it is also tractable for #CSP^2.
% should I say that? maybe I should. (as it is "formally" easier
% for people to check logical dependence.
% \#CSP$_2^c$.
%Then Holant$^c(\mathscr{F})\equiv_{\rm T}$CSP$^2(\mathscr{F})$.
\end{lemma}
\begin{proof}
By definition,  all four values of a
non-interpolatable function  are nonzero.
Thinking in terms of the six faces of the cube $\{0, 1\}^3$,
the three function values $f(100), f(010), f(001)$ are
either equal or negative of each other. If they are all equal,
then  the function is symmetric having the
 form  $[a,b,-a, -b]$ with $ab \not =0$, and we can get
the unary function $[a,b]$ by pinning.
If they are not all equal,
 then without loss of generality we can assume that
  $-f(100)=f(010)=f(001)$ and the function has
the form $\begin{pmatrix} a &  b & b & -a \\
                 -b  &  a &  a & b
\end{pmatrix}$. By pinning, we can
get the unary function $[b,a]$. Connecting this unary
back to $x_1$ of the function $f$, we get the {\sc Disequality} function:
$(0, a^2+b^2, a^2+b^2, 0)$ a nonzero multiple of $(0,1,1,0)$
since $ab \not =0$ and $a, b \in \mathbb{R}$.
 Connecting the {\sc Disequality} function back to $x_1$ of the function $f$,
we get the function
$\begin{pmatrix}   -b  &  a &  a & b \\ a &  b & b & -a \end{pmatrix}$.
%
%%% JYC: I think i need separate out
% to make $[1,0,-1]$ here, because i am using the unary
% of a specific form, to a specific form of the f.
%
This is a symmetric function $[-b,a,b,-a]$ with $ab \not =0$, and we can
also get the unary $[-b,a]$.
%form  $[c,d,-c, -d]$ after renaming.
%In the following, we can assume that we have $[a,b,-a, -b]$ with $ab \not =0$.

Connecting one unary $[a, b]$ back to
 $[a,b,-a, -b]$, or $[-b,a]$ back to $[-b,a,b,-a]$,
  we get the function $[1,0,-1]$, again
because $a,b$ are nonzero real
numbers we have
$a^2 + b^2 \not =0$.
Let $Z=\begin{pmatrix} 1 &  1  \\
                i  & -i
\end{pmatrix}$,  we have $[a,b,-a, -b]=Z^{\otimes 3} [ c, 0, 0, d]$
for some nonzero $c, d  \in \mathbb{C}$,
  and $[1,0,-1] Z^{\otimes 2} = 2 [1,0,1]$.
So we have the following reduction:
\begin{eqnarray*}
& & {\rm Holant}(Z^{-1}\mathscr{F}
\cup \{[c, 0, 0, d], Z^{-1} \Delta_0, Z^{-1} \Delta_1\})\\
&\equiv_{\rm T} &
 {\rm Holant}( [1, 0, 1] \mid Z^{-1}\mathscr{F}
\cup \{[c, 0, 0, d], Z^{-1} \Delta_0, Z^{-1} \Delta_1\})\\
&\equiv_{\rm T} &
{\rm Holant}( [1, 0, 1](Z^{-1})^{\otimes 2} \mid \mathscr{F}
\cup \{Z^{\otimes 3} [ c, 0, 0, d], \Delta_0, \Delta_1\})\\
&\equiv_{\rm T} &
{\rm Holant}( [1,0,-1] \mid  \mathscr{F}
\cup \{[a,b,-a, -b], \Delta_0, \Delta_1\})\\
& \leq_{\rm T} & {\rm Holant}^c(\mathscr{F}).
\end{eqnarray*}
The reduction for $[-b,a,b,-a]$ is the same.
%%% really? equivalent?  why is Z^{-1} F \subset P ==> F \subset P?
% also what is Z^{-1} F \subset alpha-A ==> ???
% [1 0 \\ 0 alpha]^{-1} Z [1 0 \\ 0 alpha] = [1 alpha\\ alpha -i]  then what?
% don't think [1 alpha\\ alpha -i] \in Stab(A)
%
% xmj: Is Lemma 5.2's proof strong enough, such that we can conclude for Holant beside Holant^c.
By Lemma~\ref{lemma-odd-equality}, we know that Holant$(\mathscr{F})$
is \#P-hard, unless $\mathscr{H} \subseteq \mathscr{A}$, $\mathscr{H} \subseteq \mathscr{A}^\alpha$ or $\mathscr{H} \subseteq \mathscr{P}$,
 where $\mathscr{H}=Z^{-1}\mathscr{F} \cup \{[c, 0, 0, d] , Z^{-1} \Delta_0, Z^{-1} \Delta_1\}$.
Notice that $Z^{-1} \Delta_0
= \frac{1}{2} \begin{pmatrix} 1 &  -i  \\
                1  & i
\end{pmatrix}
\begin{pmatrix} 1 \\ 0 \end{pmatrix}
= \frac{1}{2} \begin{pmatrix} 1 \\ 1 \end{pmatrix}$, and
the unary function $[1,1] \not \in \mathscr{A}^\alpha$.
We conclude that
 Holant$(\mathscr{F})$
is \#P-hard, unless
$\mathscr{H} \subseteq \mathscr{A}$ or $\mathscr{H} \subseteq \mathscr{P}$.
Since $Z$ is in the Stablizer group of $\mathscr{A}$,
we have $Z \mathscr{A} = \mathscr{A}$, and so the first
condition translates to $\mathscr{F}  \subseteq \mathscr{A}$,
which is a tractable condition for \#CSP$(\mathscr{F})$.
The second condition translates to
$\mathscr{F}  \subseteq Z \mathscr{P}$.
This implies that $\mathscr{F}$ is $\mathscr{P}$-transformable,
namely in ${\rm Holant}(=_2 \mid \mathscr{F})
\equiv_{\rm T} {\rm Holant}((=_2) Z^{\otimes 2} \mid Z^{-1} \mathscr{F})$,
 both $(=_2)Z^{\otimes 2} = (\neq_2) \in \mathscr{P}$
and $Z^{-1} \mathscr{F} \subseteq \mathscr{P}$.
This is one of the tractable families for 
Holant$^*$ problems in Theorem \ref{thm old holant-star-dichotomy}. 
 \end{proof}

\begin{lemma}\label{lemma-pairty1}
Let $f\in \mathscr{F}$ be a  non-decomposable real valued
 function of arity $3$ with the form $\begin{pmatrix} a &  0 & 0 & c \\
                 b  &  0 & 0 & d
\end{pmatrix}$ or
$\begin{pmatrix}   0 & a & c & 0  \\
0 & b & d & 0
\end{pmatrix}$ with $ab\neq 0$.  Then {\rm Holant}$^c(\mathscr{F})$ is \#P-hard unless $\mathscr{F}$ is  a tractable family for {\rm Holant}$^*$ or \#{\rm CSP}$_2^c$.
\end{lemma}
\begin{proof}
By being non-decomposable, $c$ and $d$ cannot be both 0.
For the form $\begin{pmatrix}   0 & a & c & 0  \\
0 & b & d & 0
\end{pmatrix}$,  we can get a generalized {\sc Disequality} by pinning and
then use this generalized {\sc Disequality} to change $f$ to the form
% flip x_3
$\begin{pmatrix} a &  0 & 0 & c \\
                 b  &  0 & 0 & d
\end{pmatrix}$. So, we only need to deal with this case.

We can get the unary function $[a,b]$ by pinning. Connecting
 $[a,b]$ to $x_2$ of $f$ we get
$\begin{pmatrix} a^2 &  b c \\
                a b  & b d
\end{pmatrix}$, which also gives us
$\begin{pmatrix} a^2 &  a b \\
 b c & b d
\end{pmatrix}$ by switching the two variables.
%%% said this, because lm {lemma-interpolation} requires top row nonzero.
  Applying Lemma \ref{lemma-interpolation},
we are done unless the binary function is non-interpolatable.
If so, we know that $cd\neq 0$. Similar, we can realized unary function $[c,d]$ by pinning. Connection $[c,d]$ to $x_2$ of $f$ we get
$\begin{pmatrix} a c &  c d \\
                b c  &  d^2
\end{pmatrix}$.  We are done unless this binary function is non-interpolatable.

Now, we assume that both binary functions are non-interpolatable.
From the first one, we get $a^2=\pm b d$ and $ab=\mp bc$.
 If $a^2=-b d$ and $ab=bc$, then $a=c$.
By being real, $ac = a^2 >0$, so from the second one,
$ac = d^2$ and $bc = - cd$. So $a = \pm d$ and $b = -d$.
This gives us the function
$\begin{pmatrix} 1 &  0 & 0 & 1 \\
                 1  &  0 & 0 & -1
\end{pmatrix}$ or
$\begin{pmatrix} 1 &  0 & 0 & 1 \\
                 -1  &  0 & 0 & 1
\end{pmatrix}$ after scaling.
If $a^2= b d$ and $ab=-bc$, then $a=-c$.
By being real, $ac = -a^2 < 0$,  so from the second one,
$ac = -d^2$ and $bc = cd$.  So $a = \pm d$ and $b = d$.
This gives us the function
$\begin{pmatrix} 1 &  0 & 0 & -1 \\
                 1  &  0 & 0 & 1
\end{pmatrix}$ or
$\begin{pmatrix} 1 &  0 & 0 & -1 \\
                 -1  &  0 & 0 & -1
\end{pmatrix}$ after scaling.
If we concentrate on the $2 \times 2$ nonzero submatrix,
the four matrices are obtained from
$\begin{pmatrix} 1 &   1 \\
                 1  &  -1
\end{pmatrix}$ by pre- or post- multiplying
by the orthogonal
$\begin{pmatrix} 1 &   0 \\
                 0  &  -1
\end{pmatrix}$, and thus all are orthogonal up to a scalar $1/\sqrt{2}$.
Therefore, by computing $M^{\tt T} M$ for  the
$4 \times 2$ matrix $M$, we get the signature of the {\sc Equality}
function of arity $4$, up to a scalar $2$.
This is realized by connecting the $x_1$ variable
of two copies of the function with matrix $M$.
So we are done by Lemma~\ref{lemma-even-equality}.
\end{proof}

\begin{lemma}\label{lemma-mix}
Let $f\in \mathscr{F}$ be a  non-decomposable function of arity $3$.
Suppose all six binary functions $f^{x_i = 0}$ and $f^{x_i = 1}$
($1 \le i \le 3$)
are either non-interpolatable or  degenerate, and
furthermore both types occur.  Then {\rm Holant}$^c(\mathscr{F})$ is \#P-hard unless
either {\rm Holant}$^*(\mathscr{F})$ or {\rm CSP}$^c_2(\mathscr{F})$  is tractable
in polynomial time.
% a tractable family for Holant$^*$ or \#CSP$_2^c$.
\end{lemma}

\begin{proof}
Recall that the signature matrix of a non-interpolatable binary function
is a nonzero multiple of an orthogonal matrix with 4 nonzero entries.
Suppose  $f^{x_i = \epsilon}$ is non-interpolatable ($\epsilon = 0, 1$).
On any face $x_j = 0, 1$ for $j \not = i$ there are at least two nonzero
entries from $f^{x_i = \epsilon}$, and hence if $f^{x_j = 0}$ or $f^{x_j = 1}$
has  a zero entry it must be degenerate and have two zero entries.
Then $f^{x_i = 1-\epsilon}$ must be identically zero, a contradiction to
$f$ being non-decomposable.
 Hence $f$ has no zero entries among all eight values.
 So we can assume that the function is of the form
$A = \begin{pmatrix} a &  b & \lambda a  & \lambda b \\
                     \pm b &  \mp a & x & y
\end{pmatrix}$ where $\lambda\neq 0$ , up to some bit flips.
% or
%$\begin{pmatrix} a &  b & b &  -a \\
%                 \lambda a  & \lambda b & x & y
%\end{pmatrix}$ up to some bit flips.
% All the cases are similar,
%we prove for the case $\begin{pmatrix} a &  b & -b &  a \\
%                 \lambda a  & \lambda b & x & y
%\end{pmatrix}$.
Since $A$ is a real matrix of rank 2,
the real symmetric matrix $A A^{\tt T}$ has rank 2
and positive trace. If $(a, b, \lambda a, \lambda b)$ is not
orthogonal to $(\pm b, \mp a, x, y)$,
which is equivalent to $\lambda(a, b)$ is not
orthogonal to $(x, y)$,
then $A A^{\tt T}$ has nonzero off diagonal,
and we can interpolate all unary functions using
$A A^{\tt T}$ and a unary $[1, 0]$.
% since $[1, 0]$ is not its eigenvector.
% A A^{\tt T} is full rank 2
% eigenvaluase ratio \not = \pm 1. ...
% ratio = [-trace + \sqrt{disc>0}]/[-trace - \sqrt{disc>0}],
% is real and has norm \not =1
%
So we may assume $\lambda(a, b)$ is orthogonal to $(x, y)$.
Since $\lambda\neq 0$, we have
$(a, b)$ is orthogonal to $(x, y)$.
Thus $f$ has the form
$A = \begin{pmatrix} a &  b & \lambda a  & \lambda b \\
             \sigma b &  -\sigma a & \mu b  & -\mu a
\end{pmatrix}$,
($\sigma = \pm 1$).
Clearly $\mu \not = \sigma \lambda$, since $f$ is non-decomposable.

By pinning we can get
$\begin{pmatrix} a & \lambda a\\
                \sigma  b & \mu  b
\end{pmatrix}$,
and
$\begin{pmatrix} b & \lambda b\\
                -\sigma  a & -\mu  a
\end{pmatrix}$.
By assumption both are either non-interpolatable or degenerate.
By  $\mu \not = \sigma \lambda$, both are non-degenerate.
So both are non-interpolatable.
Hence the columns are orthogonal,
\begin{equation}\label{lambda-mu-orthogonal}
\lambda a^2 +  \sigma  \mu  b^2 =0 ~~~~\mbox{and}~~~~
\lambda b^2 +  \sigma  \mu  a^2 =0.
\end{equation}

Now we consider the gadget with signature
\[A^{\tt T} A
= \begin{pmatrix} a & \sigma b \\
                  b & - \sigma a \\
                  \lambda a  & \mu b \\
                  \lambda b & -\mu a
\end{pmatrix}
\begin{pmatrix} a &  b & \lambda a  & \lambda b \\
             \sigma b &  -\sigma a & \mu b  & -\mu a
\end{pmatrix}.\]
We can pin to get its first two rows
$\begin{pmatrix} a^2 + b^2 & 0 & 0 & (\lambda - \sigma \mu) ab \\
                 0 &  a^2 + b^2  & (\lambda - \sigma \mu) ab &  0
\end{pmatrix}$.
Here we used (\ref{lambda-mu-orthogonal}).
Note that $(\lambda - \sigma \mu) ab \not =0$.
Hence this ternary function is non-decomposable.
By Lemma~\ref{lemma-pairty} we are done.

\end{proof}

%%%%%%% earlier note from Pinyan probaby moot
%
%TODO: there are many similar cases in the above proof and I haven't found a way to unify them.

Now we are ready to prove the reduction from Holant$^c$ problems
to  Holant$^*$ or \#CSP$_2^c$.

\begin{theorem}\label{thm:reduction-of-holant-c-to-holstar-csp2}
Let $\mathscr{F}$ be a set of real valued functions. Then {\rm Holant}$^c(\mathscr{F})$ is \#P-hard unless $\mathscr{F}$ is  a tractable family for {\rm Holant}$^*$ or \#{\rm CSP}$_2^c$, for both we have explicit dichotomy theorems.
\end{theorem}

\begin{proof}
By Lemma~\ref{lemma-decomposable}, we can assume that functions
 in $\mathscr{F}$ are non-decomposable.  If every function in $\mathscr{F}$
has arity at most two, then  Holant$^c(\mathscr{F})$ is tractable.
 Now we assume that $\mathscr{F}$ contains a function $f$ of arity
 at least $3$. Since $f$ is non-decomposable,
there are at least two nonzero function values. Let
\[D_0 = \min \{ d(x, y) \mid x \not = y, f(x) \not =0, f(y)  \not =0\},\]
the minimum Hamming distance between two inputs with nonzero values.

\begin{itemize}
\item
If $D_0 \geq 3$ and $D_0$ is odd, we can get a
 generalized {\sc Equality} of arity $3$ by pinning and then by self loops,
and so we are done by Lemma~\ref{lemma-odd-equality}.
\item
If $D_0\geq 4$ and $D_0$ is even, we can get a generalized
 {\sc Equality} of arity $4$, and so we are
done by Lemma \ref{lemma-even-equality}.
\item
If $D_0=2$, without loss of generality,
 we can pin $x_3 x_4\cdots x_n$ such that the remaining binary
 function $g(x_1, x_2)$ is of the form $\begin{pmatrix} a &  0  \\
                0  & b
\end{pmatrix}$ or $\begin{pmatrix} 0 &  a  \\
                b  & 0
\end{pmatrix}$,  where $ab\neq 0$.
If it is the second form, it is
a generalized {\sc Disequality} and can be used to flip the input.
 So we can assume $g(x_1, x_2)$ has the first form.
If for all other values of $x_3 x_4\cdots x_n$,
 the remaining binary function is a scaling of the above one,
 then the function $f$ is decomposable, a contradiction.
 Let $A$ be the set of bit patterns for
$x_3 x_4\cdots x_n$ for which the remaining binary function
 is a nonzero scaling of the above function, and let
 $B$ be the set of  bit patterns  for which the  remaining
binary function is not a  scaling of the above function.
By definition of this binary function, $A \not = \emptyset$.
By being non-decomposable, $B \not = \emptyset$.
Clearly $A \cap B = \emptyset$.

 Let
\[D_1
 = \min \{ d(x, y) \mid x \in A,  y \in B\},\]
be the minimum Hamming distance between the two sets.
\begin{enumerate}
\item
 If  $D_1=1$, by pinning we get a non-decomposable ternary
function and it satisfies the parity condition. This is clear
by looking at the cube $\{0, 1\}^3$, using
the fact that  $D_0 =2$, $ab \not =0$ and the definition of $B$.
%% the immediate neighbors of a and b in the other square must =0. as ab not=0
%  non-decomposable: because ab not=0, the only possible decomposision
% is in the 3rd dim, with [a,0\\0,b] or [0,a\\b,0] being a factor.
% but by definition of B this is not so.
  So we are done by Lemma~\ref{lemma-pairty}.
\item
If $D_1=2$,
without loss of generality we may assume
a pinning bit pattern for $x_5\ldots x_n$ such that
further pinning $x_3x_4=b_3b_4$ gives us the function $g(x_1, x_2)$,
and further pinning $x_3x_4=\overline{b_3} \overline{b_4}$
 gives us another binary  function  which is not a scaling of $g$.
Because $D_1=2$, the  binary  function  obtained by
 further pinning $x_3x_4=b_3\overline{b_4}$
or $x_3x_4= \overline{b_3} b_4$ must be identially 0.
It follows that we have a  function of arity $4$ after pinning,
of the form
$\begin{pmatrix} a &  0 & 0 & b \\
                 0  &  0 & 0 & 0 \\
                 0  &  0 & 0 & 0 \\
                 c &  0 & 0 & d
\end{pmatrix}$
or
$\begin{pmatrix} a &  0 & 0 & b \\
                 0  &  0 & 0 & 0 \\
                 0  &  0 & 0 & 0 \\
                 0 &  c & d & 0
\end{pmatrix}$
where the row index of $x_3x_4$ is up to a bit flip,
with $c$ and $d$ not both 0.
It is easy to verify that this function is non-decomposable by $ab \not =0$
and the definition of $B$.
% ab not 0 ==> can't decompose in the dim 1 or 2.
% then the middle two 0 rows says if it is decomposable, both c=d=0.
% but copy from B is not identically 0
%
For the latter case $\begin{pmatrix} a &  0 & 0 & b \\
                 0  &  0 & 0 & 0 \\
                 0  &  0 & 0 & 0 \\
                 0 &  c & d & 0
\end{pmatrix}$, we can pin $x_1$ or $x_2$
% depending c or d not =0
to get a generalized {\sc Equality} of arity 3
 since at least one of $c$ and $d$ is nonzero.
 Then we are done by Lemma~\ref{lemma-odd-equality}.
For the former case, by definition of $B$,
$\det
\begin{pmatrix} a & b\\
c & d
\end{pmatrix} \not =0$.  We take two copies of the function
and connect the respective $x_3$ and $x_4$ together.
This produces a symmetry function of the form
$\begin{pmatrix} x &  0 & 0 & y \\
                 0  &  0 & 0 & 0 \\
                 0  &  0 & 0 & 0 \\
                 y &  0 & 0 & z
\end{pmatrix}$, with $x, z > 0$.
% x = a^2 + c^2
% z = b^2 + d^2
 We can use it to realize or interpolate an {\sc Equality} of arity $4$. So we are done by Lemma \ref{lemma-even-equality}.
\item
If $D_1\geq 3$, we can get a generalized {\sc Equality} of arity at least $3$, and we are done by Lemma \ref{lemma-odd-equality} or  Lemma \ref{lemma-even-equality}.

%%% when pinned, all internal nodes are 0 , as it produces a function
% can't be a non-zero scaling of g, (ow it is shorter dist from A to this in B)
% can't be non-scaling of g,  (ow it is shorter dist from g to B
% so must be identically 0
\end{enumerate}
\item
If $D_0=1$,  without loss of generality,
 we can assume that there is a pinning for $x_2x_3 x_4\cdots x_n$
such that the remaining unary function is $[a,b]$ with $ab\neq 0$.
  If for all  other values of $x_2 x_3 x_4\cdots x_n$,
the remaining unary function is a scaling of the above one,
then the function $f$ is decomposable, a contradiction.
 Let $A$ be the set of bit patterns for  $x_2x_3 x_4\cdots x_n$
 for which the remaining unary function is a nonzero scaling of
 $[a,b]$, and $B$ be the set of patterns for which the
 remaining unary function is not a  scaling of $[a,b]$.
By the above argument, both sets $A$ and $B$ are non-empty.

 Let $D_2$ be the minimum Hamming distance between the two sets
$A$ and $B$. Again,
\begin{enumerate}
\item If $D_2\geq 3$, we can get a generalized
 {\sc Equality} of arity at least $3$, and we are done
by Lemma~\ref{lemma-odd-equality} or  Lemma~\ref{lemma-even-equality}.
\item
 If $D_2=2$, we have a non-decomposable ternary function
taking the form in Lemma~\ref{lemma-pairty1},
 and we are done by that lemma.
\item
If $D_2=1$,  without loss of generality,
 we can assume that there is a pinning for $x_3 x_4\cdots x_n$
 such that the remaining binary function is of
form $\begin{pmatrix} a &  b  \\
                c  & d
\end{pmatrix}$ where $ab\neq 0$ and $ad\neq bc$.
 We are done by Lemma~\ref{lemma-interpolation}
unless it is non-interpolatable.
Now we assume that it is non-interpolatable.
In particular, $abcd \not =0$.
 If for all other values of $x_3 x_4\cdots x_n$, the remaining
 binary function is a scaling of the above one, then the function
$f$ is decomposable, a contradiction.
Let $A$ be the set of bit patterns for  $x_3 x_4\cdots x_n$
for which the remaining binary function is a nonzero scaling
 of the above function, and $B$ be the set of  bit patterns for which
 the remaining binary function is not a  scaling of the above function.
By the above argument,  both sets $A$ and $B$ are non-empty.

 Let $D_3$ be the minimum Hamming distance between the two sets
$A$ and $B$. Again,
\begin{description}
\item{Case $D_3\geq 3$:}
%if $D_3\geq 3$, we can get a generalized
We can get a generalized
{\sc Equality} of arity at least $3$, and we are done by
 Lemma~\ref{lemma-odd-equality} or  Lemma~\ref{lemma-even-equality}.
 %The only two remaining case is $D_3=2$ and $D_3=1$.
\item{Case $D_3=2$:}
We have a function with arity $4$: for $x_3=a_3, x_4=a_4$,
we have a non-interpolatable binary function;
for $x_3=1-a_3, x_4=1-a_4$, we have a binary function
which is not a scaling of the above one;
for the other two values of $x_3, x_4$, the function is entirely zero.
Up to a flip on the row index bits $x_3$ and $x_4$, we have
the function of arity $4$
of the form
$\begin{pmatrix} a &  b & c & d \\
                 0  &  0 & 0 & 0 \\
                 0  &  0 & 0 & 0 \\
                 a' &  b' & c' & d'
\end{pmatrix}$,
where $abcd \not =0$, and
$(a', b', c', d')$ is linearly independent of $(a,b,c,d)$.
If $a'/a \not = b'/b$, or $c'/c \not = d'/d$, or
$a'/a \not = c'/c$, then we have at least one pinning of $x_1$ or $x_2$
such that the resulting ternary function is non-decomposable.
By linear independence, one of these must hold,
and  the resulting ternary function is non-decomposable.
That function is of a form of  Lemma~\ref{lemma-pairty1}
and we are done by that lemma.
\item{Case $D_3=1$:}
%Finally if  $D_3=1$
We get a non-decomposable ternary function.
%% the first dim 2 (a b \\ c  d) has det non-0. so can't decompose there
%% by definition f mA and B, can't decompose the 3-dim either.
We know that at least one of the six faces
$\begin{pmatrix} a &  b  \\
                 c  &  d
\end{pmatrix}$ is non-interpolatable. This implies that
the four adjacent faces have at least two nonzero entries
(from $\{a,b,c,d\}$) that are of Hamming distance 1.
If any one of these 4 faces is non-degenerate and not  non-interpolatable,
then we are done by Lemma~\ref{lemma-interpolation}.
If the opposite face of
$\begin{pmatrix} a &  b  \\
                 c  &  d
\end{pmatrix}$ is  non-degenerate and not  non-interpolatable,
we are also done by Lemma~\ref{lemma-interpolation},
unless it has no two adjacent nonzero entries. But if so,
being non-degenerate, it must have exactly two nonzeros at bit positions
of same parity, and two other zero entries at bit positions
of the oppostite parity.
Then in particular any of the  four adjacent faces
of $\begin{pmatrix} a &  b  \\
                 c  &  d
\end{pmatrix}$ has exactly one 0 entry and thus
both  non-degenerate and not  non-interpolatable.
   Thus we conclude that
if any one of six faces is
non-degenerate and not  non-interpolatable,
 then we are done by Lemma~\ref{lemma-interpolation}.
Now, suppose each of its six faces is
 either degenerate or non-interpolatable.
We already know that at least one of them is non-interpolatable.
 If all of them are non-interpolatable, we are done by Lemma~\ref{lemma-ab-a-b}. Otherwise, we are done by Lemma \ref{lemma-mix}.
\end{description}
\end{enumerate}
\end{itemize}
This completes the proof of Theorem~\ref{thm:reduction-of-holant-c-to-holstar-csp2}.
\end{proof}

\section*{Acknowledgments}
We sincerely thank Zhiguo Fu for his very insightful comments, in particular
his gave the key insight to a simplified proof of
Lemma~\ref{lemma-mix}.

\bibliography{bib}

\end{document}